\documentclass[12pt]{article}

\usepackage[margin=1in]{geometry}
\usepackage{setspace}
\usepackage{tikz}
\usepackage{graphicx}
\usepackage{float}
\usepackage{booktabs}
\usepackage{amsmath,amssymb,amsthm,mathtools}
\usepackage{enumitem}
\usepackage{booktabs}
\usepackage{microtype}
\usepackage[normalem]{ulem} 
\usepackage{color}
\usepackage[longnamesfirst]{natbib}
\usepackage{etoolbox}
\usepackage[colorlinks=true,linkcolor=blue,citecolor=blue,urlcolor=blue]{hyperref}
\usepackage[bottom]{footmisc}
\newtheorem{lemma}{Lemma}
\newtheorem{proposition}{Proposition}
\newtheorem{theorem}{Theorem}
\newtheorem{corollary}{Corollary}
\theoremstyle{definition}

\newrobustcmd*{\citefirstlastauthor}{\AtNextCite{\DeclareNameAlias{labelname}{given-family}}\citeauthor}

\DeclareMathOperator*{\argmin}{arg\,min}

\newcommand{\Dist}{\mathrm{Dist}}
\newcommand{\MRS}{\mathrm{MRS}}
\newcommand{\MRT}{\mathrm{MRT}}
\newcommand{\Agg}{\mathrm{State}}

\newcommand{\rfirst}{R^{\text{FPTP}}}
\title{Representation in District-Based Elections\thanks{We are grateful to Nageeb Ali, Peter Buisseret, Steve Callander, Andres Carvajal, Hulya Eraslan, Ezra Friedman, SoRelle Wyckoff Gaynor, Matthew Jackson, Jake Jares, Jonathan Katz, Micheal Kistner, Gary King, Murat Koksalan, Justin Levitt, Alessandro Lizzeri, Scott Page, Francesco Parisi, Maggie Penn, JJ Prescott, Doron Ravid, Alessandro Riboni, Ken Shotts, Nicholas Stephanopoulos, Ian Turner, Richard Van Weelden, and seminar participants at Rochester, Texas A\&M, UC Davis, SITE, Penn State, University of Pittsburgh and Carnegie Mellon for helpful discussions and comments.}} 
\author{Yunus C. Aybas\thanks{Texas A\&M University, Department of Economics.}
  \and Oguzhan Celebi\thanks{University of Michigan, Ross School of Business.}
\and Surabhi Dutt\footnotemark[2]}
\date{October 26, 2026}

\begin{document}

\pagenumbering{gobble}
\maketitle
\thispagestyle{empty}

\begin{abstract}

District-based elections must reconcile two objectives: representing each district’s electorate and matching the statewide delegation to parties' vote shares. We study seat allocation rules when these objectives conflict. The optimal rules, characterized by a cutoff, form a family ranging from first-past-the-post to proportional representation as the weight on statewide representation increases. Along this path, gaining further seats through gerrymandering becomes increasingly difficult. The characterization yields a measure of the robustness of first-past-the-post to statewide representation concerns, which we compare with established electoral diagnostics.  In U.S. House elections from 2002 to 2024, we find that maps drawn by single-party legislatures perform worse than those produced by courts or non-partisan commissions. The efficiency gap, a common measure of partisan gerrymandering, is more strongly associated with litigated maps than the other measures, yet magnifies the extent to which Republican-legislature maps appear worse than Democratic-legislature maps.

\end{abstract}
\thispagestyle{empty}

\newpage 
\pagenumbering{arabic}

\section{Introduction}

In single-member district elections, each district elects one representative, and these representatives together form a delegation for a larger jurisdiction. For example, in U.S. congressional and state house elections, the representatives elected in a state's districts form its statewide delegation. This structure gives rise to two conflicting representational goals: representing each district's electorate and ensuring that the delegation's partisan composition reflects statewide voter support. No seat-allocation rule can, in general, both elect each district's plurality winner and allocate party seats in proportion to statewide vote shares. Choosing a rule therefore means deciding how to balance local and statewide representation.

Two familiar rules illustrate the extremes of this choice. First-past-the-post assigns every seat to the district's plurality winner, but it can produce a delegation whose party seat shares differ substantially from statewide vote shares. Proportional representation sets party seat totals in proportion to statewide vote shares, but attaining those totals may require awarding some district seats to parties that lost locally.\footnote{Throughout, proportional representation refers to a delegation-level target within a fixed district-based system. Seat totals are chosen to be as close as possible to proportional entitlements and district seats are then assigned to implement those totals. Related institutional alternatives include party-list PR (Israel, the Netherlands), mixed-member proportional systems (Germany, New Zealand), parallel mixed systems (Japan), and list-PR systems with compensatory or leveling seats. (Sweden, Norway, Denmark).}

Even when each district elects its plurality winner, the delegation's overall partisan composition remains a central concern in debates over representation. For example, Republicans received 53.3\% of North Carolina's statewide House vote in 2016 but won 10 of 13 seats---76.9\% of the delegation---under a map challenged as a partisan gerrymander in \emph{Rucho v. Common Cause} (2019).\footnote{Republican redistricting co-chair David Lewis had explicitly described the map's partisan purpose: ``I think electing Republicans is better than electing Democrats. So I drew this map to help foster what I think is better for the country.'' See \emph{Rucho v. Common Cause} 588 U.S. 684, 691 (2019).} The relationship between statewide voter support and the delegation's partisan composition also figures explicitly in judicial decisions and legislative proposals concerning congressional elections. State courts have explicitly emphasized that congressional maps should reflect statewide partisan preferences (e.g., \emph{Carter v. Chapman} (2022)), and legislative proposals such as the 2021 Freedom to Vote Act have likewise called for evaluating plans using measures based on the relationship between statewide vote shares and seat shares.\footnote{In \emph{Carter v. Chapman} (2022), a majority described consideration of partisan fairness among plans meeting traditional districting criteria as ``necessary to ensure that a congressional plan is reflective of and responsive to the partisan preferences of the Commonwealth's voters'' (\emph{Carter v. Chapman}, 270 A.3d 444 (Pa. 2022)). In the North Carolina litigation, the district court relied in part on quantitative measures of partisan asymmetry in finding the 2016 congressional map unconstitutional, see \emph{Common Cause v. Rucho}, (M.D.N.C. 2018). The Freedom to Vote Act proposed using ``standard quantitative measures of partisan fairness that relate a party's share of the statewide vote to that party's share of seats.'' Appendix~\ref{app:statewide-targets} discusses these sources and related cases, the representational benchmarks they consider, and their legal qualifications.}

A large literature addresses these concerns by examining how district boundaries shape representation and developing criteria for evaluating maps \citep{mcghee2020partisan}. The Supreme Court has, however, held partisan-gerrymandering claims nonjusticiable in federal court (\emph{Rucho v. Common Cause} (2019)).\footnote{The majority concluded that the Constitution supplies no judicially manageable standard for determining when partisan advantage becomes excessive. Racial-gerrymandering and Voting Rights Act claims are governed by distinct standards. In \emph{Louisiana v. Callais} 608 U.S. 85 (2026), the Court upheld the invalidation of Louisiana's 2024 congressional map and tightened the requirements for establishing vote dilution under Section~2, including requirements to distinguish racial from partisan explanations. See Appendix~\ref{app:statewide-targets} for further discussion.} Given these limits on federal judicial remedies, recent reform discussions have considered changing electoral rules to limit the gains from gerrymandering \citep{berman2026reform,aaas2025expanding}. We develop this approach by studying how electoral rules can balance local and statewide representation and how that balance shapes the gains from gerrymandering.

We study a two-party model of single-member district elections, taking district boundaries and the resulting distribution of votes as given. An electoral rule specifies one elected representative per district for every possible vote profile, allowing each district's outcome to depend on votes cast in other districts.  We characterize rules that optimally balance local and statewide representation according to the relative importance assigned to each. We then examine how these representational priorities shape the gains from gerrymandering and the evaluation of electoral fairness.

We evaluate representation through two components: a district-level component measuring how well each district's elected party reflects its local electorate, and a delegation-level component measuring how closely the delegation's partisan composition matches a target based on statewide vote shares. We measure departures from these objectives as district-level and delegation-level \textit{misrepresentation}. A relative weight and an aggregation function combine these losses into a single objective.

For a fixed party seat total, statewide misrepresentation is constant, so the district component determines the optimal assignment. A party receiving \(S\) seats is therefore assigned its \(S\) strongest districts---the \textit{top-\(S\) allocation}. Minimizing total misrepresentation thus reduces to choosing a seat total. The district-only seat total minimizes district-level misrepresentation alone; the statewide-only seat total minimizes statewide misrepresentation alone. We show that an optimal seat total lies between these two benchmarks. When these totals differ, one party's district-only seat total is lower than its statewide-only seat total. We label this party \(A\) and call it \emph{underrepresented}.

Our first result identifies when an additional seat for party \(A\) reduces total misrepresentation. Ordering districts by decreasing support for party \(A\), moving from \(S\) to \(S+1\) seats changes only the assignment of the marginal district, \(S+1\). Below the statewide-only seat total, this change raises district-level misrepresentation and lowers statewide misrepresentation. Theorem~\ref{thm:general-cutoff} expresses this tradeoff as a cutoff: the additional seat improves the allocation when the weighted statewide gain justifies the local cost.

To identify the global optimum, we assume diminishing marginal rates of substitution: additional district-level loss becomes harder to justify as it grows, while statewide improvements become less valuable as the delegation approaches its target. This condition orders the marginal comparisons, so the cutoffs identify the globally optimal party $A$ seat total (Theorem~\ref{thm:optimal-cutoff-rule}). 
As the statewide weight increases, party \(A\) receives seats in districts where its support is weaker, moving its optimal seat total monotonically toward the statewide-only seat total (Proposition~\ref{prop:monotone-support-sets}).

We then specialize to a linear model that measures both components in units of one district's population. District-level misrepresentation counts voters represented locally by their nonpreferred party. Statewide misrepresentation measures the absolute gap between the population represented by a party's seats and its statewide support. The objective is a weighted sum of these two losses. Two results motivate our focus on the linear model. First, varying the relative weight traces the full Pareto frontier of district and statewide misrepresentation, despite the discrete, nonconvex set of feasible allocations. Second, the objective admits a voter-level foundation: when voters value their preferred party's representation both locally and in the statewide delegation, maximizing aggregate voter payoffs is equivalent to minimizing the linear objective.

The linear model yields a family of rules, indexed by a statewide weight ranging from zero to one, that are optimal for every vote profile (Theorem~\ref{thm:linear-cutoff-path}). At zero weight, only district-level representation matters, and the rule is first-past-the-post: a cutoff of \(1/2\) assigns each district to its local majority winner. As the statewide weight increases, the cutoff falls for the party underrepresented under first-past-the-post, giving it additional seats until it reaches its statewide-only seat total. At weight one, only statewide representation matters, and the rule implements proportional representation by selecting the integer seat total closest to proportionality and assigning seats to the parties' strongest districts.

We establish three properties of the linear optimal family. First, two robustness properties fully characterize first-past-the-post and proportional representation within this family. First-past-the-post is the unique rule that preserves a party's district wins when its support rises district by district. Proportional representation is the unique gerrymandering-proof rule: neither party can gain seats through redistricting that preserves statewide support. Interior rules satisfy neither property globally, reflecting the competing demands of local responsiveness and statewide proportionality. 

Second, we examine the scope for gerrymandering within the optimal family by studying the seat gains attainable through partisan redistrictings that preserve statewide support. For a party already overrepresented relative to proportionality, increasing the statewide weight shrinks the set of redistrictings that deliver at least a given seat gain over its initial seat total, narrowing its opportunities for further gains through gerrymandering. This result complements approaches that focus on district boundaries under first-past-the-post: changing the rule applied to a given map can itself constrain the partisan gains that the map delivers. 

Third, using majorization to measure the concentration of district vote shares at fixed statewide support, we show that greater concentration reduces total misrepresentation if and only if the rule is proportional or calibrated to the representational preferences used to evaluate outcomes.

We also formally map familiar electoral diagnostics into our framework, making explicit the representational objectives embedded in their statewide benchmarks and measures of distance from those benchmarks.
The \textit{efficiency gap} measures departures from a seat-share benchmark that includes a bonus for the statewide majority, while disproportionality indices measure distances between seat and vote shares. \textit{Partisan bias} evaluates the seat--vote gap at a hypothetical statewide tie, constructed by uniformly shifting district vote shares while holding the map fixed.

The framework also yields a new measure for evaluating first-past-the-post outcomes. For a given election, the first-past-the-post allocation is optimal at zero statewide weight; we define the \textit{first-switch weight} as the largest statewide weight for which that allocation remains optimal. A higher value indicates greater robustness to statewide representation concerns. We compare this measure with the efficiency gap and the \textit{proportional seat gap}, the difference between a party's seat share and its statewide vote share.

We apply these measures to U.S.\ House elections from 2002 to 2024, combining district-level returns with information on who drew each congressional map and whether it was challenged or struck down in court. Litigated or struck-down maps are more concentrated in the upper tail of the efficiency-gap distribution than in the corresponding tails of the proportional seat gap or first-switch weight. This association underscores the importance of understanding the representational benchmark embedded in the efficiency gap.

Comparing map-drawing authorities reveals three further patterns. First, maps enacted by single-party legislatures are generally associated with less robust first-past-the-post allocations and larger departures from proportionality than maps drawn under split control, by nonpartisan authorities, or by courts.  Second, maps drawn by Republican legislatures show average deviations favoring their own party that are similar to or larger than those of maps drawn by Democratic legislatures.
Third, the measured Republican--Democratic asymmetry is substantially larger under the efficiency gap than under the proportional seat gap or first-switch weight. Given the efficiency gap's widespread use in evaluating electoral maps, this difference underscores how its representational benchmark can shape assessments of partisan advantage in practice.

\section{The Model}
\label{sec:general-model}

We study a two-party election between parties $A$ and $B$ that determines the composition of a statewide delegation consisting of $N\ge 3$ single-member district seats. 

Districts have equal population, and we take the distribution of
electoral support across districts as given. We denote party $A$'s vote share in district $d\in\{1,\dots,N\}$  with $p_d\in[0,1]$, and write $\mathbf p:=(p_1,\ldots,p_N)$ for the voting profile. Since there are only two parties,  party $B$'s vote share in district $d$ is $1-p_d$. We write $\rho:=\sum_{d=1}^N p_d$ for party $A$'s statewide support measured in seat units. The corresponding statewide vote share is $ \bar p:=\rho/N$.\footnote{We take $\mathbf p$ as given and interpret district vote shares as measures of underlying partisan support, assuming that participating voters are representative of their district's partisan preferences. Conditioning on this support isolates how electoral rules translate voter preferences into local and statewide representation. At the end of Section~\ref{sec:general-results}, we discuss how the framework accommodates differential turnout across districts.}

A \emph{seat-allocation rule} assigns the $N$ district seats as a function of the voting profile $\mathbf p$. Formally, for each profile $\mathbf p$, the rule $R$ selects a subset $  R(\mathbf p)\subseteq\{1,\dots,N\},$ where $d\in R(\mathbf p)$ means that party $A$ wins district $d$ and party $B$
wins the complement. Equivalently, an allocation can be represented by the vector $\mathbf x:=(x_1,\dots,x_N)$, where $x_d=1$ if party $A$ wins district $d$ and $x_d=0$ otherwise. For example, if $R$ denotes the first-past-the-post rule, party $A$ wins those districts in which it receives a majority: $R(\mathbf p)=\{d:p_d>1/2\}$. For a generic allocation $\mathbf x$, write $S(\mathbf x):=\sum_{d=1}^N x_d$ for the number of seats awarded to party $A$.

Throughout the text, we relabel districts so that $p_1\ge p_2\ge\cdots\ge p_N$.\footnote{\label{footnote:noties}For notational convenience, we assume distinct district vote shares and no district-level electoral tie, that is, $p_i\neq p_j$ for $i\neq j$ and $p_i\neq 1/2$ for all $i$.} We will frequently refer to the allocation that assigns party $A$ its $S$ strongest districts and call this the top-$S$ allocation. Formally, for any seat total $S$, let $\mathbf x^S$ denote the top-$S$ allocation defined by $x_d^S=1$ if $d\le S$, and $x_d^S=0$ if $d>S$. 

Our goal is to evaluate an electoral outcome by how accurately it translates votes into political representation. In district-based elections, two objectives are natural, yet often in tension. First, \emph{within each district}, the winning party should represent as many local voters as possible. Second, \emph{at the statewide level}, the delegation's party balance should align with party support in the electorate as a whole. 

\medskip
\noindent
\textbf{District-level representation.} 
The district-level component captures the idea that each district seat represents its own local electorate, and we define local misrepresentation as the mismatch between that assignment and the voters in that district.

We impose that any admissible measure of district-level representation must be \textit{anonymous}: two districts with the same party $A$ vote share and the same assignment must generate the same local loss. For any district with party $A$ vote share $p$, we denote the loss from assigning it to party $A$ by $\delta(1,p)$, and the loss from assigning it to party $B$ by $\delta(0,p)$. We write
\[
  \Delta(p):=\delta(1,p)-\delta(0,p)
\]
for the district switch margin, the change in local loss from assigning the district to party $A$ rather than party $B$.

We also assume additive separability across districts. The goal of this component is to capture the local mismatch between a district’s vote share and the party assigned to represent it. That mismatch should not depend on assignments made elsewhere. Otherwise, the loss assigned to one district could vary because of outcomes in other districts, thereby importing statewide considerations into a component intended to be local.

The two assumptions above are especially natural when district-level misrepresentation is interpreted as a function of the voters whose own district seat is held by their nonpreferred party: districts are exchangeable, and voter-level losses accumulate across districts.\footnote{
Although it is natural, we do not require district losses to be party neutral as our results do not rely on it. When \textit{party neutrality} is imposed, the loss from assigning a district to a party depends only on that party's vote share in the district, not on the party's name. In that case, the loss function takes the one-dimensional form $ \delta(p):=\delta(1,p),$ and $\delta(1-p):=\delta(0,p)$.}
Under these two assumptions, total district-level misrepresentation is
\begin{equation*}
  \Dist(\mathbf x,\mathbf p):=\sum_{d=1}^N \delta(x_d,p_d),
\end{equation*}
and we define the district-optimal seat total by
\[
  S^\Dist
  :=
  \min\left\{
    S(\mathbf x):
    \mathbf x\in\argmin_{\mathbf y\in\{0,1\}^N}\Dist(\mathbf y,\mathbf p)
  \right\}.
\]
Thus $S^\Dist$ is the smallest party $A$ seat total attained by any allocation that minimizes district-level misrepresentation. 

Finally, we impose the ordinal restriction of \textit{district responsiveness}: as the party $A$ vote share in a district rises, assigning the district to party $A$ should not become less representative relative to assigning it to party $B$. Equivalently, additional support for party $A$ may strengthen the case for assigning the district to party $A$, or leave that comparison unchanged, but it cannot strengthen the case for assigning the district to party $B$. Formally, the district switch margin $\Delta(p)$ is weakly decreasing in $p$.

\medskip
\noindent
\textbf{Statewide representation.}
The statewide component formalizes the second dimension of representation. Even if each district seat is assigned in a way that best fits its own local electorate, the resulting delegation may still give a party too many or too few seats relative to its statewide support. The statewide delegation has a party composition, and we define statewide misrepresentation as the mismatch between that composition and party support in the electorate as a whole.

To capture this, we assume that there is a \textit{target} seat-vote curve $T(\rho)\in[0,N]$ for party $A$ that depends on its statewide vote total $\rho$, and that statewide misrepresentation depends on party $A$'s seat total $S(\mathbf{x})$ and its distance from this target $T(\rho)$.\footnote{Because there are only two parties and $N$ seats, party $A$'s seat total pins down the full partisan composition of the delegation: if party $A$ receives $S$ seats, party $B$ receives $N-S$. We therefore express the statewide component in terms of party $A$'s seat total and statewide support.} A natural target is statewide proportionality, which is the special case $T(\rho)=\rho$.

Formally, let $Q:[0,N]\to\mathbb R_+$ be strictly increasing and weakly convex, with $Q(0)=0$, and define
\begin{equation*}
  \Agg(\mathbf{x},\mathbf{p})
  :=
  Q\bigl(|S(\mathbf{x})-T(\rho)|\bigr).
\end{equation*}
Because statewide misrepresentation depends on an allocation only through its seat total, we also write $\Agg(S,\rho):=Q(|S-T(\rho)|)$ for the statewide loss associated with seat total $S$ and statewide vote total $\rho$. Since \(Q\) is strictly increasing, larger gaps are more costly. And since \(Q\) is weakly convex, the marginal cost of distance from the target is weakly increasing.\footnote{Although party neutrality is natural,  our results do not require it. Imposing it implies that the target satisfies $T(N-\rho)=N-T(\rho)$.} We define 
\[S^\Agg:=\min\argmin_{S\in\{0,\ldots,N\}}\Agg(S,\rho)\]
as the smallest statewide-optimal integer seat total. Since $Q$ is strictly increasing, $S^\Agg$ is the smallest integer seat total closest to the target $T(\rho)$.

\medskip
\noindent
\textbf{Total misrepresentation.}
We evaluate allocations by aggregating the two losses via $F(\cdot,\cdot)$. For $w\in[0,1]$, define
\begin{equation*}
   M(\mathbf x,\mathbf p;w)
  :=
  F\Big(
    (1-w)\Dist(\mathbf x,\mathbf p),
    \;
    w\Agg(\mathbf x,\mathbf p)
  \Big),
\end{equation*}
where $F:\mathbb R_+^2\to\mathbb R$ is  continuously differentiable
with strictly positive partial derivatives. Since both arguments are misrepresentation terms, lower values of $M$ are preferred.

The weight $w$ determines the relative scale of the two losses, while $F$ determines how the planner substitutes between them once they are placed on that scale. This substitution is summarized locally by the marginal rate of substitution (MRS), $F_2(y_1,y_2)/F_1(y_1,y_2)$. This is the compensating increase in district-level misrepresentation that leaves total misrepresentation unchanged after a one-unit reduction in statewide misrepresentation, evaluated at the loss pair $(y_1,y_2)$. For $w\in[0,1)$, we write $\phi:=w/(1-w)$ for the relative weight on statewide representation.

We say that $F$ has diminishing MRS if the ratio
$F_2(y_1,y_2)/F_1(y_1,y_2)$ is weakly decreasing in $y_1$
and weakly increasing in $y_2$. Under this condition, additional district-level loss becomes harder to compensate when district-level misrepresentation is already high, while reductions in statewide misrepresentation become more valuable when statewide misrepresentation is high. 

A useful benchmark is the linear aggregation, which adds the weighted district and statewide losses: $F(y_1,y_2)=y_1+y_2$, so that $M(\mathbf x,\mathbf p;w)=(1-w)\Dist(\mathbf x,\mathbf p)+w\Agg(\mathbf x,\mathbf p)$. For this aggregator, $F_1=F_2=1$. Hence the planner's pointwise marginal rate of substitution is constant and equal to one.

\section{Optimal Rules}
\label{sec:general-results}

Given a fixed vote profile $\mathbf p$, an allocation $\mathbf x$ \emph{minimizes misrepresentation} at weight $w$ if
\[
  \mathbf x\in
  \argmin_{\mathbf y\in\{0,1\}^N}
  M(\mathbf y,\mathbf p;w).
\]
A rule $R(\cdot)$ is \emph{optimal at weight $w$} if it minimizes misrepresentation at weight $w$ for every vote profile $\mathbf p$.

We characterize the optimal rule in three steps. First, holding the seat total fixed, district responsiveness implies that party $A$ should receive its strongest districts. Second, this conditional result reduces the allocation problem to a one-dimensional choice over seat totals, where optimality can be studied by comparing $S$ and $S+1$ seats. Third, single crossing orders the marginal seat comparisons and yields a cutoff implementation.

\medskip
\noindent
\textbf{Conditional assignment.}
Holding fixed the number of seats awarded to party $A$, district responsiveness implies that those seats can be assigned to the districts where party $A$ is strongest.

Suppose an allocation with $S$ seats for party $A$ gives one of those seats to district $j$ but not to district $i$, where $p_i\ge p_j$. Consider swapping the assignments. By district responsiveness, the relative loss from assigning a district to party $A$ rather than party $B$ is weakly decreasing in $p$, thus the swap cannot increase district-level misrepresentation. Repeating this argument eliminates every inversion, leaving an allocation in which party $A$ receives the $S$ districts where it is strongest, namely the top-$S$ allocation.

\begin{lemma}
\label{lem:conditional-district-assignment}
The top-$S$ allocation $\mathbf x^S$ minimizes district-level misrepresentation among all allocations that award party $A$ a total of $S$ seats.
\end{lemma}

\medskip
\noindent
\textbf{Reduction to seat totals.} For seat totals and any ordered profile $\mathbf p$, we write $  \Dist(S,\mathbf p):=\Dist(\mathbf x^S,\mathbf p).$ By Lemma~\ref{lem:conditional-district-assignment}, this is the minimum district loss among allocations that award party $A$ a total of $S$ seats.

The statewide term depends on an allocation only through the seat total $S(\mathbf x)$, while Lemma~\ref{lem:conditional-district-assignment} identifies an optimal assignment of districts conditional on any fixed seat total. This reduces the allocation problem to a one-dimensional choice over seat totals: choose $S$, then implement it by a top-$S$ allocation.

\begin{lemma}
\label{lem:general-reduction}
Every optimal seat total has an optimal top-$S$ implementation.
Moreover, the optimal seat total is attained by some $S$ between $S^\Dist$ and $S^\Agg$.
\end{lemma}

Seat totals outside the range between the two component-wise optima can be ruled out directly. Suppose that $S^\Dist \leq S^\Agg$. If $S<S^\Dist$, then raising the seat total to $S^\Dist$ improves the district component and also moves the delegation closer to the statewide target. If $S>S^\Agg$, then lowering the seat total to $S^\Agg$ improves the statewide component and, because this remains above the district optimum, single crossing in $\delta$ implies that this does not worsen the district component. Since $F$ is increasing in both loss arguments, any seat total outside $[S^\Dist,S^\Agg]$ is dominated by an endpoint of the interval.

Throughout the remainder of the main text, unless stated otherwise, we consider without loss of generality the case $S^\Dist\le S^\Agg$. This is the case in which the district-optimal allocation gives party $A$ weakly fewer seats than the statewide-optimal seat total. If $S^\Dist>S^\Agg$, the same arguments apply after reversing party labels. If $S^\Dist=S^\Agg$, the district and statewide components select the same seat total, so the tradeoff is degenerate.

\medskip
\noindent
\textbf{District margin.}
For each $S$, the district loss of the top-$S$ allocation is
\begin{equation*}
  \Dist(S,\mathbf p)
  =
  \sum_{d\le S}\delta(1,p_d)
  +
  \sum_{d>S}\delta(0,p_d).
\end{equation*}
Moving from $S$ to $S+1$ keeps the top-$S$ districts assigned to party $A$ and changes only district $S+1$, which is switched from party $B$ to party $A$. The marginal district cost of increasing party $A$'s seat total from $S$ to $S+1$ can be written as
\begin{equation*}
  \Delta\Dist(S,\mathbf p)
  :=
  \Dist(S+1,\mathbf p)-\Dist(S,\mathbf p)
  =
  \Delta(p_{S+1}).
\end{equation*}
District responsiveness implies that $p\mapsto \Delta(p)$ is weakly decreasing. Since the shares are ordered, $p_{S+1}$ is weakly decreasing in $S$. Consequently, the sequence $S\mapsto\Delta\Dist(S,\mathbf p)$ is weakly increasing and therefore $\Dist(S,\mathbf p)$ is discretely convex.

\medskip
\noindent
\textbf{Statewide margin.}
Consider increasing party $A$'s seat total from $S$ to $S+1$. Along the range $S=S^\Dist,\ldots,S^\Agg-1$, the move raises district-level misrepresentation and lowers statewide misrepresentation. The marginal statewide improvement is
\begin{equation*}
  \Delta\Agg(S,\rho):=\Agg(S,\rho)-\Agg(S+1,\rho)>0.
\end{equation*}
Because $S\mapsto\Agg(S,\rho)$ is discretely convex and decreases as $S$ moves toward $S^\Agg$ from the left, the sequence $S\mapsto\Delta\Agg(S,\rho)$ is weakly decreasing on this range. Intuitively, each additional seat moves the delegation closer to the statewide target, but the remaining statewide gain becomes weakly smaller as the target is approached.

Under the top-\(S\) representation, both components of the problem live on the same one-dimensional seat-count domain. The seat total \(S\) determines the district loss \(\Dist(S,\mathbf p)\) and the statewide loss \(\Agg(S,\rho)\), so we can write the full objective directly as a function of \(S\):
\begin{equation*}
  M(S,\mathbf p;w)
  :=   M(\mathbf x^S,\mathbf p;w) =
  F\bigl((1-w)\Dist(S,\mathbf p),w\Agg(S,\rho)\bigr).
\end{equation*}

\medskip
\noindent
\textbf{Marginal rate of substitution.}
Increasing party $A$'s seat total from $S$ to $S+1$ entails discrete changes in both district-level and statewide misrepresentation. Evaluating this move requires accounting for the planner's marginal rate of substitution over the line segment connecting the weighted loss pairs associated with $S$ and $S+1$. We therefore use the average MRS along this segment. To define this average, fix $S\in\{S^\Dist,\ldots,S^\Agg-1\}$
and let $z_S(\nu;\mathbf p,w)$ trace the segment from the weighted
loss pair at $S$ to that at $S+1$:
\begin{equation*}
  z_S(\nu;\mathbf p,w)
  :=
  \Bigl(
    (1-w)\bigl(\Dist(S,\mathbf p)+\nu\Delta\Dist(S,\mathbf p)\bigr),
    w\bigl(\Agg(S,\rho)-\nu\Delta\Agg(S,\rho)\bigr)
  \Bigr),
  \quad \nu\in[0,1].
\end{equation*}
The corresponding average marginal rate of substitution along this line segment is
\begin{equation*}
  \MRS(S,\mathbf p;w)
  :=
  \frac{\int_0^1 F_2(z_S(\nu;\mathbf p,w))\,d\nu}
       {\int_0^1 F_1(z_S(\nu;\mathbf p,w))\,d\nu}.
\end{equation*}
This expression averages the pointwise MRS $F_2/F_1$ along the
segment, with weights proportional to $F_1$. It measures the planner’s average willingness to accept additional weighted district-level misrepresentation in exchange for a one-unit reduction in weighted statewide misrepresentation over the move from $S$ to $S+1$. A higher $\MRS(S,\mathbf p;w)$ means that the planner is willing to accept a larger district-level cost for the same statewide improvement.

Under linear aggregation, $F(y_1,y_2)=y_1+y_2$, both partial derivatives equal
one, so $\MRS(S,\mathbf p;w)=1$. When the profile and weight
are fixed, we suppress these arguments and write $z_S(\nu)$ and
$\MRS(S)$.

\medskip
\noindent
\textbf{Marginal rate of transformation.}
Increasing party $A$'s seat total from $S$ to $S+1$ reduces statewide
misrepresentation at the cost of increasing district-level
misrepresentation. The marginal rate of transformation (MRT)
measures this cost per unit of statewide improvement. For $w\in(0,1)$,
we define it as
\begin{equation*}
  \MRT(S,\mathbf p;w)
  :=
  \frac{(1-w) \cdot \Delta\Dist(S,\mathbf p)}
       {w \cdot              \Delta\Agg(S,\rho)}
  =
  \frac{\Delta\Dist(S,\mathbf p)}
       {\phi\cdot\Delta\Agg(S,\rho)}.
\end{equation*}
This ratio measures how much weighted district-level misrepresentation
the move adds per unit of weighted statewide misrepresentation it
eliminates. As $S$ increases, successive seats carry higher
district costs and yield weakly smaller statewide gains.
Consequently, $\MRT$ is increasing in $S$. When the profile
and weight are fixed, we write $\MRT(S)$.

\medskip
\noindent
\textbf{The cutoff.}
The move from $S$ to $S+1$ weakly reduces total misrepresentation
when the planner values the statewide improvement at least as much
as the district loss it requires. For $w\in(0,1)$, this gives the familiar
MRT$\leq$MRS condition:
\begin{equation}
\label{eq:kappa-foc}
  M(S+1,\mathbf p;w)\le M(S,\mathbf p;w)
  \quad\Longleftrightarrow\quad
  \MRT(S,\mathbf p;w)\le\MRS(S,\mathbf p;w).
\end{equation}
Using the definition of $\MRT(\cdot)$ and
$\Delta\Dist(S,\mathbf p)=\Delta(p_{S+1})$, we can restate this condition as:
\[
  \Delta(p_{S+1})
  \le
  \phi\cdot\Delta\Agg(S,\rho)\cdot\MRS(S,\mathbf p;w).
\]
Since $\Delta$ is weakly decreasing, this inequality yields the
vote-share cutoff
\begin{equation*}
  \tau(S,\mathbf p;w)
  :=
  \Delta^{-1}
  \left(
    \phi\cdot\Delta\Agg(S,\rho)\cdot\MRS(S,\mathbf p;w)
  \right).
\end{equation*}
The move from $S$ to $S+1$ weakly reduces total misrepresentation
if $p_{S+1}$ is above the cutoff and strictly increases it if
$p_{S+1}$ is below the cutoff.\footnote{\label{fn:generalized-inverse}We use the generalized inverse
$\Delta^{-1}(c)=\inf\{u\in[0,1]:\Delta(u)\le c\}$,
with $\inf\emptyset=+\infty$. Shares above the cutoff
qualify and shares below it do not. A share at
the cutoff qualifies if and only if $\Delta(u)\le c$.
This convention accommodates discontinuities in $\Delta(\cdot)$. Throughout, $u \geq \Delta^{-1}(c)$ is understood to mean either $u > \Delta^{-1}(c)$, or $u = \Delta^{-1}(c)$ and $\Delta(u) \leq c$.}
When the profile and weight are fixed, we write $\tau(S)$. The following theorem states this cutoff characterization.

\begin{theorem}
\label{thm:general-cutoff}
For each $S\in \{ S^\Dist,\ldots,S^\Agg-1\}$, the one-seat change from $S$ to $S+1$ is weakly improving if and only if the next district in the ordered list clears the cutoff:
\[
  M(S+1,\mathbf p;w)\le M(S,\mathbf p;w)
  \quad\Longleftrightarrow\quad
  p_{S+1}\ge \tau(S,\mathbf p;w).
\]
\end{theorem}

\medskip
\noindent
\textbf{Optimal allocation rule.} Theorem~\ref{thm:general-cutoff} determines whether increasing party $A$'s seat total from $S$ to $S+1$ weakly lowers total misrepresentation. With diminishing MRS, these one-seat comparisons identify an optimal seat total, starting from $S^{\Dist}$ and considering successive seats up to $S^{\Agg}$.

\begin{theorem}
\label{thm:optimal-cutoff-rule}
Suppose that \(F\) has diminishing MRS. 
\begin{itemize}
    \item  \(\tau(S-1,\mathbf p;w)\) is weakly increasing in \(S\) on \(\{S^\Dist+1,\ldots,S^\Agg\}\).
    \item There is an optimal cutoff allocation: party \(A\) receives those districts whose ordered vote shares satisfy $p_S\ge \tau(S-1,\mathbf p;w)$.
\end{itemize}
\end{theorem}

Theorem~\ref{thm:optimal-cutoff-rule} turns the marginal comparison into a full allocation rule. Diminishing MRS makes the cutoff tests single crossing, so once the one-seat change from $S$ to $S+1$ fails to reduce misrepresentation, all subsequent one-seat changes also fail to reduce it. The next step is to relate the statewide weight \(w\) to the seat total selected by the optimal allocation. 

\medskip
\noindent
\textbf{Supporting weights.} For each \(S\in\{S^\Dist,\ldots,S^\Agg\}\), define \(\mathcal W_S(\mathbf p)\) as the set of weights \(w\in[0,1]\) for which \(S\) is an optimal seat total:
\[
  \mathcal W_S(\mathbf p)
  :=
  \left\{
    w\in[0,1]:
    S\in\argmin_{s\in\{0,\ldots,N\}}M(s,\mathbf p;w)
  \right\}.
\]

Proposition~\ref{prop:monotone-support-sets} shows that these supporting-weight sets are intervals, and that they are ordered in the same direction as the seat totals.

\begin{proposition}
\label{prop:monotone-support-sets}
If $F$ has diminishing MRS, then each $\mathcal W_S(\mathbf p)$ is a nonempty closed interval. Moreover, these intervals are ordered by seat total: writing $ \mathcal W_S(\mathbf p)=[\ell_S,u_S],$ we have
\[
  0=\ell_{S^\Dist}
  \le
  \ell_{S^\Dist+1}
  \le
  \cdots
  \le
  \ell_{S^\Agg},
  \quad \text{ and } \quad 
  u_{S^\Dist}
  \le
  u_{S^\Dist+1}
  \le
  \cdots
  \le
  u_{S^\Agg}=1.
\]
\end{proposition}

The result also formalizes the monotone path from the district-optimal seat total to the statewide-optimal seat total. At low $w$, the objective places little value on reducing statewide misrepresentation, so the optimal allocation stays near $S^\Dist$. As $w$ rises, the statewide gain from adding seats for party $A$ receives more weight, and the rule becomes willing to admit districts in which party $A$ has weaker support. Under diminishing MRS, once a higher seat total is justified, further increases in $w$ do not push the optimum back toward fewer seats, so the optimal seat total moves weakly upward with the relative weight on statewide representation.

\begin{figure}[t]
\centering
\includegraphics[width=0.6\linewidth]{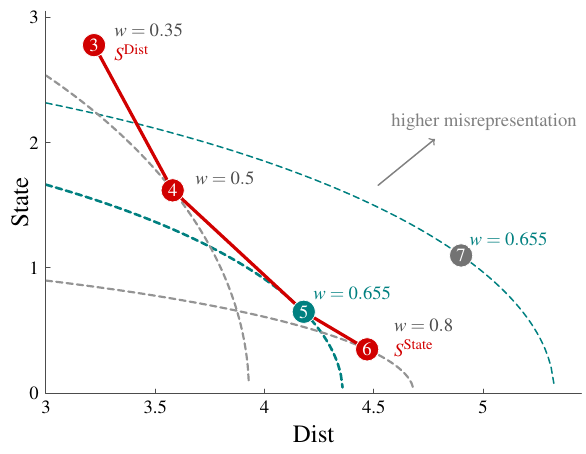}
\caption{\small Red circles mark Pareto-frontier top-\(S\) allocations, and the red path connects them. Dashed curves are iso-misrepresentation level sets. The teal point and lower teal curve highlight the allocation selected at \(w=0.655\), while the upper teal curve shows a higher-misrepresentation level set passing through a Pareto-inefficient top-\(S\) allocation.}
\label{fig:thm-illustrate}
\end{figure}

\medskip
\noindent
\textbf{Illustration.}
Figure~\ref{fig:thm-illustrate} illustrates the optimality condition in \eqref{eq:kappa-foc} and the comparative static in Proposition~\ref{prop:monotone-support-sets}. For a fixed vote profile, the red points are the feasible frontier points generated by Pareto optimal top-\(S\) allocations, and the red path connects them for visual representation.\footnote{Because seat totals are discrete, the frontier consists of the displayed top-\(S\) points, and the connecting line is only a visual representation. So an iso-misrepresentation curve may cross it between neighboring seat totals without affecting the discrete optimality comparison. Allocations assigning fewer than \(3\) seats or more than \(8\) seats to party \(A\) are not plotted because their misrepresentation pairs lie outside the displayed axes.} Moving southeast along the frontier awards the underrepresented party \(A\) more seats: district-level misrepresentation rises, while statewide misrepresentation falls. 

The optimality condition compares the frontier cost of the next seat increase, \(\MRT(S,\mathbf p;w)\), with the planner's willingness to trade statewide for district-level loss, \(\MRS(S,\mathbf p;w)\). This is a discrete tangency condition: the frontier's marginal rate of transformation is compared with the planner's marginal rate of substitution. Graphically, the selected frontier point is the top-\(S\) allocation supported by the lowest attainable iso-misrepresentation level set at that value of \(w\). With diminishing MRS, supporting $w$ intervals are ordered: increasing \(w\) makes statewide loss more important and moves the supported seat total weakly southeast toward the statewide target.

\medskip \noindent \textbf{Differential Turnout.} The same approach extends to differential turnout across districts. Let $T_d>0$ denote turnout in district $d$ relative to average district turnout, and interpret $p_d$ as the vote share among voters who turned out to vote. Taking turnout as given, weighting each district's loss by $T_d$ gives the district switch margin $c_d:=T_d\Delta(p_d)$, while statewide support becomes $\rho:=\sum_d T_dp_d$. The statewide component retains its form, evaluated at this statewide support. Conditional on a seat total, district-level misrepresentation is minimized by assigning party $A$ the districts with the smallest $c_d$. In this ordering, marginal district costs are weakly increasing. The reduction to seat totals and, under diminishing MRS, the single crossing argument in Theorem~\ref{thm:optimal-cutoff-rule} therefore continue to apply, with a cutoff in the district switch margins $c_d$ instead of $p_d$.

\section{The Linear Model}
\label{sec:linear}

This section studies a linear special case of our model. District-level misrepresentation is the mass of voters whose own district is represented by their nonpreferred party, statewide misrepresentation is the vote-seat gap, and the planner aggregates the two losses linearly.

Formally, for a district with party $A$ vote share $p$, we set $  \delta(1,p)=1-p,$ and $  \delta(0,p)=p.$ For any allocation $\mathbf x$, district-level misrepresentation becomes:
\[
  \Dist(\mathbf x,\mathbf p)
  =
  \sum_{d=1}^N
  \Bigl[
    x_d(1-p_d)+(1-x_d)p_d
  \Bigr].
\]

For statewide misrepresentation, we set the statewide target equal to party $A$'s proportional seat total, $T(\rho)=\rho$, and use the linear loss function $Q(z)=z$ to obtain:
\[
  \Agg(\mathbf{x},\mathbf{p})=|S(\mathbf{x})-\rho|.
\]

Note that both components lie on the same scale of $[0,N]$.\footnote{For $\Dist$, the value is $0$ when every district is assigned to a party supported by all voters in that district, and $N$ when every district is assigned to a party supported by no voters in that district. For $\Agg$, the value is $0$ when the seat total equals party $A$'s proportional seat total, $S(\mathbf x)=\rho$, and $N$ when one party receives all seats despite receiving no statewide support.} $\Dist$ is the normalized mass of voters locally represented by their nonpreferred party, while $\Agg$ is the number of seats by which the delegation departs from proportionality.

Finally, we specialize to the linear aggregator: $F(y_1,y_2)=y_1+y_2$. For $w\in[0,1]$, the misrepresentation objective becomes:
\begin{equation*}
  M(\mathbf x,\mathbf p;w)
  :=
  (1-w)     \sum_{d=1}^N
  \Bigl[
    x_d(1-p_d)+(1-x_d)p_d
  \Bigr]
  +
  w|S(\mathbf x)-\rho|.
\end{equation*}
The parameter $w$ is the weight on statewide misrepresentation. At $w=0$, the objective is purely district-level and awards each district to its local majority. At $w=1$, only statewide misrepresentation matters, so the rule selects an integer seat total closest to proportionality.

\subsection{Pareto frontier}\label{subsec:frontier}
For a fixed vote-share profile $\mathbf p$, each allocation generates a loss pair \(  \bigl(\Dist(\mathbf x,\mathbf p),\Agg(S(\mathbf x),\rho)\bigr).\)
An allocation $\mathbf x$ is \emph{Pareto efficient} if no alternative allocation weakly reduces both district-level and statewide misrepresentation, with at least one strict reduction.

In multiobjective problems, weighted sums identify \emph{supported} Pareto points: points on the lower convex envelope of the feasible set of loss pairs. In finite, nonconvex problems, some Pareto-efficient points can be \emph{unsupported}: they are efficient, but no linear weight makes them optimal. The next proposition shows that, under these assumptions,every Pareto-efficient point is supported.

\begin{proposition}
\label{prop:linear-frontier}
An allocation $\widehat{\mathbf x}$ is Pareto efficient if and only if there exists a weight $w\in[0,1)$ such that
\[
  \widehat{\mathbf x}
  \in
  \argmin_{\mathbf x\in\{0,1\}^N}
  \left\{
    (1-w)\Dist(\mathbf x,\mathbf p)
    +
    w|S(\mathbf x)-\rho|
  \right\}.
\]
Moreover, $ \widehat{\mathbf x}$ is a top-$S$ allocation.
\end{proposition}

The intuition for Proposition~\ref{prop:linear-frontier} is as follows. For a fixed seat total $S$, replacing an allocation with the top-$S$ allocation leaves statewide misrepresentation unchanged and weakly lowers district-level misrepresentation. We can therefore identify the Pareto frontier among top-$S$ allocations. In the linear model, the marginal tradeoffs along this frontier are ordered, so every Pareto-efficient frontier point is supported at some weight.

The relative weight $\phi=w/(1-w)$ has the revealed-preference interpretation used in Section~\ref{sec:empirics}. It is the amount of additional district-level misrepresentation the rule is willing to accept in order to reduce statewide misrepresentation by one unit. Equivalently, if an observed allocation is rationalized by normalized weight $w$, then $\phi$ is the shadow price of a one-unit reduction in statewide misrepresentation, expressed in units of the district-level misrepresentation measure.

\begin{figure}[t]
\centering
    \includegraphics[width=.7\linewidth]{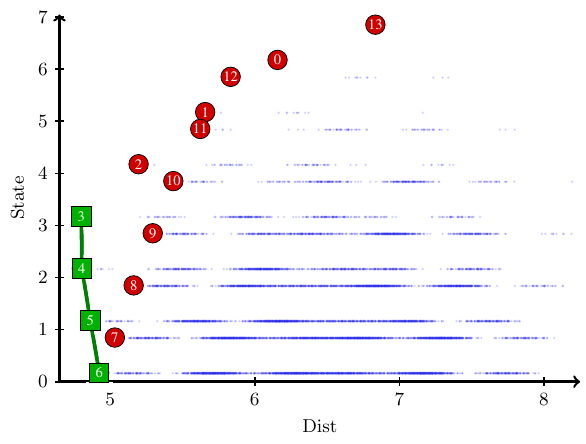}
\caption{\small Feasible misrepresentation pairs $(\Dist,\Agg)$ for the profile in the left panel of Figure~\ref{fig:mainthm}. Red points are the top-$S$ allocations. Green squares trace the optimal allocation as $w$ varies. The number inside each top-$S$ allocation denotes the corresponding seat total \(S\). Blue points are all other feasible allocations.}  \label{fig:dist-agg-all-allocations}
\end{figure}

Figure~\ref{fig:dist-agg-all-allocations} visualizes the Pareto frontier of feasible misrepresentation pairs \((\Dist,\Agg)\) via the green squares. As \(w\) increases, the misrepresentation-minimizing allocation moves along this frontier by jumping between successive green points. At any switching weight between two consecutive green squares, the indifference condition implies that the connecting segment has slope \(-1/\phi\), i.e. the weighted marginal rate of substitution between \(\Agg\) and \(\Dist\) at the switch.
 
\medskip
\noindent
\textbf{Welfare Theorems.}
Under the assumptions of Section~\ref{sec:general-model}  and district responsiveness, Proposition~\ref{prop:linear-frontier} extends to the general losses $\Dist$ and $\Agg$ for $\phi>0$. Any minimizer of $\Dist+\phi\Agg$ with $\phi>0$ is Pareto efficient, by the logic of the \textit{First Welfare Theorem.} Conversely, the top-$S$ frontier has weakly increasing marginal tradeoffs, so every efficient point admits a supporting price $\phi>0$, giving an analog of the \textit{Second Welfare Theorem}. Thus varying $\phi$ traces the entire efficient frontier.

\subsection{A Voter-Loss Foundation}

The frontier result shows that the linear family is rich enough to capture all efficient tradeoffs. We next give the same objective a voter-level interpretation. Voters value representation along two dimensions. First, they value local representation: they prefer their own district seat to be held by their preferred party. Second, they value statewide representation: they prefer their preferred party to hold more seats in the statewide delegation. We express these preferences in loss form. Voters incur a local loss when their own district is represented by their nonpreferred party and a statewide loss from any shortfall between their party's statewide representation and its statewide support. We normalize this shortfall by the party's proportional seat target,
so equal proportional shortfalls generate equal individual losses
across parties.

We write $(z)_+$ to mean $\max\{z,0\}$. For a voter who supports party $A$ in district $d$, define loss under allocation $\mathbf x$ by
\begin{equation*}
  \ell_d^A(\mathbf x,\mathbf p;w)
  =
  (1-w)(1-x_d)
  +
  w\frac{(\rho-S(\mathbf x))_+}{\rho}.
\end{equation*}
For a voter who supports party $B$ in district $d$, define loss under allocation $\mathbf x$ by
\begin{equation*}
  \ell_d^B(\mathbf x,\mathbf p;w)
  =
  (1-w)x_d
  +
  w\frac{(S(\mathbf x)-\rho)_+}{N-\rho}.
\end{equation*}
The first term in each expression captures local representation. It reflects the loss a voter experiences when her own district is represented by a member of the other party: an $A$-voter incurs this loss when district $d$ is assigned to party $B$, and a $B$-voter incurs it when district $d$ is assigned to party $A$. For an $A$-voter, the statewide loss term measures the shortfall between party $A$'s seat total and its statewide support.\footnote{At $\rho=0$, set $(\rho-S(\mathbf x))^+/\rho=0$, and
at $\rho=N$, set $(S(mathbfx)-\rho)_+/(N-\rho)=0$.} The loss is maximal when party $A$ receives no seats, decreases linearly as $S(\mathbf x)$ rises, and reaches zero once party $A$'s seat total reaches its statewide support $\rho$. In loss form, this term is $(\rho-S(\mathbf x))_+/\rho$. For a $B$-voter, the analogous shortfall is between party $B$'s seat total $N-S(\mathbf x)$ and its statewide support $N-\rho$, which gives the term $(S(\mathbf x)-\rho)_+/(N-\rho)$.

The aggregate voter loss is:
\begin{equation*}
  L(\mathbf x,\mathbf p;w)
  :=
  \sum_{d=1}^N
  \left[
    \underbrace{p_d \ell_d^A(\mathbf x,\mathbf p;w)}_{\substack{\text{total loss of party $A$ voters}\\\text{in district $d$}}}
    +
    \underbrace{(1-p_d)\ell_d^B(\mathbf x,\mathbf p;w)}_{\substack{\text{total loss of party $B$ voters}\\\text{in district $d$}}}
  \right].
\end{equation*}

\begin{proposition}
\label{prop:microfoundation}
For any allocation $\mathbf x$, aggregate voter loss is
\begin{equation*}
  L(\mathbf x,\mathbf p;w)
  =
  (1-w)\Dist(\mathbf x,\mathbf p)+w|S(\mathbf x)-\rho|.
\end{equation*}
Therefore, minimizing aggregate voter loss is equivalent to minimizing the linear objective.
\end{proposition}

\subsection{The Optimal Election Rule}
\label{subsec:linear-cutoff-rule}

We now fully characterize the optimal election rules in the linear model. To do so, we first formalize the class of cutoff rules used in Theorem~\ref{thm:optimal-cutoff-rule}: rules that assign seats to Party \(A\) in the districts where its vote share clears a threshold.

\paragraph{Cutoff rules.}
A rule \(R\) is a cutoff rule if, for every profile \(\mathbf p\), there is a threshold \(t(\mathbf p)\in[0,1]\) such that party \(A\) wins  the districts whose vote shares are above the cutoff:
\[
  R(\mathbf p)=\{d:p_d\ge t(\mathbf p)\}.
\]

Two widely used rules are cutoff rules. In the two-party setting, first-past-the-post uses the cutoff $t(\mathbf p)=1/2$.

In our district-based setting, proportional representation adapts proportionality to district-based seat assignments. It chooses the integer seat total closest to party $A$'s proportional entitlement,
$S^{PR}(\mathbf p):=\min\argmin_{S\in\{0,\ldots,N\}}|S-\rho|$,
and assigns those seats to party $A$'s strongest districts: $  R(\mathbf p)=\{d:d\le S^{PR}(\mathbf p)\}$. Thus, proportional representation is a cutoff rule with a
profile-dependent threshold:
\[
  t(\mathbf p)
  :=
  \begin{cases}
    1, & S^{PR}(\mathbf p)=0,\\
    p_{S^{PR}(\mathbf p)}, & S^{PR}(\mathbf p)\in\{1,\ldots,N\}.
  \end{cases}
\]

Note that the proportional seat total depends on the fractional part of \(\rho\). Let \(\underline\rho:=\lfloor\rho\rfloor\) denote the integer part of statewide support.\footnote{For notational convenience, when $\underline{\rho}=0$, set $p_0:=1$, so the cutoff formula below is well defined at the boundary.} Then \(S^{PR}(\mathbf p)=\underline\rho\) when \(\rho-\underline\rho\le 1/2\), while \(S^{PR}(\mathbf p)=\underline\rho+1\) when \(\rho-\underline\rho>1/2\). In the first case, proportional representation rounds down; in the second, it rounds up.

\paragraph{Optimal rule.}
Without loss of generality, label parties so that $S^\Dist\le\underline\rho$. The linear model yields a simple formula for how the optimal cutoff varies with the relative weight on statewide misrepresentation. Under the linear district loss, $\Delta(p)=1-2p$, and linear aggregation gives $\MRS(S,\mathbf p;w)=1$.

Consider first the round-down case, $\rho-\underline\rho\le 1/2$. Each additional seat up to $\underline\rho$ reduces statewide loss by one unit. Theorem~\ref{thm:general-cutoff} therefore gives the cutoff $(1-\phi)/2$ for these seats. As $\phi$ increases, this cutoff falls, awarding party $A$ successively weaker districts until it reaches $\underline\rho$ seats.

Once $\underline\rho$ seats are reached, the top-$\underline\rho$ allocation is already proportional representation. Awarding party $A$ one additional seat would increase statewide misrepresentation, since $\rho$ is weakly closer to $\underline\rho$ than to $\underline\rho+1$, and it would also increase district-level misrepresentation, since the next district was lost by party $A$ under first-past-the-post. Hence, once proportional representation is reached, the cutoff remains constant at $p_{\underline\rho}$. An optimal cutoff in the round-down case is therefore
\[
t^{\mathrm D}(\mathbf p,\phi)
:=
\max\left\{
p_{\underline\rho},
\frac{1-\phi}{2}
\right\}.
\]
The maximum captures the two phases illustrated in the left panel of Figure~\ref{fig:mainthm}: the cutoff follows $(1-\phi)/2$ until proportional representation is reached and remains constant thereafter. At $\phi=0$, the rule coincides with first-past-the-post. As $\phi$ rises, the cutoff falls along $(1-\phi)/2$, awarding party \(A\) seats in the strongest districts it lost under first-past-the-post, until the top-$\underline\rho$ allocation is reached. After that point, the cutoff stays fixed and continues to implement proportional representation. 

Now consider the round-up case, $\rho-\underline\rho>1/2$, so proportional representation requires one additional seat beyond $\underline\rho$. The comparisons for seats up to $\underline\rho$ are unchanged. The statewide gain from the final seat is smaller than the gain from the earlier seats: moving from $\underline\rho$ to $\underline\rho+1$ reduces the statewide gap from $\rho-\underline\rho$ to $\underline\rho+1-\rho$, a gain of $\Delta\Agg(\underline\rho,\rho)
=
2(\rho-\underline\rho)-1.$ Thus the cutoff associated with the rounding seat is $\left[1-\phi\cdot\Delta\Agg(\underline\rho,\rho)\right]/2,$ which is larger than $(1-\phi)/2$ whenever $\phi>0$. Its smaller statewide gain requires a stricter cutoff than the preceding seats.

Once the final seat is awarded, any further seat would increase both district-level and statewide misrepresentation. The rounding-seat cutoff therefore stops at $p_{\underline\rho+1}$. Define
\[
t^{\mathrm U}(\mathbf p,\phi)
:=
\max\left\{
p_{\underline\rho+1},
\frac{1-\phi\cdot
\bigl(\Delta\Agg(\underline\rho,\rho)\bigr)_+}{2}
\right\}.
\]
Using $\bigl(\Delta\Agg(\underline\rho,\rho)\bigr)_+$ replaces any negative statewide gain with zero, ensuring that, when proportionality rounds down, $t^{\mathrm U}=1/2$ and the cutoff is inactive.

\begin{figure}[t]
    \centering
    \includegraphics[width=1\linewidth]{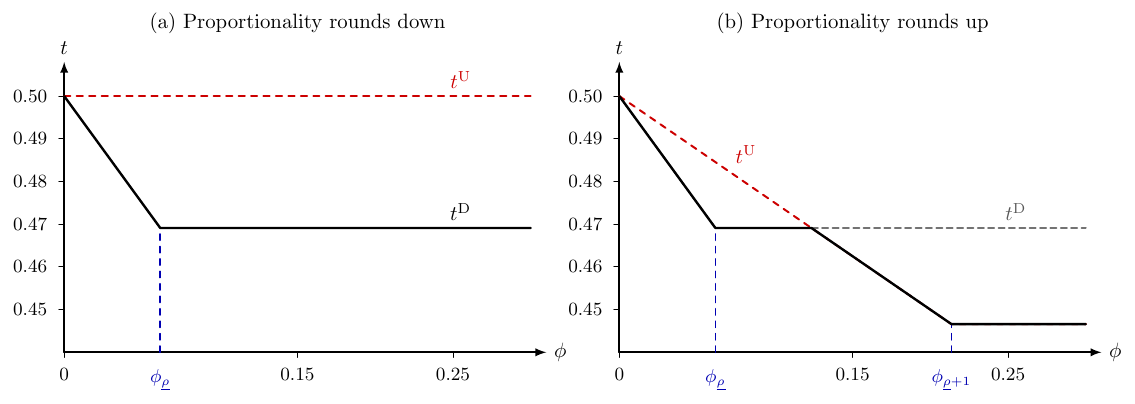}
    \caption{\small Optimal cutoff rule $t_\phi(\cdot)$ as the relative weight on statewide misrepresentation $\phi$ varies. Left: the actual 2018 North Carolina vote profile; proportionality rounds down. Right: a perturbed version; proportionality rounds up. The solid black curve shows $t_\phi(\mathbf p)=\min\{t^{\mathrm D}(\mathbf p,\phi),t^{\mathrm U}(\mathbf p,\phi)\}$. Red and gray dashed lines show  $t^{\mathrm U}$ and $t^{\mathrm D}$, respectively.}
    \label{fig:mainthm}
\end{figure}

The cutoff $t^{\mathrm U}$, however, need not be the operative cutoff immediately after the top-$\underline\rho$ allocation is reached. The final seat yields a smaller statewide gain than the unit improvement associated with each preceding seat, so $t^{\mathrm U}$ falls more slowly with $\phi$ than $(1-\phi)/2$. In the right panel of Figure~\ref{fig:mainthm}, this difference appears in the shallower slope of the red dashed line relative to the initial black segment: when the black line reaches $p_{\underline\rho}$, the red dashed line is still above it. The optimal cutoff therefore combines the two comparisons:
\[
t_\phi(\mathbf p)
:=
\min\left\{
t^{\mathrm D}(\mathbf p,\phi),
t^{\mathrm U}(\mathbf p,\phi)
\right\}.
\]
The corresponding optimal cutoff rule is $R_w(\mathbf p):=\{d:p_d\ge t_\phi(\mathbf p)\}$, where $\phi=w/(1-w)$.

Taking the minimum of the two cutoffs assigns party $A$ every district selected by either cutoff. Once party $A$ receives $\underline\rho$ seats, $t^{\mathrm D}$ remains at $p_{\underline\rho}$. In the round-down case, $t^{\mathrm U}$ remains at $1/2$, and $t^{\mathrm D}$ alone determines the allocation. In the round-up case, when $t^{\mathrm U}$ falls below $p_{\underline\rho}$, the optimal cutoff follows it until reaching $p_{\underline\rho+1}$, where party $A$ receives the final seat needed for proportional representation. The cutoff remains constant thereafter.

\begin{theorem}
\label{thm:linear-cutoff-path}
The family $\{ R_w(\mathbf p)\}_{w \in [0,1]}$ implements the optimal allocation under vote profile $\mathbf p$ and relative weight $\phi=\frac{w}{1-w}$.
\end{theorem}

The theorem gives a cutoff implementation of the optimal rule.\footnote{At \(w=1\), we set \(\phi=+\infty\), and the cutoff implements \(S^{PR}\).} It also reveals the underlying seat-total path. Proposition~\ref{prop:monotone-support-sets} shows generally that, under diminishing MRS, the weights supporting each seat total form ordered intervals. In the linear model, this interval structure reduces to marginal switch weights: except at the switch points, the optimal seat total is unique, so the weight space is partitioned into regions that select successive seat totals from \(S^\Dist\) to \(S^\Agg\). 

\begin{corollary}
\label{cor:step-function}
If $S^\Dist = S^\Agg$, this seat total is optimal for every $w$. Otherwise, there exist weights $  0=:w_{S^\Dist}<w_{S^\Dist+1}\le\cdots\le w_{S^\Agg}\le 1$ such that:
\begin{enumerate}[label=\textup{(\roman*)},leftmargin=2em]
  \item for each $S\in\{S^\Dist,\ldots,S^\Agg-1\}$, if $w\in(w_S,w_{S+1})$, then the optimal seat total is unique and equals $S$;
  \item if $w \in (w_{S^\Agg},1)$, then the optimal seat total is unique and equals $S^\Agg$;
  \item at each switch weight $w=w_{S+1}$, both $S$ and $S+1$ minimize the linear objective.
\end{enumerate}
\end{corollary}

Thus, for an underrepresented party, increasing the statewide weight weakly raises its optimal seat total, moving the allocation monotonically from the district-optimal outcome toward the statewide-optimal one. The movement is discrete because seats change only when the cutoff crosses an observed district vote share, as illustrated in Figure~\ref{fig:weight-seats}.

\begin{figure}[t] 
  \centering
  \includegraphics[width=.88\linewidth]{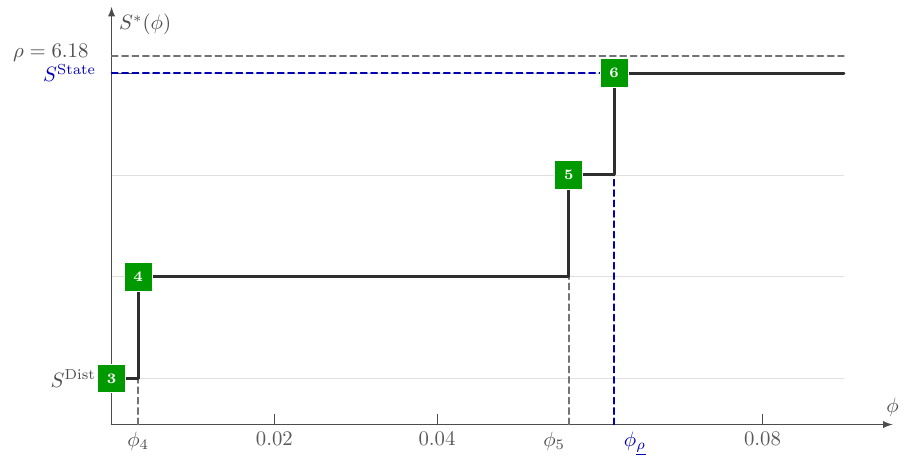}
  \caption{\small Optimal seat total $S^*(\phi)$ for the profile in the left panel of Figure~\ref{fig:mainthm}.}
  \label{fig:weight-seats}
\end{figure}

\section{Analysis of the Optimal Family}

Theorem~\ref{thm:linear-cutoff-path} characterizes the optimal rules in the linear model. We collect these rules in the family
\[
  \mathcal R:=\{R_\lambda:\lambda\in[0,1]\},
\]
where \(R_\lambda\) is the cutoff rule that minimizes misrepresentation at weight \(\lambda\). We use \(\lambda\) to index the rule itself, reserving \(w\) for the designer's preferences over district-level and statewide misrepresentation, or equivalently \(\phi\) for the induced relative weight, in the comparisons below. This family runs from first-past-the-post at \(\lambda=0\) to proportional representation at \(\lambda=1\), and it traces the optimal tradeoffs between district-level and statewide misrepresentation in the linear model. We focus on this family because, at any vote-share profile, every allocation is either selected by some $R_\lambda\in\mathcal R$ or Pareto dominated by an allocation selected by a rule in $\mathcal R$.

\subsection{Characterization of Misrepresentation-Minimizing Rules}

We begin by characterizing two prominent members of the optimal family \(\mathcal R\). The district endpoint \(R_0\) is first-past-the-post: each seat is assigned to the party with the higher vote share in that district. The statewide endpoint \(R_1\) is the adaptation of proportional representation to our district-based setting: it selects the integer seat total closest to party \(A\)'s statewide support \(\rho\). We show that these two endpoints are fully characterized within \(\mathcal R\) by two properties: strong monotonicity and gerrymandering-proofness.

\medskip
\noindent
\textbf{A monotonicity characterization of first-past-the-post.}
A rule satisfies \emph{strong monotonicity} if, for all profiles $\mathbf p,\mathbf p'$ with $\mathbf p'\ge \mathbf p$ componentwise, $  R(\mathbf p)\subseteq R(\mathbf p').$ Strong monotonicity says that raising party $A$'s support district by district cannot cause it to lose any district it previously won.

\begin{proposition}
\label{prop:fptp-char}
$R_{\lambda} \in \mathcal R$ satisfies strong monotonicity if and only if $\lambda=0$.
\end{proposition}

The intuition follows from the structure of the objective. When $\lambda=0$ the criterion is district-separable and each district is awarded to its local majority, whereas any $\lambda>0$ couples districts through the statewide term and induces cross-district reallocation of seats on some profiles, violating monotonicity. For each \(\lambda>0\), one can construct a pair of profiles with the following property: increasing Party $A$'s support in one district can leave the optimal seat total unchanged while changing the ranking of two sufficiently close districts, so a seat shifts to the improved district. As a result, a district can lose its seat even though its own vote share is unchanged (and is still above $\tfrac12$), which is ruled out under the fixed-threshold rule of $R_0$.

\medskip 
\noindent \textbf{A gerrymandering-proof characterization of proportionality.}
A rule is \emph{gerrymandering--proof} if no party can redistrict their aggregate votes across districts and increase their seat share. Define the set of profiles where Party $A$ obtains a vote share of $m$:
\begin{equation*}
  \mathcal P(m)
  :=
  \{\mathbf q\in[0,1]^N:\bar q=m\}.
\end{equation*}
If $\mathbf{q} \in \mathcal P(\bar p)$, $\mathbf q $ is a redistricting of support across districts relative to $\mathbf p$.\footnote{Throughout this analysis, profiles in $\mathcal P(m)$ are understood to satisfy the distinct-share assumption in Footnote \ref{footnote:noties}; tied profiles used for illustration are interpreted as boundary points of this domain.} Gerrymandering is possible if this redistricting changes the seat total in favor of some party.\footnote{We use redistricting in an abstract sense. A redistricting profile is any reassignment of fixed voters to districts that preserves voters' party allegiances and the statewide vote total, so it includes changes that could be generated by redrawing district boundaries. The feasible set \(\mathcal P(\bar p)\) does not impose geographic, legal, or administrative constraints. This abstraction isolates how the allocation rule responds to the distribution of partisan support across districts, rather than modeling the mapmaking constraints that determine which redistributions are feasible in a particular state.} Requiring that no such redistricting can increase a party's seat total implies that whenever $\mathbf{q} \in \mathcal P ({\bar p})$, we must have $|R(\mathbf p)|\ge |R(\mathbf q)|$. Applying the same logic after swapping the roles of $\mathbf p$ and $\mathbf q$ yields the reverse inequality, so the seat totals must coincide whenever $\bar p=\bar q$, leading to the following definition: a rule is \emph{gerrymandering-proof} if  $ |R(\mathbf p)|=|R(\mathbf q)|$ for every $\mathbf p$ and $\mathbf q\in\mathcal P(\bar p).$

Note that for any $\lambda<1$, the objective trades off district-level and statewide distortions, so the optimal seat total can depend on how support is distributed across districts even when the statewide mean is fixed. Only at $\lambda=1$ does the seat total depend only on the statewide vote total. This gives the following result.

\begin{proposition}
\label{prop:gp-one}
$R_{\lambda} \in \mathcal R$ satisfies gerrymandering-proofness if and only if $\lambda=1$.
\end{proposition}

\medskip
\noindent \textbf{Incompatibility of Electoral Objectives.}
The endpoint characterizations imply an incompatibility of properties within the optimal family. No two of the following can hold simultaneously: strong monotonicity, gerrymandering-proofness, and a nondegenerate concern for both district-level and statewide representation. At \(\lambda=0\), the rule is first-past-the-post: raising party \(A\)'s support district by district cannot take away a seat it already won, but the seat total can change under mean-preserving redistricting of votes. At \(\lambda=1\), the rule is proportional: the seat total is invariant to such redistrictings, but district outcomes can depend on votes cast elsewhere. For \(\lambda\in(0,1)\), the rule keeps both objectives in the criterion and therefore gives up the global version of both properties. Moving from first-past-the-post to proportionality is therefore not only a change in formulas; it is a movement between two incompatible notions of electoral principles.

\subsection{Gerrymandering Vulnerability}

Gerrymandering-proofness is a binary property: a rule either makes the seat total invariant to mean-preserving redistrictings, or it does not. We now measure vulnerability to gerrymandering more finely by asking how much room a rule leaves for such redistrictings to change the outcome. Formally, we fix party \(A\)'s statewide vote share and consider the redistributions of the same statewide support that can give party \(A\) additional seats. For a baseline profile $\mathbf p$, an integer $k\in\{0,\ldots,N\}$, and  $\lambda\in[0,1]$, define
\begin{equation*}
  \mathcal G_k(\mathbf p;\lambda)
  :=
  \left\{
    \mathbf q\in\mathcal P(\bar p):
    |R_\lambda(\mathbf q)|\ge k
  \right\}.
\end{equation*}
Thus $\mathcal G_k(\mathbf p;\lambda)$ is the set of mean-preserving redistrictings of $\mathbf p$ that give party $A$ at least $k$ seats under $R_{\lambda}$.\footnote{For ease of exposition, in this subsection we resolve ties between optimal seat totals in favor of the larger party-$A$ seat total.}

\medskip \noindent \textbf{Attainable seat gains.} Fix a baseline index \(\lambda_0\in[0,1)\) such that party $A$ is overrepresented at the baseline profile: $  |R_{\lambda_0}(\mathbf p)|>\rho.$ Also fix an integer $  k>|R_{\lambda_0}(\mathbf p)|.$ Thus \(k\) is a seat total strictly above the number of seats party \(A\) already receives at the baseline.  The next result compares, across $\lambda \in [\lambda_0,1]$, the set of mean-preserving redistrictings that can deliver at least \(k\) seats.

\begin{proposition}
\label{prop:gerrymandering-set-inclusion}
For every $\lambda',\lambda$ with $1\ge \lambda' > \lambda\ge \lambda_0$, either
\begin{enumerate}
[label=\textup{(\roman*)},leftmargin=2em]
    \item $\mathcal G_k(\mathbf p;\lambda') = \mathcal G_k(\mathbf p;\lambda) = \emptyset$, or
    \item $\mathcal G_k(\mathbf p;\lambda') \subsetneq \mathcal G_k(\mathbf p;\lambda)$.
\end{enumerate}
\end{proposition}

The proposition states that the attainable set is nested along the optimal family. When party \(A\) is already overrepresented, reaching \(k\) seats moves the delegation still farther from the statewide target. A higher value of \(\lambda\) selects a rule with greater concern for statewide proportionality, so success under a higher-\(\lambda\) rule implies success under every lower-\(\lambda\) rule, while some redistrictings that succeed at lower \(\lambda\) may fail once \(\lambda\) rises, unless the target is already unattainable at the lower $\lambda$. Thus increasing \(\lambda\) makes gerrymandering harder in a set-theoretic sense: it reduces the collection of maps on which the additional seats are attainable.

\begin{figure}[t]
\centering
\begingroup
\definecolor{figureblue}{HTML}{377EB8}
\definecolor{figureorange}{HTML}{E69F00}
\definecolor{figuregreen}{HTML}{009E73}
\definecolor{figurered}{HTML}{E41A1C}
\colorlet{figureaxis}{black!70}
\colorlet{figuregrid}{black!12}

\begin{tikzpicture}[x=14.0cm,y=9.5cm,line cap=round,line join=round]

\foreach \y in {0.50,0.60,0.70,0.80,0.90} {
  \draw[figuregrid,line width=0.35pt] (0,\y) -- (0.65,\y);
}
\draw[->,>=latex,figureaxis,line width=0.60pt]
  (0,0.40) -- (0.67,0.40)
  node[below,text=figureaxis,font=\small] {$q_3$};
\draw[->,>=latex,figureaxis,line width=0.60pt]
  (0,0.40) -- (0,0.95)
  node[left,text=figureaxis,font=\small] {$q_2$};

\foreach \x/\lab in {0.30/0.30,0.50/0.50,0.55/0.55,0.60/0.60} {
  \draw[figureaxis,line width=0.45pt] (\x,0.397) -- (\x,0.403);
  \node[below,text=black!60,font=\scriptsize]
    at (\x,0.397) {$\lab$};
}
\foreach \y/\lab in {
  0.50/0.50,0.60/0.60,0.70/0.70,0.80/0.80,0.90/0.90
} {
  \draw[figureaxis,line width=0.45pt] (-0.003,\y) -- (0.003,\y);
  \node[left,text=black!60,font=\scriptsize]
    at (-0.003,\y) {$\lab$};
}

\coordinate (C) at (0.00,0.86);
\coordinate (D) at (0.00,0.93);
\coordinate (A) at (0.43,0.43);
\coordinate (B) at (0.62,0.62);

\fill[black!5] (C)--(D)--(B)--(A)--cycle;
\draw[black!55,line width=0.80pt] (C)--(D)--(B)--(A)--cycle;
\draw[black!45,densely dashed,line width=0.75pt] (A)--(B);

\fill[figureblue!18] (0.50,0.50)--(0.50,0.68)--(B)--cycle;
\draw[figureblue!80!black,line width=1.05pt]
  (0.50,0.50)--(0.50,0.68)--(B)--cycle;

\fill[figureorange!25] (0.55,0.55)--(0.55,0.655)--(B)--cycle;
\draw[figureorange!88!black,line width=1.05pt]
  (0.55,0.55)--(0.55,0.655)--(B)--cycle;

\fill[figuregreen!38] (0.60,0.60)--(0.60,0.63)--(B)--cycle;
\draw[figuregreen!85!black,line width=1.05pt]
  (0.60,0.60)--(0.60,0.63)--(B)--cycle;

\draw[black!30,densely dotted,line width=0.65pt]
  (0.50,0.40)--(0.50,0.435);
\draw[black!30,densely dotted,line width=0.65pt]
  (0.55,0.40)--(0.55,0.55);
\draw[black!30,densely dotted,line width=0.65pt]
  (0.60,0.40)--(0.60,0.60);
\draw[black!30,densely dotted,line width=0.65pt]
  (0.62,0.40)--(0.62,0.62);

\fill[black!75] (B) circle (1.7pt);
\fill[figurered] (0.30,0.70) circle (2.0pt);
\node[text=figurered,anchor=east,font=\small,inner sep=1pt]
  at (0.292,0.708) {$\mathbf p$};

\node[anchor=south west,align=left,text=black!70,font=\small]
  at (0.015,0.948)
  {redistrictings\\[-1pt]with $\bar q=0.62$};
\node[anchor=north east,align=left,text=black!70,font=\small]
  at (0.655,0.925)
  {$\mathbf p=(0.86,0.70,0.30)$\\[-1pt]
   $|R_0(\mathbf p)|=2$};

\node[anchor=north,text=figureblue!80!black,font=\small]
  (wzero) at (0.500,0.475) {$\phi_\lambda=0$};
\draw[figureblue!80!black,line width=0.65pt]
  (wzero.north) -- (0.500,0.497);

\node[anchor=south,text=figureorange!88!black,font=\small]
  (wten) at (0.515,0.745) {$\phi_\lambda=0.10$};
\draw[figureorange!88!black,line width=0.65pt]
  (wten.south) -- (0.550,0.660);

\node[anchor=west,text=figuregreen!85!black,font=\small]
  (wtwenty) at (0.635,0.555) {$\phi_\lambda=0.20$};
\draw[figuregreen!85!black,line width=0.65pt]
  (wtwenty.west) -- (0.603,0.598);

\node[anchor=west,text=black!65,font=\small]
  (wtwentyfour) at (0.675,0.675) {$\phi_\lambda=0.24$};
\draw[black!55,line width=0.65pt]
  (wtwentyfour.west) -- (0.624,0.623);

\end{tikzpicture}
\endgroup
\caption{\small Gray area denotes the set of redistrictings $\mathbf q$ with mean $\bar q = \bar p =0.62$. On the relative rule-index scale \(\phi_\lambda=\lambda/(1-\lambda)\), Party $A$ gets three seats iff $q_3\ge (1+\phi_\lambda)/2$, denoted by the colored areas at different \(\phi_\lambda\) values. Thus the feasible gerrymandering sets that allocate 3 seats to party $A$, $\mathcal G_3(\mathbf p;\lambda)$, shrink as \(\lambda\) rises, and collapse to $(0.62,0.62,0.62)$ at \(\phi_\lambda=0.24\).}
\label{fig:gerrymandering-nested-sets}
\end{figure}

Figure~\ref{fig:gerrymandering-nested-sets} illustrates the same nesting on the relative rule-index scale \(\phi_\lambda=\lambda/(1-\lambda)\). The gray polygon is the set of sorted mean-preserving redistrictings, while each colored triangle imposes the additional cutoff needed to give party \(A\) all three seats. As \(\phi_\lambda\) rises, the cutoff moves right and the set shrinks; at \(\phi_\lambda=0.24\), only the equalized profile remains.\footnote{A cost version of this result follows by assigning each mean-preserving redistricting a cost based on its distance from the baseline profile. Appendix~\ref{app:gerrymandering-costs} defines this gerrymandering cost and shows that increasing \(\lambda\) weakly raises the minimal cost of obtaining additional seats for an already overrepresented party.} Proposition~\ref{prop:gerrymandering-capacity-bound} gives a bound on $\lambda$ above which this shrinking set is empty: no mean-preserving redistricting
can deliver the target $k$ seats.

\medskip
\noindent \textbf{A capacity bound.} The next result gives an upper bound on how much overrepresentation can be generated by redistributing the statewide support.

\begin{proposition}
\label{prop:gerrymandering-capacity-bound}
Consider a voting profile \(\mathbf p\) with statewide support
\(\rho>0\), and a seat target \(k\) such that
\(k\ge 1+\rho \). If $\lambda>1-\frac{k}{2\rho}$ then  $  \mathcal G_k(\mathbf p;\lambda)=\emptyset .$
\end{proposition}

The intuition is that seat gains must be supported by enough votes in the strongest districts. When party \(A\) is already overrepresented, increasing \(\lambda\) raises the cutoff for awarding it another seat. To make \(k\) seats attainable, the \(k\)\textsuperscript{th} strongest district in some mean-preserving redistricting must clear the cutoff \(\frac{1}{2(1-\lambda)}\).  But then the top \(k\) districts together require at least \(\frac{k}{2(1-\lambda)}\) units of statewide support. Once this exceeds the available support \(\rho\), no redistricting of the same statewide vote can deliver \(k\) seats.

\subsection{Concentration of Support and Majorization Monotonicity}
\label{sec:comparative-statics}

We now study how the rules in \(\mathcal R\) respond when the same statewide support for party \(A\) is distributed more or less unevenly across districts. Holding the statewide vote share fixed, differences across profiles come from the district-loss term. Under a top-\(S\) allocation, the relevant statistic is the cumulative support \(\sum_{d\le S}p_d\) in party \(A\)'s strongest districts, which makes majorization the natural order for comparing profiles.

Formally, given profiles $\mathbf p$ and $\mathbf q$ with the same mean, say that $\mathbf p$ \emph{majorizes} $\mathbf q$ if, after ordering both profiles in decreasing order, $  \sum_{d=1}^k p_d  \ge  \sum_{d=1}^k q_d$ for every $k=1,\ldots,N.$ Thus $\mathbf p$ is more concentrated than $\mathbf q$ if party $A$ has weakly more cumulative support in its strongest districts at every cutoff.

The next lemma shows that misrepresentation satisfies a monotonicity property with respect to majorization for a fixed seat count $S$.

\begin{lemma}
\label{lem:maj-fixedS}
If $\mathbf p$ majorizes $\mathbf q$, then $  M(S,\mathbf p;w)\le M(S,\mathbf q;w)$ for every \(S\in\{0,\ldots,N\}\) and every $w\in[0,1]$.
\end{lemma}
Intuitively, majorization increases Party~$A$'s cumulative support in its strongest districts, so for any fixed $S$, the top-$S$ allocation places seats where fewer voters are represented by their non-preferred party.

We now move from a fixed seat total $S$ to evaluating cutoff rules, which choose the seat total endogenously as a function of the vote-share profile.  Write $S_R(\mathbf r)$ for the seat total selected by a rule $R$ at profile $\mathbf r$. The misrepresentation induced by $R$ at evaluation weight $w$ is $  M(R,\mathbf r;w)   := M(S_R(\mathbf r),\mathbf r;w).$ It turns out that the ordering in Lemma~\ref{lem:maj-fixedS} is preserved when the seat total is not fixed but instead is chosen optimally.

\begin{proposition}
\label{prop:maj-opt}
If $\mathbf p$ majorizes $\mathbf q$, then $  M(R_w,\mathbf p;w)
  \le
  M(R_w,\mathbf q;w)$ for every \(w\in[0,1]\).
\end{proposition}

This conclusion need not extend to rules that keep the district winning threshold fixed. For example, under first-past-the-post the threshold is $\tfrac12$, so the induced seat total can vary across profiles even when the statewide mean $\bar p$ is unchanged. When $w>0$, these discrete seat swings can raise statewide distortion enough to offset the within-district improvement implied by majorization.\footnote{Appendix~\ref{app:fptp-majorization-example} gives a four-district example, with a graphical illustration, in which a mean-preserving majorization lowers district-level misrepresentation under first-past-the-post but increases the first-past-the-post seat total, raising statewide misrepresentation enough to increase total misrepresentation.} This motivates the next question: when do the majorization comparative statics continue to hold once a \emph{rule} chooses the seat total endogenously?

\medskip
\noindent
\textbf{Majorization monotonicity.}
A rule $R$ satisfies \emph{majorization monotonicity at} $w$ if, for all profiles $\mathbf p,\mathbf q$ with the same mean and with $\mathbf p$ majorizing $\mathbf q$,
\begin{equation}\tag{MM}\label{eq:MM}
  M(R,\mathbf p;w)
  \le
  M(R,\mathbf q;w).
\end{equation}

We have already shown majorization monotonicity when the seat total is fixed and when the seat total is chosen optimally for the same weight used to evaluate misrepresentation. Lemma \ref{lem:maj-fixedS} implies the property for proportional representation (that is, $R_{1}$) because majorization comparisons hold statewide support fixed, and proportional representation fixes the party seat total as a function of that support. The next theorem shows that these are the only such cases.

\begin{theorem}
\label{thm:maj-char-weight-family}
$R_{\lambda} \in \mathcal R$ satisfies majorization monotonicity at $w$  if and only if either $\lambda=w$ or $\lambda=1$.
\end{theorem}

We prove Theorem~\ref{thm:maj-char-weight-family} by first showing that, for $\lambda<1$ with $\lambda\ne w$, one can find a more concentrated profile at which $R_\lambda$ becomes indifferent between different seat totals, but that indifference does not hold under $w$. The profiles can be chosen so that $R_\lambda$ switches to a seat total that is strictly worse under $w$ at the more concentrated profile. By choosing the two profiles close enough to make the change in the district component sufficiently small, we obtain a contradiction to \eqref{eq:MM}. Hence, among rules with $\lambda<1$, majorization monotonicity uniquely pins down $\lambda=w$. The result turns a natural requirement on how representation responds to the geography of support into a characterization of the planner's preferred rule.

\section{Measures of Representation and Fairness} \label{sec:eg-comparison}

In this section we use our framework to organize measures of partisan advantage. We focus on three measures---first switch weight, the efficiency gap, and the signed proportional seat gap---and then relate standard disproportionality and counterfactual symmetry measures to the same framework.

\subsection{First Switch Weight}
\label{sec:robustness-measure}

Given a vote-share profile \(\mathbf p\), recall that \(S^\Dist\)  denotes the seat total produced by first-past-the-post. When $S^\Dist < S^\Agg$, Corollary~\ref{cor:step-function} implies that the first switch away from $S^\Dist$ occurs at $w_{S^\Dist+1}(\mathbf p)$, the largest weight on statewide proportionality for which the first-past-the-post seat total $S^\Dist$ remains optimal:
\[
w_{S^\Dist+1}(\mathbf p)
  :=
  \max\{w\in[0,1]: S^\Dist \in \argmin_S M(S,\mathbf p;w)\}.
\]

We refer to this as the first switch weight. It converts the district-statewide tradeoff into a robustness measure.\footnote{For this calculation, party A denotes the party whose
seat total increases in the first corrective step. If necessary,
we interchange party labels and reorder district vote shares
before applying the formula. Thus $S^{\mathrm{Dist}}+1$ refers
to the next seat of the gaining party. If first-past-the-post
already minimizes statewide loss, the first-switch weight is one.} A high value means that the first-past-the-post allocation remains optimal even when the objective puts substantial weight on statewide proportionality. A low value means that the outcome depends on giving almost full priority to district-level representation. In this case, even a small increase in the statewide weight moves the optimum to a different seat total.

\subsection{Efficiency Gap}
The efficiency gap is widely used to measure partisan gerrymandering, and it is often invoked in debates about electoral representativeness more broadly. This prominence is empirically relevant in our data: Section~\ref{subsec:emp-compare} shows that maps struck down or litigated in court are concentrated much more strongly in the upper tail of the efficiency-gap distribution than in the corresponding tails of the alternative measures. 

The efficiency gap is defined using the concept of \emph{wasted votes}. Under first-past-the-post, all votes cast for the losing party in a district are wasted, because they do not elect a representative. For the winning party, votes above the one-half threshold are wasted; these votes are surplus to the minimum needed to win.\footnote{To see the construction, suppose party \(A\) wins a district with vote share \(1/2+x\), so party \(B\)'s vote share is \(1/2-x\). Then party \(B\)'s wasted votes are all of its votes, \(1/2-x\), while party \(A\)'s wasted votes are only its surplus above the winning threshold, \(x\).} Let \(W_A(\mathbf p)\) and \(W_B(\mathbf p)\) denote the total wasted vote shares of parties \(A\) and \(B\) across $N$ districts under the first-past-the-post allocation. The signed efficiency gap for party \(A\) is
\[
  EG_A(\mathbf p)
  :=
  \frac{W_B(\mathbf p)-W_A(\mathbf p)}{N}.
\]

The rationale for the measure is that differences in wasted votes summarize the packing-and-cracking logic of gerrymandering: packing produces surplus votes in lopsided wins, whereas cracking produces votes cast in losing districts. Positive values mean that party \(B\) wastes more votes than party \(A\), so party \(A\) is advantaged.

\begin{proposition}
\label{prop:eg-statewide}
At profile \(\mathbf p\),

\[
  |EG_A(\mathbf p)|
  =
  \frac{1}{N}
  \left|
    S^\Dist-\left(2\rho-\frac{N}{2}\right)
  \right|.
\]
Hence the absolute efficiency gap is the statewide-misrepresentation term with linear cost \(Q(z)=z/N\) and target \(T^{EG}(\rho)=2\rho-N/2\).\end{proposition}

Proposition~\ref{prop:eg-statewide} shows that although the efficiency gap is defined from wasted votes computed district by district, its aggregate content is entirely captured by a statewide target in our framework. Once the first-past-the-post seat total  \(S^\Dist\) and statewide support \(\rho\) are fixed, the measure is the target distance between \(S^\Dist\) and \(T^{EG}(\rho)=2\rho-N/2\), normalized by $N$.\footnote{Over the full domain of $\rho$, diagnostic targets such
as $T^{EG}(\rho)$ may lie outside $[0,N]$. For allocation in the
linear benchmark, we use
$\widetilde T(\rho)=\min\{N,\max\{0,T(\rho)\}\}$, equivalently
bounding the target seat share between zero and one. This leaves
optimal allocations unchanged.}

Replacing the proportional statewide term with the efficiency-gap target leaves the top-$S$ and cutoff logic unchanged. What changes is the endpoint. As the statewide weight rises, the rule moves toward $T^{EG}(\rho)$, rather than toward proportionality.

\medskip
\noindent 
\textbf{Proportionality versus the Efficiency Gap.}
We define party \(A\)'s \emph{proportional seat gap} under first-past-the-post  as the difference between its seat share and statewide vote share:
\[
  \mathrm{PG}_A(\mathbf p)
  :=\frac{S^\Dist-\rho}{N}
\]
Note that its absolute value is proportional to statewide misrepresentation at the first-past-the-post allocation; $N \left|\mathrm{PG}_A(\mathbf p)\right| = \Agg(S^\Dist,\rho)$.

As the efficiency gap evaluates the same first-past-the-post seat total against the target \(T^{EG}(\rho)=2\rho-N/2=\rho+(\rho-N/2)\), the proportional seat gap is its proportional-target analogue. The targets coincide only when \(\rho=N/2\); otherwise, the efficiency-gap target shifts away from proportionality in favor of the statewide winner. Thus a zero efficiency gap generally incorporates a winner's bonus rather than proportional representation.\footnote{See \citet{chambers2017flaws} for an early criticism of the efficiency gap based on its winner's-bonus benchmark, and \citet{stephanopoulos2018measure} for a reply defending its equal-wasted-votes principle and implied responsiveness.} Two examples  make the distinction between efficiency gap and our statewide measure transparent. 

\medskip 
\noindent 
\textbf{Example 1.} Consider the profile $\mathbf p   =   (0.51,0.51,0.51,0.21,0.21,0.21,0.21,0.21,0.21,0.21).$ Here \(\rho=3\), so proportionality gives \(S^\Agg=3\). First-past-the-post also gives \(S^\Dist=3\). Hence first-past-the-post is proportional in seat totals. But $EG_A   =   \frac{3}{10}   -   2\frac{3}{10}   +   \frac12   =   0.2$, a large value relative to the empirical distribution of observed U.S.\ House maps.\footnote{In the empirical distribution shown in Figure~\ref{fig:eg-w1}, an efficiency gap of \(0.2\) lies in the far right tail, well above the $\pm 7\%$ reference lines used in Figure~\ref{fig:eg-w1}. The $7\%$ benchmark was prominent in \cite{whitford2016} as an indicator of unusually large and durable efficiency gaps.} The efficiency gap records an imbalance in wasted votes even though statewide misrepresentation is zero. In fact, the seat total that sets the efficiency gap to zero is \(S_A=1\).

\medskip
\noindent
\textbf{Example 2.} Consider a twelve-district profile in which party $A$ receives $75\%$ of the vote in every district, so $\mathbf p=(0.75,\ldots,0.75)$. Then \(\rho=9\), so proportionality gives \(S^\Agg=9\). Under first-past-the-post, party \(A\) wins every district, so \(S^\Dist=12\). Thus party \(A\) receives all seats with only \(75\%\) of the vote. Yet $EG_A = \frac{12}{12} - 2\frac{9}{12} + \frac12 = 0.$ The efficiency gap classifies this outcome as unbiased because the winner's bonus balances wasted votes, even though the seat allocation is far from proportional, giving $\mathrm{PG}_A=(12-9)/12=0.25$.

\medskip
\noindent
\textbf{Separation across measures.}
The examples illustrate two general separation results proved in
Appendix~\ref{app:seperation}. First, even when first-past-the-post and
proportionality select the same seat total, the absolute efficiency gap can be
arbitrarily close to its largest feasible value conditional on that seat total. Second, \(EG_A=0\) does not imply that the
first-past-the-post outcome is robust: the first-switch weight can be
arbitrarily close to zero even when \(EG_A=0\). This is because the efficiency gap depends only on statewide support
and the first-past-the-post seat total, whereas the switch weights 
depend on district margin. Thus an election can balance wasted
votes yet cease to be optimal after an arbitrarily small weight is placed on
statewide proportionality.

\subsection{Other Measures}
\medskip
\noindent
\textbf{Disproportionality indices.}
The statewide component in Section~\ref{sec:general-model} nests standard two-party disproportionality indices as special normalizations of the proportional target. For example, the \citet{loosemore1971theoretical} index is
\[
  \frac12\left(
    \left|\frac{S}{N}-\frac{\rho}{N}\right|
    +
    \left|\left(1-\frac{S}{N}\right)-\left(1-\frac{\rho}{N}\right)\right|
  \right)
  =
  \left|\frac{S}{N}-\frac{\rho}{N}\right|.
\]
The \citet{gallagher1991proportionality} least-squares index gives the same expression in the two-party case. Thus both are proportional-target versions of our statewide loss, with \(T(\rho)=\rho\) and \(Q(z)=z/N\). If one works with the squared least-squares loss before taking the square root, the corresponding normalization is \(Q(z)=(z/N)^2\).

\medskip 
\noindent \textbf{Counterfactual symmetry measures.}  Partisan bias and seats--votes summaries are also nested in our framework, but they evaluate the statewide component at counterfactual profiles rather than at the observed profile. Partisan bias holds the map fixed and applies a uniform swing until statewide support is tied. For an observed profile \(\mathbf p\), let \(   \mathbf p^{1/2}   :=   \mathbf p+\left(\frac12-\bar p\right)\mathbf 1 \) denote this half-vote counterfactual, whenever it lies in \([0,1]^N\). Partisan bias asks whether first-past-the-post gives party \(A\) more or fewer than half the seats at \(\mathbf p^{1/2}\) \citep{gelman1994unified,grofman2007future}. Thus, using the already-defined first-past-the-post rule \(R_0\), \[   PB_A(\mathbf p)   :=   \frac{|R_0(\mathbf p^{1/2})|}{N}-\frac12 . \] With the proportional statewide loss \(\Agg(S,\rho)=|S-\rho|\), this gives \[   N|PB_A(\mathbf p)|   =   \left||R_0(\mathbf p^{1/2})|-\frac N2\right|   =   \Agg\left(|R_0(\mathbf p^{1/2})|,\frac N2\right). \]

\section{Empirical Analysis}
\label{sec:empirics}

We now bring that framework to U.S.\ House elections. For each state and election year, we observe the district vote-share profile $\mathbf p$. We proceed in three steps. We first analyze first-switch weights in the 2024 election. We then examine where state-years governed by struck or litigated maps fall in the distributions of the three measures. Finally, we compare the measures across map-drawing authorities, including partisan legislatures, courts, and split-control or nonpartisan bodies.

\subsection{Data and construction}
\label{subsec:emp-data}

We use U.S.\ House general-election returns from \citet{DVN/IG0UN2_2017} at the
candidate-district level, covering candidate names, party labels, and vote totals.
We treat Republicans as party~$A$ and Democrats as party~$B$ and construct
district-level two-party vote shares $p_d$ for each state-year. 

Our panel is the set of state-years with at least eight congressional districts over
2000--2024, the period for which the enacted maps are coded from
\citet{allaboutredistricting} and a same-map presidential reference is
available.\footnote{The eight-district restriction keeps attention on delegations
large enough for the seat-conversion tradeoff to be nontrivial; the model assumes
equal-population districts, a close approximation for U.S.\ House seats within a
state, though turnout and eligibility differ across districts.} \cite{allaboutredistricting} also provides map level information for elections from 2002 onward. This includes the date of the map's passing, the time period it was active for, the authority that drew the final map, the political tilt of the map drawing authority and whether the map was superseded, active, struck down by courts or continues to be litigated as of May 2026.\footnote{We use ``struck or litigated'' as a descriptive marker of legal scrutiny. Such scrutiny is plausibly related to representational concerns, since maps that attract legal challenges are unlikely to be a random subset. However, the cases in our sample involve different legal claims, standards, remedies, and procedural postures, so the classification should not be interpreted as a common judicial finding of partisan unfairness or as definitive evidence of a map's representativeness.}

To align the returns with the two-party framework, we construct
district-level vote shares from Democratic and Republican vote totals.
We flag races in which either the top two finishers are not one Democrat
and one Republican, or votes for third-party candidates, write-ins,
and invalid ballots together exceed $15\%$ of all votes cast. For these
races, we proxy the two-party vote share using other congressional
elections in the same district under the same map. When no suitable
congressional proxy is available, we use presidential returns for the
same district. We exclude state-years for which neither source provides
a proxy consistent with the observed winning party. Appendix~\ref{app:construction}
details the coding rules and proxy construction. As a robustness check, we reverse the proxy order for these races, using the presidential proxy first and the same district-map proxy as the secondary source. The main empirical patterns are unchanged under this alternative specification.

\subsection{First Switch Weight in the 2024 Election}
\label{subsec:emp-cross-section}

We now apply the first-switch-weight measure introduced in Section~\ref{sec:robustness-measure} to the 2024 U.S.\ House elections. Figure \ref{fig:2024weights} reports the first switch weight for each state in 2024. Each bar gives the largest weight on statewide proportionality for which the observed first-past-the-post seat total remains optimal. Low values indicate fragile first-past-the-post outcomes: only a small statewide component is needed to move the misrepresentation-minimizing rule to a different seat total.

\begin{figure}[b!]
  \centering
  \includegraphics[width=\linewidth]{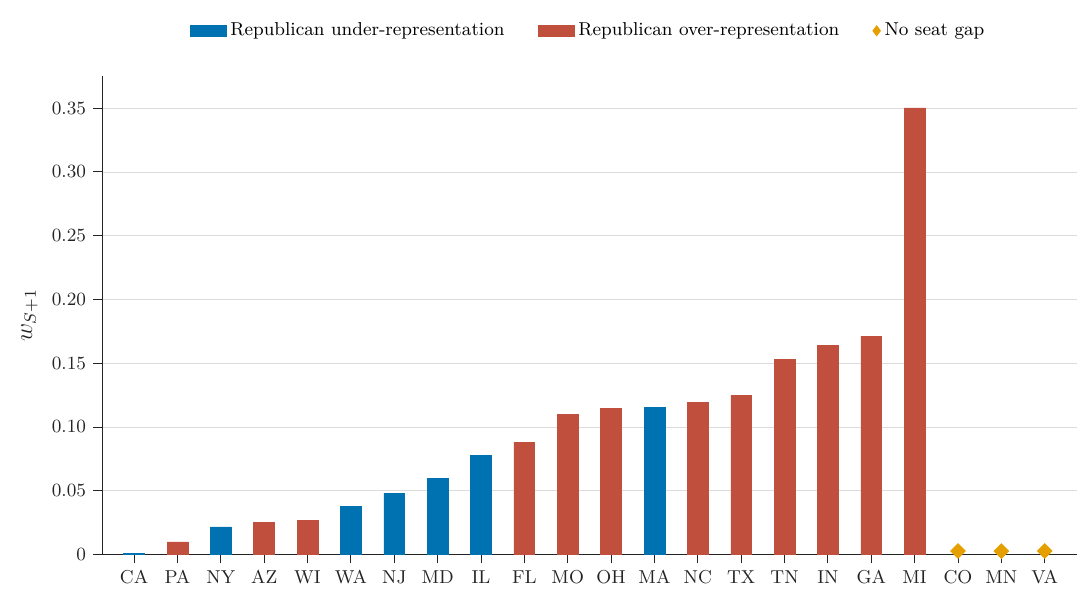}
  \caption{\small First-switch weights by state, 2024. Gray markers denote cases in which first-past-the-post and proportional representation select the same integer seat total, so no corrective seat switch is required.}   
  \label{fig:2024weights}
\end{figure}

The ordering also shows that robustness is distinct from the direction of partisan advantage. Low first switch weights appear in states with Republican under-representation, such as California and New York, and in states with Republican over-representation, such as Pennsylvania, Arizona, and Wisconsin. The common feature is the pivotal seat: it is the narrowest win of the already overrepresented party, and the closer that win is to the cutoff, the less statewide weight is needed to displace the first-past-the-post allocation. Appendix~\ref{app:proportional-switch-weights} reports the weights at which the proportional seat total becomes optimal.

\paragraph{Texas and California.} The contrast most relevant to current redistricting debate is Texas versus California, the two largest delegations and the focus of renewed attention to the interaction between maps and rules.\footnote{In 2025 Texas adopted a new congressional map widely described as engineered to add several Republican-leaning seats; a three-judge panel initially enjoined it, but the Supreme Court stayed the injunction, leaving the map eligible for 2026 while litigation proceeds \citep{kruzel2025scotusTexas}. California responded with Proposition~50, framed as a counter to Texas's mid-cycle redistricting, which the Supreme Court declined to block for 2026 \citep{chung2026scotusCalifornia}. The 2024 cross-section in Figure \ref{fig:2024weights} therefore predates both new maps and is the natural baseline against which to read their effects.} In 2024, Republicans are under-represented in California and over-represented in Texas. California's first switch weight is only a fraction of Texas's, \(0.0009\) versus \(0.1246\), so the California seat allocation stops being optimal once the objective places even a very small weight on statewide proportionality. Figure~\ref{fig:california-texas-switches} shows that this is not only a first-seat difference: California would switch ten seats to Republicans before Texas would switch its first seat to Democrats. The Texas allocation is therefore more robust over the relevant range, while California's seat total moves quickly once statewide proportionality enters the objective. This makes the Texas--California pair a natural baseline for tracking how recent developments reshape the underlying vote profiles in 2026.

\begin{figure}[h!]
  \centering
  \includegraphics[width=\linewidth]{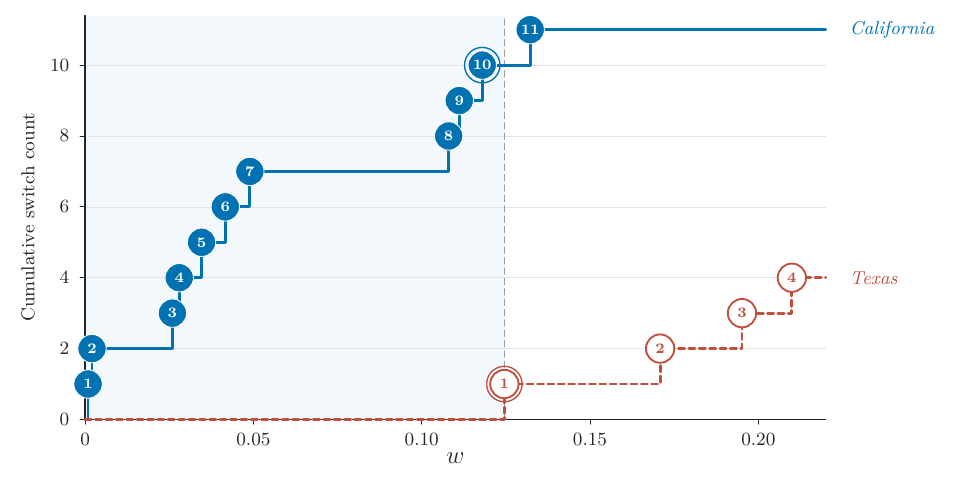}
  \caption{\small Cumulative seat switches in California and Texas in 2024 as the weight
  on proportionality increases. California reaches its tenth switch at
  $w_{\mathrm{CA},10}=0.1180$, before Texas reaches its first at
  $w_{\mathrm{TX},1}=0.1246$.}
  \label{fig:california-texas-switches}
\end{figure}

\subsection{Struck or Litigated Maps Across Representation Measures}
\label{subsec:emp-compare}

We now compare the efficiency gap, the first switch weight, and the proportional seat gap on litigated and struck maps. For these comparisons, we use the inverse first-switch weight, $1-w_{S^\Dist+1}$, so that larger values indicate less robust first-past-the-post outcomes. We also assign a sign to this measure: positive for Republican over-representation and negative for Republican under-representation.

Figure \ref{fig:litigated-percentiles} reports the percentile ranks of maps struck or litigated in court across the three measures. These maps are much more concentrated in the upper tail of the efficiency-gap distribution than in the corresponding tails of the other measures. Twelve state-years exceed the 90th percentile of the absolute efficiency-gap distribution, and eight exceed the 95th percentile. In contrast, six exceed the 90th percentile of the proportional seat gap, and none exceed the 90th percentile of the inverse first-switch weight. The concentration is particularly notable given the efficiency gap's prominence in legal and public debates over partisan gerrymandering.

\begin{figure}[b!]
  \centering
    \includegraphics[width=\textwidth]{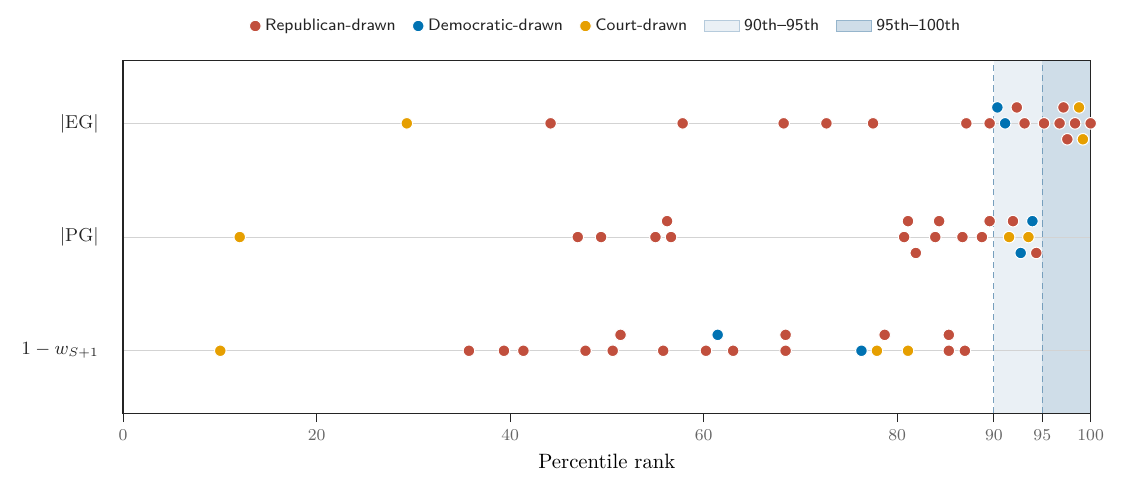}
\caption{\small Percentile ranks of struck or litigated congressional maps across three measures: the absolute efficiency gap, the absolute proportional seat gap, and the inverse first-switch weight. Each point represents a state-year, ranked relative to the full empirical sample. Colors identify the map-drawing authority, and shaded regions highlight the 90th–95th and 95th–100th percentiles. Higher ranks indicate more extreme values under the corresponding measure.}  \label{fig:litigated-percentiles}
\end{figure}

This contrast illustrates the importance of the representational benchmark used to evaluate maps. The efficiency gap evaluates departures from a winner's-bonus target rather than proportionality, while the first switch weight captures the robustness of the observed seat allocation. Consequently, these measures can identify different maps as extreme. Appendix~\ref{app:statewide-targets} discusses courts' use of the efficiency gap alongside other measures of partisan fairness, and Appendix~\ref{app:how-measures-differ} reports percentile ranks and pairwise plots across the full sample.

\subsection{Misrepresentation by map-drawing authority}
\label{subsec:emp-drawer}

The preceding subsection showed that the efficiency gap is more closely associated with struck or litigated maps than either the proportional seat gap or the first switch weight. We now ask whether the three measures also differ systematically by the institution that draws the map. This comparison helps clarify whether the efficiency gap is more sensitive to some map-drawing authorities than the other measures.

Figures~\ref{fig:drawer-leg} and~\ref{fig:drawer-npcourt} group state-years by the authority responsible for the congressional map: Democratic legislature, Republican legislature, split-control or nonpartisan authority, and court. For each group, the figures plot the distribution of three signed measures. Positive values are Republican-favorable throughout.

\begin{figure}[b!]
  \centering
  \includegraphics[width=\linewidth]{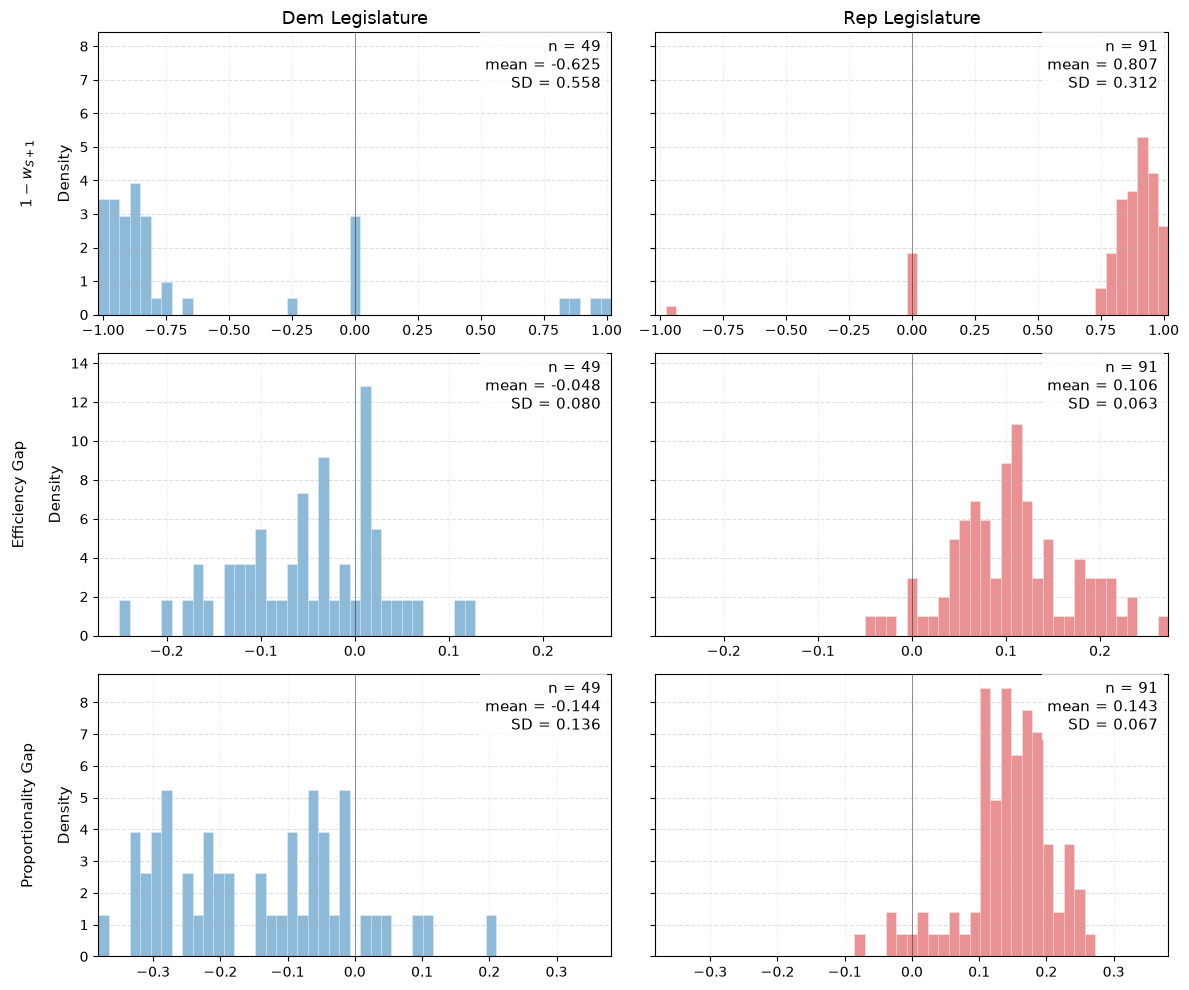}
  \caption{\small Distribution of signed inverse first switch weight, the efficiency gap, and the proportional seat gap by partisan legislatures. Positive values are Republican-favorable throughout. The proportional seat gap and efficiency gap are expressed in percentage points.}
  \label{fig:drawer-leg}
\end{figure}

\begin{figure}[b!]
  \centering
  \includegraphics[width=\linewidth]{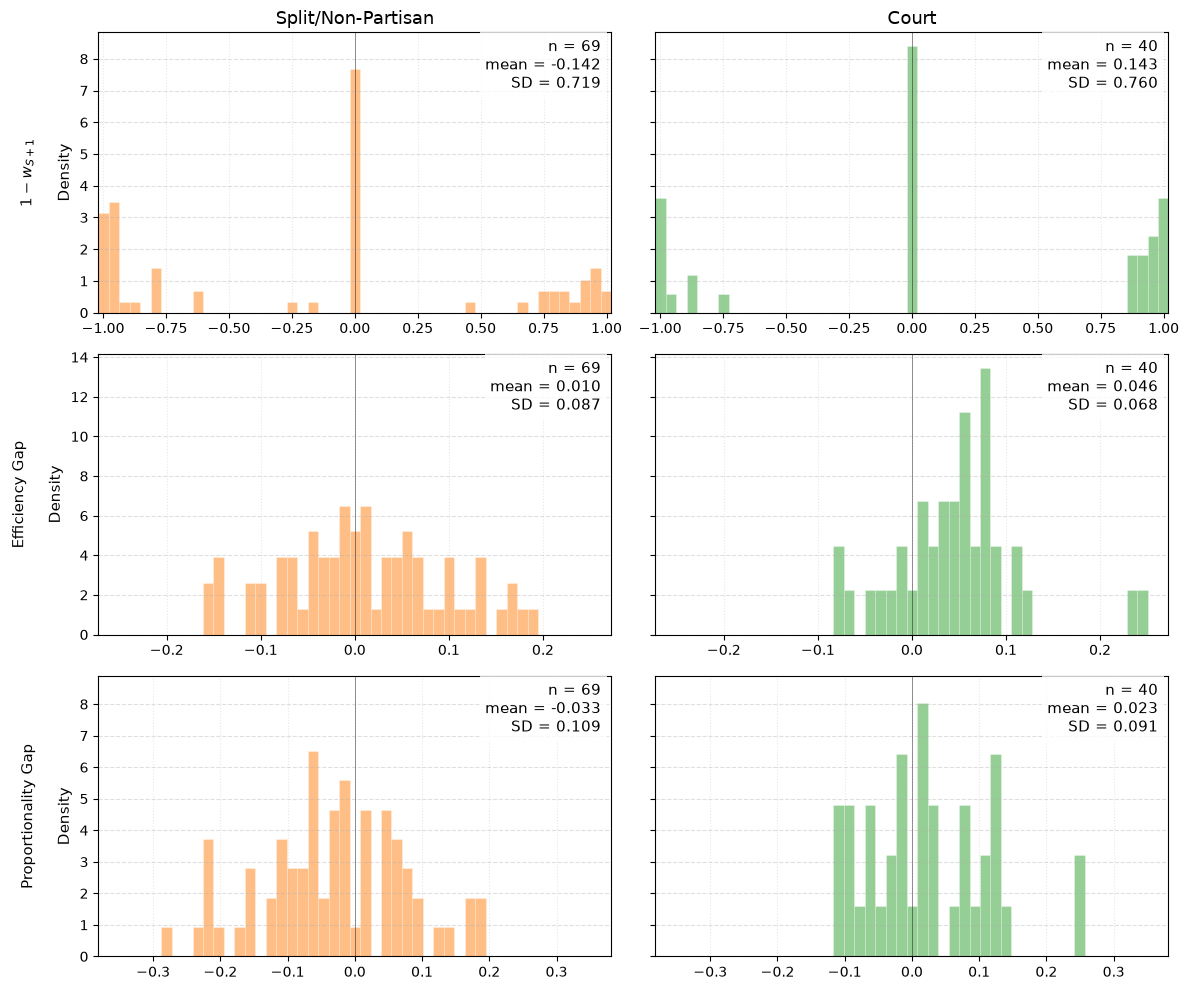}
  \caption{\small Distribution of signed inverse first switch weight, the efficiency gap, and the proportional seat gap by split/nonpartisan authorities and courts. Positive values are Republican-favorable throughout. The proportional seat gap and efficiency gap are expressed in percentage points.}
  \label{fig:drawer-npcourt}
\end{figure}

Three patterns stand out. First, comparing Figures~\ref{fig:drawer-leg} and~\ref{fig:drawer-npcourt} shows that maps enacted by single-party legislatures generally exhibit larger net partisan advantages than maps drawn under split control, by nonpartisan authorities, or by courts. Table~\ref{tab:drawer-summary} reports the group means. Across all three measures, court-drawn maps also have more Republican-favorable mean scores than maps drawn under split control or by nonpartisan authorities.

\begin{table}[ht!]
\centering
\small
\setlength{\tabcolsep}{5pt}
\begin{tabular}{@{}lrrrr@{}}
\toprule
Map-drawing authority & $n$
& \shortstack{Inverse first\\switch weight}
& \shortstack{Proportional\\seat gap (pp)}
& \shortstack{Efficiency gap\\(pp)} \\
\midrule
Democratic legislature & $49$ & $-0.625$ & $-14.42$ & $-4.79$ \\
Republican legislature & $91$ & $0.807$  & $14.27$  & $10.59$ \\
Split/nonpartisan & 69 & $-0.142$ & $-3.26$ & $0.96$ \\
Court             & 40 & $0.144$  & $2.25$  & $4.59$ \\
\bottomrule
\end{tabular}
\caption{\small Mean representation measures by map-drawing authority. Positive values are Republican-favorable. The proportional seat gap and efficiency gap are expressed in percentage points (pp). Each observation is a state-year.}
\label{tab:drawer-summary}
\end{table}

 Second, among maps drawn by partisan legislatures,
Republican-legislature maps show larger average party-favorable
deviations than Democratic-legislature maps under the inverse
first-switch weight and the efficiency gap. Expressing each group's scores in favor of the party responsible for drawing the map, the Republican-legislature mean is about \(29\%\) larger for inverse first switch weight, and about \(121\%\) larger for the efficiency gap, while the two groups are nearly identical on the proportionality gap measure.

These ratios point to a third pattern. The efficiency gap amplifies the Republican--Democratic asymmetry among partisan legislatures. Under the proportional seat gap, the the Republican-legislature and Democratic-legislature maps are nearly identical but the relative difference is substantially greater under the efficiency gap. The winner's-bonus benchmark accounts for much of the increase in this asymmetry when moving from the proportional seat gap to the efficiency gap. Writing \(\bar p=\rho/N\) for the statewide Republican vote share, the Republican-favorable efficiency gap can be written as
\begin{equation}
  EG_A
  = PG_A
    + \left(\frac12-\bar p\right).
  \label{eq:drawer-eg-decomposition}
\end{equation}
The first term is the proportional seat gap, and the second is the adjustment associated with the winner's-bonus benchmark.

\begin{table}[b!]
\centering
\small
\setlength{\tabcolsep}{5pt}
\begin{tabular}{@{}lrrr@{}}
\toprule
Map-drawing authority
& \shortstack{Proportional\\seat gap}
& \shortstack{Benchmark\\adjustment}
& \shortstack{Efficiency\\gap} \\
\midrule
Democratic legislature & $14.42$ & $-9.63$ & $4.79$ \\
Republican legislature & $14.27$ & $-3.68$ & $10.59$ \\
\midrule
Republican minus Democratic & $-0.15$ & $5.95$ & $5.80$ \\
\bottomrule
\end{tabular}
\caption{\small Decomposition of mean advantage for the party responsible for drawing the map. All entries are in percentage points. Each component retains its original sign for Republican-legislature observations and reverses sign for Democratic-legislature observations. Differences are calculated before rounding, so displayed entries may not sum exactly.}
\label{tab:drawer-eg-decomposition}
\end{table}

Table~\ref{tab:drawer-eg-decomposition} reports this decomposition in terms of advantage for the party responsible for drawing the map. The winner's-bonus adjustment reduces average measured Democratic advantage from 14.42 to 4.79 percentage points, and average measured Republican advantage from 14.27 to 10.59 percentage points. Since the two are nearly identical on proportionality gap, this difference in the winner's-bonus fully accounts for the difference in the average party-favorable efficiency gap.

\section{Discussion and Related Literature} 

In single member district elections,  electoral rules must  balance representing local electorates and reflecting statewide political support. We characterize the rules that optimally balance these objectives. The optimal rules form a family ranging from first-past-the-post to proportional representation. Along this path, giving statewide proportionality greater weight makes gerrymandering harder. 

Our framework yields a new measure of electoral representation based on the robustness of first-past-the-post to statewide representation concerns and also accommodates established electoral diagnostics. Using U.S. House elections from 2002 to 2024, we compare electoral outcomes across mapmaking institutions. The widely used efficiency gap is more closely associated with struck or litigated maps than alternative measures, but it may amplify the Republican–Democratic asymmetry. 

We discuss how these results relate to earlier work on electoral institutions, the measurement of representation, and partisan redistricting.

\smallskip
\noindent
\textbf{Electoral institutions and representation.} Our paper contributes to the literature on how electoral institutions shape party competition, political selection, public spending, accountability, and polarization \citep{cox1997making,milesiferretti2002electoral,persson2005economic}. \citet{carey2011electoral} examine how electoral systems combine proportional representation with accountability and local linkage, identifying a favorable region of an empirically estimated representation--accountability frontier. \citet{gordon2025harm} focus on dyadic representation, defined through voters' relationships with their district representatives, and use this concept to evaluate representational disparities across districting plans. We complement this literature by studying rule design within a fixed single-member district system and characterizing the allocations that optimally balance local and statewide representation.

\smallskip
\noindent
\textbf{Seat allocation across parties and districts.}
Fair-majority and biproportional methods determine aggregate party seat totals and then allocate them across districts subject to party and territorial constraints \citep{balinski2008fair,gaffkepukelsheim2008divisor}. \citet{antweiler2019electoral} is especially close to our conditional assignment problem: given proportional party totals, he allocates seats across constituencies to maximize local representation. Our framework chooses party totals jointly with district assignments to balance local and statewide representation. \citet{jeong2024cost} studies party-list compensation to reconcile local accountability with proportionality, a complementary design margin to our allocation of a fixed set of district seats.

\smallskip
\noindent
\textbf{Measures of electoral representation.}
The paper also contributes to the measurement of electoral representation. \citet{katz2020theoretical} emphasize that evaluating a partisan-fairness measure requires specifying the representational standard it is intended to capture. Disproportionality indices summarize departures of seat shares from vote shares \citep{loosemore1971theoretical,gallagher1991proportionality}, while partisan-bias and seats--votes measures evaluate partisan symmetry, often at counterfactual vote profiles \citep{king1987democratic,gelman1994unified}.  We map these measures to our framework and characterize their implied objectives, clarifying why they can evaluate the same election differently.

\smallskip
\noindent
\textbf{The efficiency gap and partisan fairness.} Building on the relative-wasted-votes approach of \citet{mcghee2014measuring}, \citet{stephanopoulos2015partisan} propose measuring partisan gerrymandering by comparing parties' statewide totals of losing and surplus winning votes. Under equal turnout, the measure is zero when seat share equals twice vote share minus one half. \citet{chambers2017flaws} highlight this winner's-bonus target and show that recalibrating the implicit value of a seat yields the proportional vote--seat gap. Our framework interprets these measures as evaluating the same first-past-the-post seat total against different statewide targets. \citet{warrington2018quantifying,warrington2019comparison} develop and compare measures that use different features of the district vote distribution, illustrating why alternative diagnostics can evaluate the same election differently.

\smallskip
\noindent
\textbf{Redistricting institutions and electoral outcomes.} Using earlier statewide elections projected onto enacted and simulated 2021--22 congressional maps, \citet{barton2024decomposition} decompose partisan advantage into contributions from fairness ideals, political geography, legal constraints, and mapmaker discretion. Their finding of larger discretionary advantages in partisan-controlled maps echoes our finding that partisan-legislature maps generally fare worse across our fairness measures in realized House elections from 2002 to 2024. Relatedly, \citet{gaynor2026new} find that partisan institutions, particularly Republican-controlled legislatures, shift vulnerable constituencies more than commissions or institutions under split control. \citet{duchin2022redistricting} study whether district design can achieve near-proportional outcomes while retaining plurality elections in single-member districts. Using alternative maps and historical returns, they find that maps meeting a near-proportionality test across multiple elections are available in most states they study, with feasibility depending on political geography.

\smallskip
\noindent
\textbf{Political geography and partisan redistricting.}
Existing models study how geography and strategic district boundaries generate vote--seat distortions under a fixed district-winner rule \citep{coateknight2007socially,friedman2008optimal,chen2013unintentional,kolotilinwolitzky2023economics,bouton2023packcrackpack}. Our approach is complementary: we hold the map and induced vote profile fixed and study the design of the allocation rule applied to that distribution.

\smallskip
\noindent
\textbf{Axiomatic electoral design.} Social-choice and apportionment theories characterize rules through properties such as responsiveness, neutrality, and consistency \citep{may1952,balinski1982fair}. We instead begin with an explicit objective and then characterize the endpoint rules within the resulting optimal family.

\singlespacing
\bibliographystyle{jfe}
\bibliography{bib-matching}
\appendix

\onehalfspacing
\section{Proofs}
\label{app:proofs}

This appendix collects the proofs of the main results. Throughout, the districts are ordered so that
\(p_1\ge p_2\ge\cdots\ge p_N\). Whenever a proof considers another profile,
its districts are also ordered in decreasing party-\(A\) vote share.

\begin{proof}[\textbf{Proof of Lemma~\ref{lem:conditional-district-assignment}.}]
Fix any allocation $\mathbf x$ with $S(\mathbf x)=S$. If $\mathbf x$ is not a top-$S$ allocation, then there exist districts $i<j$ such that $x_i=0$, $x_j=1$, and $p_i\ge p_j$. Let $\mathbf x'$ be the allocation obtained by swapping the party-$A$ assignment from district $j$ to district $i$: $x_i'=1$, $x_j'=0$, and $x_d'=x_d$ for all $d\notin\{i,j\}$. The seat total is unchanged but the change in district loss from this swap is
\begin{equation*}
\begin{aligned}
  \Dist(\mathbf x',\mathbf p)-\Dist(\mathbf x,\mathbf p)
  &=
  \delta(1,p_i)+\delta(0,p_j)
  -\delta(0,p_i)-\delta(1,p_j)\\
  &=
  \bigl[\delta(1,p_i)-\delta(0,p_i)\bigr]
  -
  \bigl[\delta(1,p_j)-\delta(0,p_j)\bigr].
\end{aligned}
\end{equation*}
By single crossing, this expression is weakly negative. Thus the swap does not increase district loss. Repeating this exchange finitely many times removes all inversions and yields the top-$S$ allocation with weakly lower district loss.
\end{proof}

\begin{proof}[\textbf{Proof of Lemma~\ref{lem:general-reduction}.}]
Fix an allocation $\mathbf x$ and let $S=S(\mathbf x)$. The statewide term depends only on the seat total, so \(\Agg(S(\mathbf x),\rho)=\Agg(S,\rho)\). By Lemma~\ref{lem:conditional-district-assignment}, the top-\(S\) allocation weakly lowers district loss: \(\Dist(S,\mathbf p)\le \Dist(\mathbf x,\mathbf p)\). Since \(F\) is increasing,
\[
  M(S,\mathbf p;w)\le M(\mathbf x,\mathbf p;w).
\]
Thus every allocation is weakly dominated by the top-\(S\) allocation with the same seat total. Since every top-\(S\) allocation is feasible, every optimal seat total has an optimal top-\(S\) implementation.

It remains to prove the optimal seat total range. Suppose \(S^\Dist\le S^\Agg\). If
\(S<S^\Dist\), then \(\Dist(S^\Dist,\mathbf p)\le \Dist(S,\mathbf p)\) because
\(S^\Dist\) is a district-loss minimizer. Also, since
\(S<S^\Dist\le S^\Agg\), moving from \(S\) to \(S^\Dist\) moves weakly toward
the statewide optimum, so \(\Agg(S^\Dist,\rho)\le \Agg(S,\rho)\). Monotonicity of \(F\) implies
\[
  M(S^\Dist,\mathbf p;w)\le M(S,\mathbf p;w).
\]

Now suppose \(S>S^\Agg\). Since \(S^\Agg\) minimizes statewide loss,
\(\Agg(S^\Agg,\rho)\le \Agg(S,\rho)\). Moreover, \(S^\Agg\ge S^\Dist\), and
district loss is weakly nondecreasing for seat totals weakly above
\(S^\Dist\). Hence \(\Dist(S^\Agg,\mathbf p)\le \Dist(S,\mathbf p)\). Monotonicity of \(F\) implies
\[
  M(S^\Agg,\mathbf p;w)\le M(S,\mathbf p;w).
\]
The case $S^\Dist > S^\Agg$ follows by the same argument with the left and right endpoints interchanged. Therefore an optimal seat total is attained between \(S^\Dist\) and \(S^\Agg\).
\end{proof}

Before proving Theorems \ref{thm:general-cutoff} and \ref{thm:optimal-cutoff-rule}, we first provide a useful lemma.

\begin{lemma}
\label{lem:kappa-single-crossing}
Suppose \(S^\Dist<S^\Agg\). For every \(w\in(0,1)\) and every
\(S\in\{S^\Dist,\ldots,S^\Agg-1\}\):

\begin{enumerate}[label=\textup{(\roman*)},leftmargin=2em]
    \item  The condition below holds:
\begin{equation*}
M(S+1, \mathbf{p} ; w) \leq M(S, \mathbf{p} ; w) \Longleftrightarrow \MRT(S, \mathbf{p} ; w) \leq \MRS(S, \mathbf{p} ; w)  \tag{1};
\end{equation*}

\item \(\MRT(S,\mathbf p;w)\) is weakly increasing in \(S\);

\item If \(F\) has diminishing MRS, then \(\MRS(S,\mathbf p;w)\) is weakly decreasing in \(S\). Hence the first order conditions stated in \eqref{eq:kappa-foc} are single crossing in \(S\), and the minimizing seat totals in
\(\{S^\Dist,\ldots,S^\Agg\}\) form an interval of consecutive integers.

\item If \(F\) has diminishing MRS, then for each fixed
\(S\in\{S^\Dist,\ldots,S^\Agg-1\}\), the first order condition is single
crossing in \(w\). That is, for any \(0\le w\le w'\le1\),
\[
  M(S+1,\mathbf p;w)-M(S,\mathbf p;w)\le0
  \quad\Longrightarrow\quad
  M(S+1,\mathbf p;w')-M(S,\mathbf p;w')\le0.
\]
Moreover, if both weights are interior and the move strictly raises
district loss, then the move becomes strictly improving at every higher weight.
\end{enumerate}

\end{lemma}

\begin{proof}[\textbf{Proof of Lemma \ref{lem:kappa-single-crossing}.}] 
\smallskip
\noindent
\underline{\textit{Part (i)}.}
By construction, $  z_S(0)
  =
  \bigl((1-w)\Dist(S,\mathbf p),w\Agg(S,\rho)\bigr),$ and  $  z_S(1)
  =
  \bigl((1-w)\Dist(S+1,\mathbf p),w\Agg(S+1,\rho)\bigr).$
The derivative of the affine path is $  z_S'(\nu)
  =
  \bigl((1-w)\Delta\Dist(S,\mathbf p),-w\Delta\Agg(S,\rho)\bigr).$
Integrating the derivative of \(F(z_S(\nu))\) along this path gives
\[
\begin{aligned}
  M(S+1,\mathbf p;w)-M(S,\mathbf p;w)
  &=
  F(z_S(1))-F(z_S(0))\\
  &=
  \int_0^1
  \left[
    F_1(z_S(\nu))(1-w)\Delta\Dist(S,\mathbf p)
    -
    F_2(z_S(\nu))w\Delta\Agg(S,\rho)
  \right]d\nu\\
  &=
  w\Delta\Agg(S,\rho)
  \left[
    \MRT(S,\mathbf p;w)\int_0^1F_1(z_S(\nu))\,d\nu
    -
    \int_0^1F_2(z_S(\nu))\,d\nu
  \right].
\end{aligned}
\]
Since \(w\Delta\Agg(S,\rho)>0\) and \(F_1>0\), $  M(S+1,\mathbf p;w)\le M(S,\mathbf p;w)$ if and only if
\[
  \MRT(S,\mathbf p;w)
  \le
  \frac{\int_0^1F_2(z_S(\nu))\,d\nu}
       {\int_0^1F_1(z_S(\nu))\,d\nu}
  =
  \MRS(S,\mathbf p;w).
\]
This proves (i).

\smallskip
\noindent
\underline{\textit{Part (ii)}.} On the stated range, \(S\mapsto\Delta\Dist(S,\mathbf p)\) is weakly
increasing and nonnegative, while \(S\mapsto\Delta\Agg(S,\rho)\) is positive and
weakly decreasing. Therefore $  S\mapsto
  \MRT(S,\mathbf p;w)
  =
  \frac{\Delta\Dist(S,\mathbf p)}
       {\phi\Delta\Agg(S,\rho)}$ is weakly increasing. This proves (ii).

\smallskip
\noindent
\underline{\textit{Part (iii)}.}  Now suppose \(F\) has diminishing MRS. Fix
\(S\in\{S^\Dist,\ldots,S^\Agg-2\}\) and take any \(\nu,\mu\in[0,1]\). Every
point on the segment starting from \(S+1\) has weakly larger district loss
term and weakly smaller statewide loss term than every point on the segment starting from \(S\). Diminishing MRS therefore implies $  \frac{F_2(z_{S+1}(\mu))}{F_1(z_{S+1}(\mu))}  \le  \frac{F_2(z_S(\nu))}{F_1(z_S(\nu))}.$

Writing,
\[
  \MRS(S,\mathbf p;w)
  =
  \frac{
    \int_0^1
    \frac{F_2(z_S(\nu))}{F_1(z_S(\nu))}
    F_1(z_S(\nu))\,d\nu
  }{
    \int_0^1F_1(z_S(\nu))\,d\nu
  },
\]
it follows that \(\MRS(S,\mathbf p;w)\) is an \(F_1\)-weighted average of the local MRS values on the segment from \(S\) to \(S+1\), since \(F_1>0\). The pointwise comparison across adjacent segments, $  \frac{F_2(z_{S+1}(\mu))}{F_1(z_{S+1}(\mu))}  \le  \frac{F_2(z_S(\nu))}{F_1(z_S(\nu))}$, therefore gives $  \MRS(S+1,\mathbf p;w)
  \le
  \MRS(S,\mathbf p;w).$

Combining this with (ii), we conclude that \(S\mapsto\MRT(S,\mathbf p;w)-\MRS(S,\mathbf p;w)\) is weakly
increasing. By the first order condition \eqref{eq:kappa-foc}, the move \(S\to S+1\) is weakly
improving  when this difference is weakly negative. Hence once an
incremental move fails to be weakly improving, every later incremental move also fails to be weakly improving. The reduced objective therefore decreases, possibly has a flat segment, and then increases. Its minimizers form an interval of consecutive integers. This proves (iii).

\smallskip
\noindent
\underline{\textit{Part (iv)}.}  Fix \(S\in\{S^\Dist,\ldots,S^\Agg-1\}\). On this range, $  \Dist(S+1,\mathbf p)\ge \Dist(S,\mathbf p),$ and $  \Agg(S+1,\rho)<\Agg(S,\rho).$

If \(\Dist(S+1,\mathbf p)=\Dist(S,\mathbf p)\), then for every \(w>0\),
strict monotonicity of \(F\) in its second argument gives $  M(S+1,\mathbf p;w)<M(S,\mathbf p;w),$ while at \(w=0\) the two values are equal. Hence the first order condition is
single crossing in \(w\) in this case.

Now suppose \(\Dist(S+1,\mathbf p)>\Dist(S,\mathbf p)\). Take
\(0<w<w'<1\) and assume
\[
  M(S+1,\mathbf p;w)\le M(S,\mathbf p;w).
\]
For \(y_1\in[\Dist(S,\mathbf p),\Dist(S+1,\mathbf p)]\), define the level set \(y_2(y_1)\)
by
\begin{equation} \label{eq:level-set}
      F((1-w)y_1,wy_2(y_1))
  =
  F\bigl((1-w)\Dist(S,\mathbf p),w\Agg(S,\rho)\bigr). 
\end{equation}
We first show that $y_2(y_1)$ is well defined for each such \(y_1\). At
\(y_2=\Agg(S+1,\rho)\), the left-hand side is weakly below the right-hand side
of the defining equation: this follows from \(y_1\le \Dist(S+1,\mathbf p)\),
monotonicity of \(F\) in its first argument, and the premise
\(M(S+1,\mathbf p;w)\le M(S,\mathbf p;w)\). At \(y_2=\Agg(S,\rho)\), the
left-hand side is weakly above the right-hand side because
\(y_1\ge\Dist(S,\mathbf p)\). Continuity and strict monotonicity of \(F\) in its
second argument therefore give a unique solution
\(y_2(y_1)\in[\Agg(S+1,\rho),\Agg(S,\rho)]\).

Differentiating equation \eqref{eq:level-set} gives
\[
  y_2'(y_1)
  =
  -\frac{1-w}{w}
  \frac{F_1((1-w)y_1,wy_2(y_1))}
       {F_2((1-w)y_1,wy_2(y_1))}.
\]
Now evaluate the same level set at the higher weight \(w'\). Its derivative with
respect to \(y_1\) is
\[
\begin{aligned}
&\frac{d}{dy_1}F((1-w')y_1,w'y_2(y_1))\\
&=
F_1((1-w')y_1,w'y_2(y_1))                                      \\
&\quad\times
\left[
(1-w')
-
w'\frac{1-w}{w}
\frac{F_1((1-w)y_1,wy_2(y_1))}
     {F_2((1-w)y_1,wy_2(y_1))}
\frac{F_2((1-w')y_1,w'y_2(y_1))}
     {F_1((1-w')y_1,w'y_2(y_1))}
\right].
\end{aligned}
\]
Because \(w'>w\), moving from weight \(w\) to weight \(w'\) lowers the scaled
district loss coordinate from \((1-w)y_1\) to \((1-w')y_1\) and raises the scaled statewide loss coordinate from \(w y_2(y_1)\) to \(w'y_2(y_1)\). By diminishing
MRS, this change weakly raises the local ratio \(F_2/F_1\). Therefore
\[
  \frac{F_2((1-w')y_1,w'y_2(y_1))}
       {F_1((1-w')y_1,w'y_2(y_1))}
  \ge
  \frac{F_2((1-w)y_1,wy_2(y_1))}
       {F_1((1-w)y_1,wy_2(y_1))}.
\]
Together with \( \frac{w'}{1-w'} > \frac{w}{1-w}\), this makes the bracketed term strictly negative, and we conclude that $  \frac{d}{dy_1}F((1-w')y_1,w'y_2(y_1))<0$.  This implies that
\[
  F((1-w')\Dist(S+1,\mathbf p),w'y_2(\Dist(S+1,\mathbf p)))
  <
  F((1-w')\Dist(S,\mathbf p),w'y_2(\Dist(S,\mathbf p))).
\]
By definition in \eqref{eq:level-set}, \(y_2(\Dist(S,\mathbf p))=\Agg(S,\rho)\). So, the right-hand side is
\(M(S,\mathbf p;w')\). Since
\(y_2(\Dist(S+1,\mathbf p))\ge \Agg(S+1,\rho)\), monotonicity of \(F\) gives
\[
  M(S+1,\mathbf p;w')
  \le
  F((1-w')\Dist(S+1,\mathbf p),w'y_2(\Dist(S+1,\mathbf p)))
  <
  M(S,\mathbf p;w').
\]
Thus the move from \(S\) to \(S+1\) is strictly improving at \(w'\).

We finally consider the endpoints. The endpoint \(w'=1\) follows directly from
\[
  M(S+1,\mathbf p;1)=F(0,\Agg(S+1,\rho))
  <
  F(0,\Agg(S,\rho))=M(S,\mathbf p;1).
\]
Finally, if \(w=0\) and \(\Dist(S+1,\mathbf p)>\Dist(S,\mathbf p)\), then
\(M(S+1,\mathbf p;0)>M(S,\mathbf p;0)\), so the premise of the implication
cannot hold. This proves \textup{(iv)}.
\end{proof}

\begin{proof}[\textbf{Proof of Theorem~\ref{thm:general-cutoff}.}]
Take \(w\in(0,1)\) and
\(S\in\{S^\Dist,\ldots,S^\Agg-1\}\). By Lemma~\ref{lem:kappa-single-crossing},
\[
  M(S+1,\mathbf p;w)\le M(S,\mathbf p;w)
  \iff
  \MRT(S,\mathbf p;w)\le\MRS(S,\mathbf p;w).
\]
Substituting the definition of \(\MRT(S,\mathbf p;w)\), the RHS becomes $  \Delta\Dist(S,\mathbf p)
  \le
  \phi\cdot\Delta\Agg(S,\rho)\cdot\MRS(S,\mathbf p;w).$  Also substituting $  \Delta\Dist(S,\mathbf p)=\Delta(p_{S+1}),$ gives the condition
\[
   M(S+1,\mathbf p;w)\le M(S,\mathbf p;w)
  \iff  \Delta(p_{S+1})
  \le
  \phi\cdot\Delta\Agg(S,\rho)\cdot\MRS(S,\mathbf p;w).
\]
If \(p\mapsto\Delta(p)\) is strictly decreasing, the RHS is equivalent to
\[
  p_{S+1}
  \ge
  \Delta^{-1}
  \left(
    \phi\cdot\Delta\Agg(S,\rho)\cdot\MRS(S,\mathbf p;w)
  \right)
  =
  \tau(S,\mathbf p;w).
\]
This proves the cutoff condition.

If the district switch margin \(\Delta\) is only weakly decreasing, the same
argument uses the generalized inverse $  \Delta^{-1}_g(c)
  :=
  \inf\{p\in[0,1]:\Delta(p)\le c\},$ together with the convention in footnote~\ref{fn:generalized-inverse} for interpreting comparisons with the generalized inverse.
\end{proof}

\begin{proof}[\textbf{Proof of Theorem~\ref{thm:optimal-cutoff-rule}.}]
If \(S^\Dist=S^\Agg\), Lemma~\ref{lem:general-reduction} implies that the
top-\(S^\Dist\) allocation is optimal. Hence suppose \(S^\Dist<S^\Agg\), and first take \(w\in(0,1)\).

First, \(\Delta\Agg(S,\rho)\) is positive and weakly decreasing
in \(S\). The average \(\MRS(S,\mathbf p;w)\) is also positive
and, by Lemma~\ref{lem:kappa-single-crossing}\textup{(iii)},
weakly decreasing in \(S\). Thus, $\phi\,\Delta\Agg(S,\rho)\,\MRS(S,\mathbf p;w)
$
is weakly decreasing in \(S\). Since the generalized inverse
\(\Delta^{-1}\) is weakly decreasing in its argument,
\[
  \tau(S,\mathbf p;w)=\Delta^{-1}(\phi\,\Delta\Agg(S,\rho)\,\MRS(S,\mathbf p;w))
\]
is weakly increasing in \(S\). This proves the first claim.

Since the ordered vote shares decrease while the cutoffs weakly
increase, the moves satisfying
\(p_{S+1}\ge\tau(S,\mathbf p;w)\) form an initial, possibly empty,
block of \(\{S^\Dist,\ldots,S^\Agg-1\}\), with comparisons
interpreted according to footnote~\ref{fn:generalized-inverse}.
Starting from \(S^\Dist\), let \(S^*\) be the seat total obtained
by accepting all moves in this block.

By Theorem~\ref{thm:general-cutoff}, the objective weakly decreases
up to \(S^*\) and strictly increases thereafter within
\(\{S^\Dist,\ldots,S^\Agg\}\). Thus \(S^*\) minimizes the objective
on this range. By Lemma~\ref{lem:general-reduction}, the
top-\(S^*\) allocation is optimal, proving the second claim.

Finally, at \(w=0\), the cutoff is constant at \(\Delta^{-1}(0)\), and moves satisfying the cutoff condition leave district loss unchanged. At \(w=1\), the limiting cutoff is zero, so the rule selects the optimal top-\(S^\Agg\) allocation. Thus both claims hold at the endpoints.

\end{proof}

\begin{proof}[\textbf{Proof of Proposition~\ref{prop:monotone-support-sets}.}]
By Lemma~\ref{lem:general-reduction}, for every $w\in[0,1]$ at least one
minimizer lies in $  \{S^\Dist,\ldots,S^\Agg\}.$ Hence it is enough to
characterize optimality on this range. If $S^\Dist=S^\Agg$, then $S^\Dist$ is
optimal for every $w\in[0,1]$, so
$\mathcal W_{S^\Dist}(\mathbf p)=[0,1]$ and the result is immediate.

Suppose therefore that $S^\Dist<S^\Agg$. For each
$S\in\{S^\Dist,\ldots,S^\Agg-1\}$, write
\[
  \Gamma_S(w):=M(S+1,\mathbf p;w)-M(S,\mathbf p;w).
\]
Thus the move from $S$ to $S+1$ is weakly improving  when
$\Gamma_S(w)\le0$. Define
\[
  \mathcal U_S:=\{w\in[0,1]:\Gamma_S(w)\le0\},
  \qquad
  \mathcal L_S:=\{w\in[0,1]:\Gamma_S(w)\ge0\}.
\]

By Lemma~\ref{lem:kappa-single-crossing}\textup{(iv)}, the first order
condition is single crossing in $w$. Therefore $\mathcal U_S$ is an upper
interval and $\mathcal L_S$ is a lower interval. Continuity of
$w\mapsto M(\cdot,\cdot;w)$ implies that $\mathcal U_S$ and $\mathcal L_S$ are
closed subsets of $[0,1]$.

Moreover, $\Gamma_S(0)\ge0$, because at $w=0$ only district loss matters and
district loss is weakly increasing on the stated range. Also, $\Gamma_S(1)<0$,
because at $w=1$ only statewide loss matters and
$\Agg(S+1,\rho)<\Agg(S,\rho)$. Hence there exist $0\le \ell_{S+1}\le u_S\le1$ such
that $  \mathcal U_S=[\ell_{S+1},1],$ and $  \mathcal L_S=[0,u_S].$

Next we convert the signs of the misrepresentation differences into support sets. By
Lemma~\ref{lem:kappa-single-crossing}\textup{(iii)}, for each fixed $w$, the
first order conditions are single crossing in $S$. Therefore a seat total is
optimal  when all incremental seat changes below it are weakly improving and all
incremental seat changes above it are weakly non-improving. Therefore, by definition $  \mathcal W_{S^\Dist}(\mathbf p)=\mathcal L_{S^\Dist},$ and $\mathcal W_{S^\Agg}(\mathbf p)=\mathcal U_{S^\Agg-1},$ and, for every $S\in\{S^\Dist+1,\ldots,S^\Agg-1\}$, $\mathcal W_S(\mathbf p)=\mathcal U_{S-1}\cap \mathcal L_S.$ 

Using the interval forms above, this can be stated as 
\begin{align*}
    \mathcal W_{S^\Dist}(\mathbf p)&=[0,u_{S^\Dist}],  \\
    \mathcal W_{S^\Agg}(\mathbf p)&=[\ell_{S^\Agg},1], \\
      \mathcal W_S(\mathbf p) &=[\ell_S,u_S]   \quad   \text{for } S\in\{S^\Dist+1,\ldots,S^\Agg-1\}.
\end{align*}

It remains to show that the support sets for interior seat totals are nonempty. Suppose, toward a
contradiction, that $\ell_S>u_S$ for some
$S\in\{S^\Dist+1,\ldots,S^\Agg-1\}$. Choose any $w\in(u_S,\ell_S)$. Then
$w\notin \mathcal U_{S-1}$, so $\Gamma_{S-1}(w)>0$, and $w\notin \mathcal L_S$, so
$\Gamma_S(w)<0$. This positive-then-negative pattern is ruled out by single
crossing in $S$. Hence $\ell_S\le u_S$, and every support set is a nonempty
closed interval.

Finally, single crossing in $S$ orders the incremental comparison sets, that is, for every
$S\in\{S^\Dist,\ldots,S^\Agg-2\}$, $  \mathcal U_{S+1}\subseteq \mathcal U_S,$ and  $  \mathcal L_S\subseteq \mathcal L_{S+1}.$ If $w\in \mathcal U_{S+1}$, then $\Gamma_{S+1}(w)\le0$, so single crossing gives
$\Gamma_S(w)\le0$, and therefore $w\in \mathcal U_S$. Similarly, if
$w\in \mathcal L_S$, then $\Gamma_S(w)\ge0$, so single crossing gives
$\Gamma_{S+1}(w)\ge0$, and therefore $w\in \mathcal L_{S+1}$. Since
$\mathcal U_S=[\ell_{S+1},1]$, $\mathcal U_{S+1}=[\ell_{S+2},1]$,
$\mathcal L_S=[0,u_S]$, and $\mathcal L_{S+1}=[0,u_{S+1}]$, these inclusions imply
$\ell_{S+1}\le \ell_{S+2}$ and $u_S\le u_{S+1}$.

Writing $\mathcal W_S(\mathbf p)=[\ell_S,u_S]$, the formulas above give
$\ell_{S^\Dist}=0$ and $u_{S^\Agg}=1$. Therefore
\[
  0=\ell_{S^\Dist}
  \le
  \ell_{S^\Dist+1}
  \le
  \cdots
  \le
  \ell_{S^\Agg},
\quad \text{ and } \quad 
  u_{S^\Dist}
  \le
  u_{S^\Dist+1}
  \le
  \cdots
  \le
  u_{S^\Agg}=1.
\]
This proves the claimed ordering of the support intervals.
\end{proof}

\begin{proof}[\textbf{Proof of Proposition~\ref{prop:linear-frontier}.}]

Fix a profile $\mathbf p$. First, note that only top-$S$ allocations can be Pareto efficient or minimize the linear objective. If an allocation with seat total $S$ assigns party $A$ to district $j$ but not to district $i$, with $p_i>p_j$, then swapping the two assignments changes district loss by
\[
  \bigl[(1-p_i)+p_j\bigr]-\bigl[p_i+(1-p_j)\bigr]
  =
  2(p_j-p_i)<0,
\]
where the strict inequality holds as $p_i\neq p_j$ for $i \neq j$. Repeating this swap eliminates every inversion. Thus, for each seat total $S$, the top-$S$ allocation is the unique district-loss minimizer among allocations with that seat total. Since statewide loss is constant within a seat-total class, if $\mathbf x$ is not a top-$S(\mathbf x)$ allocation, it cannot be Pareto efficient or minimize the linear objective for any $w<1$.

Write $  m_w(S):=(1-w)\Dist(S,\mathbf p)+w|S-\rho|$. Suppose first that $\widehat{\mathbf x}$ minimizes the linear objective for some $w\in[0,1)$. If $w=0$, then $\widehat{\mathbf x}$ is the unique district-loss minimizer because there are no district-level ties. Hence no other allocation can weakly lower district loss and strictly improve at least one coordinate. If $w\in(0,1)$, any allocation that weakly lowered both $\Dist$ and $\Agg$, with at least one strict reduction, would strictly lower the weighted objective, contradicting optimality. Thus $\widehat{\mathbf x}$ is Pareto efficient.

Conversely, suppose $\widehat{\mathbf x}$ is Pareto efficient. By the first paragraph, $\widehat{\mathbf x}$ is the top-$S^*$ allocation for $S^*=S(\widehat{\mathbf x})$. By the same dominance argument used in Lemma~\ref{lem:general-reduction}, $S^*$ lies between $S^\Dist$ and $S^\Agg$ in the maintained case $S^\Dist\le S^\Agg$.

If $S^\Dist=S^\Agg$, this common seat total minimizes both components, so the top-$S^*$ allocation minimizes the linear objective for every $w\in[0,1)$. Now suppose $S^\Dist<S^\Agg$. For $S=S^\Dist,\ldots,S^\Agg-1$, define
\[
  w_{S+1}
  :=
  \frac{\Delta\Dist(S,\mathbf p)}
       {\Delta\Dist(S,\mathbf p)+\Delta\Agg(S,\rho)},
  \qquad
  w_{S^\Dist}:=0.
\]
Along this path, $\Delta\Dist(S,\mathbf p)\ge0$ is weakly increasing and $\Delta\Agg(S,\rho)>0$ is weakly decreasing, so these critical weights are ordered and all are strictly below one. Moreover,
\[
  m_w(S+1)-m_w(S)
  =
  (1-w)\Delta\Dist(S,\mathbf p)-w\Delta\Agg(S,\rho),
\]
which is weakly negative if and only if $w\ge w_{S+1}$.

If $S^*<S^\Agg$, choose any $w\in[w_{S^*},w_{S^*+1}]$. The ordered critical weights imply that $m_w(S)$ weakly decreases up to $S^*$ and weakly increases after $S^*$, so $S^*$ minimizes $m_w$. Since $w_{S^*+1}<1$, this supporting weight satisfies $w<1$. If $S^*=S^\Agg$, choose any $w\in(w_{S^\Agg},1)$. Then all incremental moves up to $S^\Agg$ weakly lower the objective, and Lemma~\ref{lem:general-reduction} rules out minimizers beyond $S^\Agg$. Hence $S^*$ minimizes $m_w$ for some $w<1$. Therefore $\widehat{\mathbf x}$ minimizes the original allocation problem at a weight $w\in[0,1)$.
\end{proof}

\begin{proof}[\textbf{Proof of Proposition~\ref{prop:microfoundation}.}]
The local part of aggregate voter loss is
\begin{equation*}
\begin{aligned}
  (1-w)\sum_{d=1}^N
  \bigl[p_d(1-x_d)+(1-p_d)x_d\bigr]
  &=
  (1-w)
  \sum_{d=1}^N\bigl(x_d(1-p_d)+(1-x_d)p_d\bigr)\\
  &=
  (1-w)\Dist(\mathbf x,\mathbf p).
\end{aligned}
\end{equation*}
For the statewide component, use $\sum_d p_d=\rho$ and $\sum_d(1-p_d)=N-\rho$. Suppressing the argument $\mathbf x$ in $S(\mathbf x)$,
\begin{equation*}
w\sum_{d=1}^N
  \left[
    p_d\frac{(\rho-S)_+}{\rho}
    +(1-p_d)\frac{(S-\rho)_+}{N-\rho}
  \right]  =
  w\left[
    (\rho-S)_+ +(S-\rho)_+
  \right]
  =
  w|S-\rho|.
\end{equation*}
Therefore $  L(\mathbf x,\mathbf p;w)
  =
  (1-w)\Dist(\mathbf x,\mathbf p)+w|S(\mathbf x)-\rho|.$ Minimizing aggregate voter loss is  equivalent to minimizing $  (1-w)\Dist(\mathbf x,\mathbf p)+w|S(\mathbf x)-\rho|.$
\end{proof}

\begin{proof}[\textbf{Proof of Theorem~\ref{thm:linear-cutoff-path}.}]
By Lemma~\ref{lem:conditional-district-assignment}, it suffices to
consider top-$S$ allocations. For finite \(\phi\), minimizing the
linear objective is equivalent to minimizing
\(\Dist(S,\mathbf p)+\phi|S-\rho|\).
For any seat increase \(S\to S+1\),
\[
\Dist(S+1,\mathbf p)-\Dist(S,\mathbf p)
+\phi\bigl(|S+1-\rho|-|S-\rho|\bigr)
=
\begin{cases}
1-2p_{S+1}-\phi,
  & S<\underline\rho,\\[2pt]
1-2p_{S+1}+\phi\bigl(1-2(\rho-\underline\rho)\bigr),
  & S=\underline\rho,\\[2pt]
1-2p_{S+1}+\phi,
  & S>\underline\rho.
\end{cases}
\]
Since \(p_{S+1}\) is weakly decreasing in \(S\), these incremental
differences are weakly increasing in \(S\).

For \(S<\underline\rho\), we have \(p_{S+1}\ge p_{\underline\rho}\),
so district \(S+1\) clears \(t^{\mathrm D}(\mathbf p,\phi)\)
exactly when \(p_{S+1}\ge(1-\phi)/2\).
Moreover, \(t^{\mathrm U}(\mathbf p,\phi)\ge(1-\phi)/2\), since
\(0\le(\Delta\Agg(\underline\rho,\rho))_+\le1\).
Taking the minimum of the two cutoffs therefore gives
\[
p_{S+1}\ge t_\phi(\mathbf p)
\quad\Longleftrightarrow\quad
1-2p_{S+1}-\phi\le0.
\]

Now consider district \(\underline\rho+1\).
Party \(A\) wins at most \(\underline\rho\) districts under
first-past-the-post, so \(p_{\underline\rho+1}<1/2\).
Also,
\[
t^{\mathrm D}(\mathbf p,\phi)
\ge p_{\underline\rho}>p_{\underline\rho+1},
\]
using the convention \(p_0=1\) when \(\underline\rho=0\).
Thus this district clears \(t_\phi(\mathbf p)\) exactly when
it clears \(t^{\mathrm U}(\mathbf p,\phi)\).
If \(\rho-\underline\rho\le1/2\), then
\(t^{\mathrm U}(\mathbf p,\phi)=1/2\), so party \(B\) receives
the district. Its incremental difference is strictly positive.
If \(\rho-\underline\rho>1/2\), the definition of
\(t^{\mathrm U}\) gives
\[
p_{\underline\rho+1}\ge t^{\mathrm U}(\mathbf p,\phi)
\quad\Longleftrightarrow\quad
1-2p_{\underline\rho+1}
-\phi\bigl(2(\rho-\underline\rho)-1\bigr)\le0,
\]
which is exactly the condition that the incremental move is
weakly improving.

Finally, for \(S>\underline\rho\),
\[
p_{S+1}<p_{\underline\rho+1}\le t_\phi(\mathbf p),
\]
so party \(B\) receives district \(S+1\). Its incremental difference
\(1-2p_{S+1}+\phi\) is strictly positive because
\(p_{S+1}<1/2\).

Hence \(R_\phi(\mathbf p)\) assigns party \(A\) precisely those
districts whose corresponding incremental differences are weakly
negative. By the single crossing of incremental differences,
the top-$S$ allocation induced by this cutoff is optimal.
\end{proof}

\begin{proof}[\textbf{Proof of Corollary~\ref{cor:step-function}.}]
If $S^\Dist=S^\Agg$, the statement is immediate. Suppose $S^\Dist<S^\Agg$. By the first order condition \eqref{eq:kappa-foc}, specialized to the linear benchmark as in Theorem~\ref{thm:linear-cutoff-path}, an incremental move from $S$ to $S+1$ on the stated range is weakly profitable if and only if
\begin{equation*}
  (1-w)\Delta\Dist(S,\mathbf p)\le w\Delta\Agg(S,\rho).
\end{equation*}
For each $S\in\{S^\Dist,\ldots,S^\Agg-1\}$, define
\begin{equation*}
  w_{S+1}
  :=
  \frac{\Delta\Dist(S,\mathbf p)}{\Delta\Dist(S,\mathbf p)+\Delta\Agg(S,\rho)},
\end{equation*}
and set $w_{S^\Dist}:=0$. On this range, $\Delta\Dist(S,\mathbf p)\ge0$ is weakly increasing and $\Delta\Agg(S,\rho)>0$ is weakly decreasing, so these critical weights are ordered:
\begin{equation*}
0=w_{S^\Dist}< w_{S^\Dist+1}\le\cdots\le w_{S^\Agg}< 1.
\end{equation*}

By Lemma~\ref{lem:general-reduction}, it suffices to compare seat totals in $\{S^\Dist,\ldots,S^\Agg\}$. Now take $w\in(w_S,w_{S+1})$. For every $L<S$, the ordered critical weights imply that the move $L\to L+1$ strictly lowers the objective. For every $L\in\{S,\ldots,S^\Agg-1\}$, the move $L\to L+1$ strictly raises the objective. Hence the objective decreases up to $S$ and increases after $S$, so $S$ is the unique optimal seat total. If $w>w_{S^\Agg}$, every move up to $S^\Agg$ lowers the objective, and Lemma~\ref{lem:general-reduction} rules out minimizers beyond $S^\Agg$. Hence $S^\Agg$ is the unique optimal seat total.

At the switch weight $w=w_{S+1}$, the incremental move condition holds with equality, so $S$ and $S+1$ both minimize. If several switch weights coincide, the same argument gives the corresponding block of tied minimizers.
\end{proof}

\begin{proof}[\textbf{Proof of Lemma~\ref{lem:maj-fixedS}.}]
Let \(\rho:=\sum_{d=1}^N p_d=\sum_{d=1}^N q_d\). For either ordered profile \(\mathbf r\in\{\mathbf p,\mathbf q\}\), 

\[ \Dist(S,\mathbf r)=S+\rho-2\sum_{d=1}^S r_d. \] 

Since statewide misrepresentation is the same at any fixed \(S\), 

\[ M(S,\mathbf p;w)-M(S,\mathbf q;w) =-2(1-w)\left(\sum_{d=1}^S p_d-\sum_{d=1}^S q_d\right) \le0, \] 

where the inequality follows from majorization.
\end{proof}

\begin{proof}[\textbf{Proof of Proposition~\ref{prop:maj-opt}.}]
Fix \(w\in[0,1]\) and let \(S^*:=S_{R_w}(\mathbf q)\). Then
\[
  M(R_w,\mathbf p;w)
  \le M(S^*,\mathbf p;w)
  \le M(S^*,\mathbf q;w)
  =M(R_w,\mathbf q;w).
\]
The first inequality follows from the optimality of \(R_w\) at
\(\mathbf p\), and the second from Lemma~\ref{lem:maj-fixedS}.
\end{proof}

\begin{proof}[\textbf{Proof of Theorem~\ref{thm:maj-char-weight-family}.}]
Fix \(w\in[0,1]\). If \(\lambda=w\), the result follows from
Proposition~\ref{prop:maj-opt}. If \(\lambda=1\), profiles with
the same mean have the same proportional seat total, so the
result follows from Lemma~\ref{lem:maj-fixedS}.

Conversely, suppose \(\lambda<1\) and \(\lambda\ne w\), and let
\(s:=\operatorname{sgn}(w-\lambda)\). Choose
\(0<\varepsilon<1/16\) and \(0<\delta_\lambda<\varepsilon/2\)
such that $ (1-w)\delta_\lambda<|w-\lambda|\varepsilon$ and $\delta_\lambda\ne\lambda\varepsilon.$

For fixed \(w\) and \(\varepsilon\), such a choice is possible
for every \(\lambda\ne w\): take \(\delta_\lambda\) sufficiently
small and exclude the single value \(\lambda\varepsilon\).

For \(d=4,\ldots,N\), set \(a_d:=2^{-(d+1)}\) and let
\(L:=\sum_{d=4}^N a_d<1/16\), with an empty tail and \(L=0\)
when \(N=3\). Define
\[
  \mathbf q=
  \left(
    \frac34,\,
    \frac12+s\lambda\varepsilon-\delta_\lambda,\,
    \frac14-s\varepsilon+\delta_\lambda-L,\,
    a_4,\ldots,a_N
  \right),
\]
and
\[
  \mathbf p
  =\mathbf q+(0,2\delta_\lambda,-2\delta_\lambda,0,\ldots,0).
\]
The chosen bounds ensure that both profiles are strictly ordered,
lie in \((0,1)^N\), and have no share equal to \(1/2\).
Their common statewide support is
\[
  \rho=\frac32-s(1-\lambda)\varepsilon\in(1,2).
\]
Their ordered partial sums agree except at the second coordinate,
where the sum for \(\mathbf p\) exceeds that for \(\mathbf q\)
by \(2\delta_\lambda\). Thus \(\mathbf p\) majorizes \(\mathbf q\).

For either profile \(\mathbf r\in\{\mathbf p,\mathbf q\}\),
\[
  M(1,\mathbf r;\lambda)-M(0,\mathbf r;\lambda)
  =-\frac{1+\lambda}{2}<0.
\]
Every seat beyond the second increases both district-level and
statewide misrepresentation, since the remaining shares are below
\(1/2\) and \(\rho<2\). Moreover,
\[
\begin{aligned}
  M(2,\mathbf q;\lambda)-M(1,\mathbf q;\lambda)
    &=2(1-\lambda)\delta_\lambda>0,\\
  M(2,\mathbf p;\lambda)-M(1,\mathbf p;\lambda)
    &=-2(1-\lambda)\delta_\lambda<0.
\end{aligned}
\]
Hence the unique optimal seat totals satisfy
\(S_{R_\lambda}(\mathbf q)=1\) and
\(S_{R_\lambda}(\mathbf p)=2\).

Finally, compare the selected allocations at weight \(w\).
Since the profiles have the same first coordinate and statewide
support, \(M(1,\mathbf p;w)=M(1,\mathbf q;w)\). Thus the difference
between the selected outcomes equals the change from one seat to
two at \(\mathbf p\):
\[
\begin{aligned}
  M(R_\lambda,\mathbf p;w)-M(R_\lambda,\mathbf q;w) &=M(2,\mathbf p;w)-M(1,\mathbf p;w)\\
  &=(1-w)(1-2p_2)+w(3-2\rho)\\
  &=(1-w)(-2s\lambda\varepsilon-2\delta_\lambda)
          +2ws(1-\lambda)\varepsilon\\
  &=2s(w-\lambda)\varepsilon-2(1-w)\delta_\lambda\\
  &=2|w-\lambda|\varepsilon-2(1-w)\delta_\lambda>0.
\end{aligned}
\]
The last equality uses \(s(w-\lambda)=|w-\lambda|\),
and the strict inequality follows from the choice of
\(\delta_\lambda\). For each \(\lambda\ne w\), the transfer is
chosen small enough to satisfy
\((1-w)\delta_\lambda<|w-\lambda|\varepsilon\); the resulting
profiles therefore depend on \(\lambda\). Hence, for every
\(\lambda<1\) with \(\lambda\ne w\), there are profiles
\(\mathbf p\) and \(\mathbf q\) such that \(\mathbf p\) majorizes
\(\mathbf q\) while \(R_\lambda\) yields strictly higher loss at
\(\mathbf p\) when evaluated at \(w\). This violates majorization
monotonicity.
\end{proof}

\begin{proof}[\textbf{Proof of Proposition~\ref{prop:fptp-char}.}]
For this proof, district labels are held fixed across profiles, as required by the componentwise comparison in strong monotonicity. We sort shares only when identifying the top-$S$ allocation.

\textbf{If.} Suppose $\lambda=0$. Then the linear objective is purely district-level: $  M(\mathbf x,\mathbf p;0)=\Dist(\mathbf x,\mathbf p).$ Thus the rule is first-past-the-post: $R_0(\mathbf p)=\{d:p_d\ge1/2\}.$

If $\mathbf p'\ge\mathbf p$ componentwise, then $p_d\ge1/2$ implies $p_d'\ge1/2$, so $R_0(\mathbf p)\subseteq R_0(\mathbf p')$. Hence first-past-the-post satisfies strong monotonicity.

\smallskip
\noindent
\textbf{Only if.}
Now fix $\lambda\in(0,1]$. Choose $  0<\varepsilon<
  \min\left\{\frac{1}{12},\frac{\lambda}{2+4\lambda}\right\}.$
Consider the profiles\footnote{For $N>3$, we use tied district
shares for simplicity. A sufficiently small common perturbation
of the zero entries makes all shares distinct while preserving
the componentwise comparison and the selected districts.}
\begin{equation*}
  \mathbf p=
  \left(\frac12+\varepsilon,\frac12-\varepsilon,0,\ldots,0\right)
\quad \text{ and } \quad 
  \mathbf q=
  \left(\frac12+\varepsilon,\frac12+2\varepsilon,0,\ldots,0\right).
\end{equation*}
Then $\mathbf q\ge\mathbf p$ componentwise.

At profile $\mathbf p$, district $1$ has share $1/2+\varepsilon$ and district $2$ has share $1/2-\varepsilon$, and the statewide vote total is $\rho=1$.The  difference from $0$ to $1$ seats is $  (1-\lambda)(-2\varepsilon)+\lambda(-1)<0,$  and the difference from $1$ to $2$ seats is $  (1-\lambda)(2\varepsilon)+\lambda>0.$ Thus the unique optimal seat total is $S=1$, and the unique top-ranked district is district $1$. Hence $R_\lambda(\mathbf p)=\{1\}$.

At profile $\mathbf q$, district $1$ has share $1/2+\varepsilon$ and district $2$ has share $1/2+2\varepsilon$, the statewide vote total is $\rho_{\mathbf q}:=\sum_i q_i=1+3\varepsilon$, and the top-ranked district is district $2$. The difference from $0$ to $1$ seats is $  (1-\lambda)(-4\varepsilon)+\lambda(-1)<0.$ 

The difference from $1$ to $2$ seats is $  (1-\lambda)(-2\varepsilon)+\lambda(1-6\varepsilon).$ This expression equals $\lambda-(2+4\lambda)\varepsilon$, which is strictly positive by the choice of $\varepsilon$. Therefore the unique optimal seat total at $\mathbf q$ is again $S=1$, but the unique top-ranked district is now district $2$. Hence $R_\lambda(\mathbf q)=\{2\}$.

Thus $\mathbf q\ge\mathbf p$, but district $1$ belongs to $R_\lambda(\mathbf p)$ and not to $R_\lambda(\mathbf q)$. Strong monotonicity fails for every $\lambda>0$. Therefore, within the linear misrepresentation-minimizing family, strong monotonicity holds if and only if $\lambda=0$.
\end{proof}

\begin{proof}[\textbf{Proof of Proposition~\ref{prop:gp-one}.}]
\textbf{If.} Suppose $\lambda=1$. Then only statewide misrepresentation matters. The selected seat total is an integer closest to $\rho$, using the maintained tie-breaking convention. Hence the seat total depends only on the statewide mean, not on the distribution of support across districts. The rule is gerrymandering-proof.

\smallskip
\noindent
\textbf{Only if.} Suppose $\lambda<1$. Pick $\delta\in(0,\tfrac14)$ and $\varepsilon>0$ such that $  \varepsilon<
  \min\left\{\delta(1-\lambda),\,\frac12-2\delta\right\}.$ Consider the profiles\footnote{As in the proof of Proposition \ref{prop:fptp-char}, for $N>3$, we use tied district
shares for simplicity. A sufficiently small common perturbation
of the zero entries makes all shares distinct while preserving
the componentwise comparison and the selected districts.}
\begin{equation*}
  \mathbf p
  :=
  \left(
  \underbrace{1-\frac{\delta-\varepsilon}{N-2},\ldots,
  1-\frac{\delta-\varepsilon}{N-2}}_{N-2},
  \frac12+\delta,
  0
  \right)
\quad \text{ and } \quad 
  \mathbf q=
  \left(
  \underbrace{1,\ldots,1}_{N-2},
  \frac12-\delta,
  \varepsilon+\delta
  \right).
\end{equation*}
By construction, both profiles have the same statewide vote total $  \rho=N-2+\frac12+\varepsilon.$ Moreover, $\varepsilon<\delta$ implies $1-(\delta-\varepsilon)/(N-2)<1$, and $\delta<1/4$ implies
\begin{equation*}
  1-\frac{\delta-\varepsilon}{N-2}
  \ge
  1-\delta
  >
  \frac12+\delta.
\end{equation*}
Finally, $\varepsilon<1/2-2\delta$ implies $1/2-\delta\ge\varepsilon+\delta$. Hence both profiles are ordered.

At profile $\mathbf p$, the incremental differences around $N-1$ seats are
\begin{equation*}
  M(N-1,\mathbf p;\lambda)-M(N-2,\mathbf p;\lambda)
  =
  -2(1-\lambda)\delta-2\lambda\varepsilon
  <
  0
\end{equation*}
and
\begin{equation*}
  M(N,\mathbf p;\lambda)-M(N-1,\mathbf p;\lambda)=1>0.
\end{equation*}
Therefore $S_{R_\lambda}(\mathbf p)=N-1$.

At profile $\mathbf q$, the incremental differences around $N-2$ seats are
\begin{equation*}
  M(N-1,\mathbf q;\lambda)-M(N-2,\mathbf q;\lambda)
  =
  2(1-\lambda)\delta-2\lambda\varepsilon
  >
  0,
\end{equation*}
where the inequality follows from $\varepsilon<\delta(1-\lambda)$ when $\lambda>0$ and is immediate when $\lambda=0$. Also,
\begin{equation*}
  M(N-2,\mathbf q;\lambda)-M(N-3,\mathbf q;\lambda)=-1<0.
\end{equation*}
Therefore $S_{R_\lambda}(\mathbf q)=N-2$.

The profiles have the same statewide mean but induce different seat totals under $R_\lambda$. Hence $R_\lambda$ is not gerrymandering-proof. Therefore, within the normalized linear family, gerrymandering-proofness holds if and only if $\lambda=1$.
\end{proof}

\begin{proof}[\textbf{Proof of Proposition~\ref{prop:gerrymandering-set-inclusion}.}]
Fix a baseline profile $\mathbf p$ and $\lambda_0$ such that $|R_{\lambda_0}(\mathbf p)|>\rho$. Let $k>|R_{\lambda_0}(\mathbf p)|$. Then $k>\rho$.

Since $|R_{\lambda_0}(\mathbf p)|$ is an integer and $k>|R_{\lambda_0}(\mathbf p)|>\rho$, we have $k-1\ge \rho$. Fix any same-mean profile $\mathbf r\in\mathcal P(\bar p)$. Increasing the seat total from $k-1$ to $k$ changes the statewide loss by
\begin{equation*}
  |\rho-k|-|\rho-(k-1)|=1.
\end{equation*}
Thus
\begin{equation*}
  \Delta_k(\mathbf r;\lambda)
  :=
  M(k,\mathbf r;\lambda)-M(k-1,\mathbf r;\lambda)
  =
  (1-\lambda)(1-2r_k)+\lambda.
\end{equation*}
If $\lambda'=1$, then $\mathcal G_k(\mathbf p;\lambda')=\emptyset$: the pure statewide rule cannot award an overrepresentation-side marginal seat. The inclusion is then immediate.

Suppose, therefore, that $\lambda'<1$. If $\mathbf r\in\mathcal G_k(\mathbf p;\lambda')$, then the objective at weight $\lambda'$ is unimodal in the seat total. Under the maintained tie-breaking convention in favor of party $A$, obtaining at least $k$ seats implies $  M(k,\mathbf r;\lambda')\le M(k-1,\mathbf r;\lambda'),$  or $\Delta_k(\mathbf r;\lambda')\le0$. Since
\begin{equation*}
  \Delta_k(\mathbf r;\lambda)
  =
  1-2(1-\lambda)r_k
  \le
  1-2(1-\lambda')r_k
  =
  \Delta_k(\mathbf r;\lambda')
\end{equation*}
because $\lambda'>\lambda$. Hence $\Delta_k(\mathbf r;\lambda)\le0$, so $\mathbf r\in\mathcal G_k(\mathbf p;\lambda)$. Therefore $  \mathcal G_k(\mathbf p;\lambda')
  \subseteq
  \mathcal G_k(\mathbf p;\lambda).$

For strictness, suppose $\mathcal G_k(\mathbf p;\lambda)\neq\emptyset$ and $\lambda'>\lambda$. For $\lambda<1$, the condition $\Delta_k(\mathbf r;\lambda)\le0$ is equivalent to the threshold condition $  r_k\ge \theta(\lambda),$ and $  \theta(\lambda):=\frac{1}{2(1-\lambda)}.$

The threshold is strictly increasing in $\lambda$. If $\lambda'=1$, then $\mathcal G_k(\mathbf p;\lambda')=\emptyset$ while $\mathcal G_k(\mathbf p;\lambda)$ is nonempty by assumption. If $\lambda'<1$, nonemptiness and distinct shares imply $\theta(\lambda)<\rho/k$. Choose $t\in(\theta(\lambda),\min\{\theta(\lambda'),\rho/k\})$, set the first $k$ districts initially equal to $t$, set all later districts equal to zero, and distribute the remaining mass $\rho-kt$ among the first $k-1$ districts. This is feasible because $k-1\ge \rho$. A sufficiently small mean-preserving perturbation makes all district shares distinct while keeping the $k$th share strictly between $\theta(\lambda)$ and $\theta(\lambda')$. The constructed profile belongs to $\mathcal G_k(\mathbf p;\lambda)$ but not to $\mathcal G_k(\mathbf p;\lambda')$. Hence the inclusion is strict.
\end{proof}

\begin{proof}[\textbf{Proof of Proposition \ref{prop:gerrymandering-capacity-bound}.}]
If \(\lambda=1\), the rule minimizes only the statewide loss. Since \(k\ge 1+\rho\),
we have \(k-1\ge \rho\), so every seat total \(S\ge k\) is farther from
\(\rho\) than \(k-1\). Hence no minimizer can give party \(A\) at least \(k\)
seats, and \(\mathcal G_k(\mathbf p;1)=\emptyset\).

Now let \(\lambda<1\). Suppose, toward a contradiction, that
\(\mathcal G_k(\mathbf p;\lambda)\neq\emptyset\). Then there exists
\(\mathbf q\in\mathcal P(\bar p)\) such that \(|R_\lambda(\mathbf q)|\ge k\). Order
\(\mathbf q\) so that \(q_1\ge\cdots\ge q_N\). By the top-\(S\) reduction, the
seat total selected by \(R_\lambda\) is an optimal seat total for the reduced problem.
Since this optimal seat total is at least \(k\), the move from
\(k-1\) to \(k\) must be weakly improving: $  M(k,\mathbf q;\lambda)\le M(k-1,\mathbf q;\lambda).$

Because \(k-1\ge\rho\), the incremental difference is
\[
  M(k,\mathbf q;\lambda)-M(k-1,\mathbf q;\lambda)
  =
  (1-\lambda)(1-2q_k)+\lambda \leq 0.
\]
Equivalently, $  q_k\ge \frac{1}{2(1-\lambda)}.$ Since the profile is ordered and \(\mathbf q\in\mathcal P(\bar p)\),
\[
  \rho
  =
  \sum_{d=1}^N q_d
  \ge
  \sum_{d=1}^k q_d
  \ge
  kq_k
  \ge
  \frac{k}{2(1-\lambda)}.
\]
Rearranging gives $  \lambda\le 1-\frac{k}{2\rho},$ contradicting the hypothesis. Therefore
\(\mathcal G_k(\mathbf p;\lambda)=\emptyset\).
\end{proof}

\begin{proof}[\textbf{Proof of Proposition \ref{prop:eg-statewide}.}]
Let \(x_d^F=1\) if party \(A\) wins district \(d\) under first-past-the-post,
and let \(x_d^F=0\) otherwise. Then \(\sum_{d=1}^N x_d^F=S\).

If party \(A\) wins district \(d\), party \(B\) wastes \(1-p_d\) votes and
party \(A\) wastes \(p_d-\frac12\). Hence the district's contribution to
\(W_B(\mathbf p)-W_A(\mathbf p)\) is $  (1-p_d)-\left(p_d-\frac12\right)
  =
  \frac32-2p_d.$ If party \(A\) loses district \(d\), party \(B\) wastes \(\frac12-p_d\) votes
and party \(A\) wastes \(p_d\). Hence the contribution is $  \left(\frac12-p_d\right)-p_d
  =
  \frac12-2p_d.$ Combining the two cases, district \(d\)'s contribution is $  x_d^F+\frac12-2p_d.$ Summing across districts gives
\[
  W_B(\mathbf p)-W_A(\mathbf p)
  =
  \sum_{d=1}^N
  \left(x_d^F+\frac12-2p_d\right)
  =
  S+\frac N2-2\rho.
\]
Therefore
\[
  EG_A(\mathbf p)
  =
  \frac{W_B(\mathbf p)-W_A(\mathbf p)}{N}
  =
  \frac{1}{N}
  \left[
    S-\left(2\rho-\frac N2\right)
  \right].
\]
Taking absolute values yields
\[
  |EG_A(\mathbf p)|
  =
  \frac{1}{N}
  \left|
    S-\left(2\rho-\frac N2\right)
  \right|.
\]
Thus the absolute efficiency gap is the statewide-misrepresentation term with
linear cost \(Q(z)=z/N\) and target \(T^{EG}(\rho)=2\rho-N/2\).
Equivalently, its $N$-scaled version uses $Q(z)=z$.
\end{proof}

\section{Measures of Representation} \label{app:seperation}

This appendix shows how several standard representation measures can be written inside the paper's district--statewide loss framework. Throughout, \(S\) denotes party \(A\)'s seat total and \(\rho=\sum_d p_d\) denotes party \(A\)'s statewide support in seat units.

\subsection{Efficiency Gap.}

\begin{proposition}
\label{prop:eg-linear-separations}
In the linear benchmark:
\begin{enumerate}[label=(\roman*)]
  \item Fix \(s\in\{0,\ldots,N\}\). Among all vote profiles such that \(S^\Dist=S^\Agg=s\),
  \[
    |EG_A|
    \le
    \overline{EG}(N,s)
    :=
    \begin{cases}
      \frac12, & \text{if }s\in\{0,N\},\\[6pt]
      \left|\frac12-\frac{s}{N}\right|+\frac1N,
      & \text{if }1\le s\le N-1.
    \end{cases}
  \]
\item For every \(s\in\{0,\ldots,N\}\), the bound is tight. That is, there are vote profiles with \(S^\Dist=S^\Agg=s\) for which \(|EG_A|\) is arbitrarily close to \(\overline{EG}(N,s)\).
  \item For every \(N\ge4\) and every \(\epsilon>0\), there exists a vote profile such that \(EG_A=0\), but \(w_{S^\Dist+1}<\epsilon\).
\end{enumerate}
\end{proposition}

\begin{proof}[\textbf{Proof of Proposition~\ref{prop:eg-linear-separations}.}] 
The constructions below use tied shares for simplicity. Arbitrarily small perturbations make the shares distinct while preserving the selected seat totals and the limits in parts (i)-(ii). In part (iii), perturb only the last $N-1$ shares while preserving their sum, which leaves both the eﬃciency gap and the first-switch weight unchanged.

Suppose first that \(S^\Dist=S^\Agg=s\). Since \(S^\Agg=s\), the statewide support \(\rho\) lies within one half-seat of \(s\), up to the maintained tie-breaking convention: $  \left|\rho-s\right|\le \frac12 .$

For \(1\le s\le N-1\), write \(\rho=s+r\), where \(r\in[-1/2,1/2]\). Then
\[
  EG_A
  =
  \frac{s}{N}
  -
  2\frac{s+r}{N}
  +
  \frac12
  =
  \frac12-\frac{s}{N}-\frac{2r}{N} \implies 
  |EG_A|
  \le
  \left|\frac12-\frac{s}{N}\right|+\frac1N.
\]
If \(s=0\), then \(\rho\in[0,1/2]\), so
\[
  EG_A
  =
  \frac12-2\frac{\rho}{N}
  \in
  \left[\frac12-\frac1N,\frac12\right] \implies |EG_A|\le1/2. \] If \(s=N\), then \(\rho\in[N-1/2,N]\), so
\[
  EG_A
  =
  \frac32-2\frac{\rho}{N}
  \in
  \left[-\frac12,-\frac12+\frac1N\right],
\implies |EG_A|\le1/2. \]
This proves part \textup{(i)}.

For tightness, if \(s=0\), take \(p_d=0\) for every district. Then \(S^\Dist=S^\Agg=0\) and \(EG_A=1/2\). If \(s=N\), take \(p_d=1\) for every district. Then \(S^\Dist=S^\Agg=N\) and \(EG_A=-1/2\).

Now let \(1\le s\le N-1\). If \(s\le N/2\), take \(\eta>0\) small and set
\[
  p_1=\cdots=p_s
  =
  1-\frac{1}{2s}+\frac{\eta}{s},
  \qquad
  p_{s+1}=\cdots=p_N=0.
\]
Then \(S^\Dist=s\) and \(\rho=s-1/2+\eta\). For small \(\eta\), the proportional seat total is \(S^\Agg=s\), and
\[
  EG_A
  =
  \frac12-\frac{s}{N}+\frac1N-\frac{2\eta}{N}.
\]
Thus \(|EG_A|\) can be made arbitrarily close to
\(\left|1/2-s/N\right|+1/N\).

If \(s>N/2\), take \(\eta>0\) small and set
\[
  p_1=\cdots=p_s=1,
  \qquad
  p_{s+1}=\frac12-\eta,
  \qquad
  p_{s+2}=\cdots=p_N=0.
\]
Then \(S^\Dist=s\) and \(\rho=s+1/2-\eta\). For small \(\eta\), the proportional seat total is \(S^\Agg=s\), and
\[
  EG_A
  =
  \frac12-\frac{s}{N}-\frac1N+\frac{2\eta}{N}.
\]
Thus \(|EG_A|\) can again be made arbitrarily close to the stated bound. This proves the tightness claim.

For part \textup{(iii)}, fix \(N\ge4\) and \(\epsilon>0\). Choose \(0<\varepsilon<1/4\) such that $  \frac{2\varepsilon}{1+2\varepsilon}<\epsilon.$ Let
\[
  p_1=\frac12-\varepsilon,
  \qquad
  p_2=\cdots=p_N
  =
  \frac{N-2}{4(N-1)}+\frac{\varepsilon}{N-1}.
\]
For \(\varepsilon\) small enough, every district has party \(A\) vote share below one half, so \(S^\Dist=0\). The statewide support is
\[
  \rho
  =
  \left(\frac12-\varepsilon\right)
  +
  (N-1)
  \left(
  \frac{N-2}{4(N-1)}+\frac{\varepsilon}{N-1}
  \right)
  =
  \frac{N}{4} \implies 
  EG_A
  =
  \frac{0}{N}
  -
  2\frac{N/4}{N}
  +
  \frac12
  =
  0.
\]
Since \(N\ge4\), adding the first seat for party \(A\) moves the seat total weakly toward proportionality. The first possible switch away from first-past-the-post assigns party \(A\) the strongest district, where \(p_1=1/2-\varepsilon\). The incremental comparison is
\[
  M(1,\mathbf p;w)-M(0,\mathbf p;w)
  =
  (1-w)(1-2p_1)-w.
\]
Substituting \(p_1=1/2-\varepsilon\), this becomes \(2\varepsilon(1-w)-w\). The first switch occurs when this expression is weakly negative, so $  w_{S^\Dist+1}
  =
  \frac{2\varepsilon}{1+2\varepsilon}
  <
  \epsilon.$
  
Thus \(EG_A=0\), while first-past-the-post ceases to be optimal for an arbitrarily small positive weight on statewide proportionality.
\end{proof}

\medskip
\noindent
\textbf{Weighted wasted-vote targets.}
Dividing weighted wasted-vote differences by $N$ gives a family containing the efficiency gap at $\alpha=1$ and the proportional seat gap at $\alpha=0$. Let \(\alpha\ge0\) be the weight on surplus winning votes where losing votes receive a fixed weight of one. If party \(A\) wins district \(d\), the district contribution to \(W_B^\alpha-W_A^\alpha\) is
\[
  (1-p_d)-\alpha\left(p_d-\frac12\right)
  =
  1+\frac{\alpha}{2}-(1+\alpha)p_d.
\]
If party \(A\) loses district \(d\), the contribution is
\[
  \alpha\left(\frac12-p_d\right)-p_d
  =
  \frac{\alpha}{2}-(1+\alpha)p_d.
\]
Summing across districts, $  W_B^\alpha(\mathbf p)-W_A^\alpha(\mathbf p)
  =
  S^\Dist+\frac{\alpha N}{2}-(1+\alpha)\rho.$ Hence the absolute weighted wasted-vote gap is
\[
  \left|
    S^\Dist-
    \left((1+\alpha)\rho-\frac{\alpha N}{2}\right)
  \right|.
\]
It is a statewide loss with target $  T^\alpha(\rho)
  =
  (1+\alpha)\rho-\frac{\alpha N}{2}.$ The case \(\alpha=1\) gives the efficiency-gap target \(2\rho-N/2\). The case \(\alpha=0\), which counts only losing votes as wasted, gives the proportional target \(\rho\).

\medskip
\noindent
\textbf{Linear winner's-bonus targets.}
Equivalently, we can define
\[
  T^\beta(\rho)
  =
  \rho+\beta\left(\rho-\frac N2\right),
  \qquad
  \beta\ge0.
\]
The proportional target is \(\beta=0\), and the efficiency-gap target is \(\beta=1\). With linear gap cost,
\[
  \Agg^\beta(S,\rho)
  =
  \left|
    S-
    \left[
      \rho+\beta\left(\rho-\frac N2\right)
    \right]
  \right|.
\]
Positive \(\beta\) gives the statewide vote winner a seat bonus above proportionality and penalizes the statewide loser symmetrically.

\subsection{Partisan Bias}
\label{app:partisan-bias}
Let $S^\text{FPTP}(\mathbf r)$ be the first-past-the-post seat total at profile \(\mathbf r\). For any profile
for which the half-vote counterfactual \(\mathbf p^{1/2}\) lies in \([0,1]^N\),
\[
  S^\text{FPTP}(\mathbf p^{1/2})
  =
  \sum_{d=1}^N
  \mathbf 1\left\{
    p_d+\frac12-\bar p\ge \frac12
  \right\}
  =
  \sum_{d=1}^N \mathbf 1\{p_d\ge \bar p\}.
\]
Moreover,
\[
  \sum_{d=1}^N p^{1/2}_d
  =
  \sum_{d=1}^N p_d
  +
  N\left(\frac12-\bar p\right)
  =
  \frac N2.
\]
Thus, under the proportional statewide loss \(\Agg(S,\rho)=|S-\rho|\),
\[
  N|PB_A(\mathbf p)|
  =
  \left|S^\text{FPTP}(\mathbf p^{1/2})-\frac N2\right|
  =
  \Agg\left(S^\text{FPTP}(\mathbf p^{1/2}),\frac N2\right).
\]
This identity uses the counterfactual profile \(\mathbf p^{1/2}\), ignores the
district-level component, and fixes the allocation rule to first-past-the-post.
It is therefore distinct from total misrepresentation at the observed profile.

Since \(S^\text{FPTP}(\mathbf p^{1/2})\) is an integer,
\[
  |PB_A(\mathbf p)|
  =
  \left|
    \frac{S^\text{FPTP}(\mathbf p^{1/2})}{N}
    -
    \frac12
  \right|.
\]
Thus exact zero partisan bias is possible only when \(N\) is even. If \(N\) is
odd, the smallest possible value of \(|PB_A(\mathbf p)|\) is
\(\frac1{2N}\).

\begin{proposition}
\label{prop:pb-separation}
Consider the linear benchmark, and suppose \(N\ge3\).
\begin{enumerate}[label=\textup{(\roman*)},leftmargin=2em]
  \item For every \(\eta>0\), there exists a profile \(\mathbf p\) with
  \[
    |PB_A(\mathbf p)|
    =
    \begin{cases}
      0, & \text{if }N\text{ is even},\\[0.3em]
      \frac1{2N}, & \text{if }N\text{ is odd},
    \end{cases}
  \]
  but $  \Agg(S^\Dist,\rho)> \frac N2-\eta,$ and first-past-the-post ceases to
  be misrepresentation-minimizing at a statewide weight smaller than \(\eta\).

  \item There exists a profile \(\mathbf q\) with
  $  \Agg(S^\Dist,\rho)=0,$ but
  $  |PB_A(\mathbf q)|=\frac12-\frac1N.$
\end{enumerate}
\end{proposition}

\begin{proof}[\textbf{Proof of Proposition~\ref{prop:pb-separation}.}]
For part \textup{(i)}, let \(k=\lfloor N/2\rfloor\). Choose
\(\varepsilon>0\) small enough that $  \frac12+2\varepsilon<1,$  $  (N+k)\varepsilon<\frac N2-1,$ $  (N+k)\varepsilon<\eta,$ and $  \frac{2\varepsilon}{1+2\varepsilon}<\min\{\eta,1\}.$ This is possible because \(N\ge3\). Consider the profile
\[
  \mathbf p^\varepsilon
  =
  \left(
    \underbrace{\frac12+2\varepsilon,\ldots,\frac12+2\varepsilon}_{k},
    \underbrace{\frac12+\varepsilon,\ldots,\frac12+\varepsilon}_{N-k}
  \right).
\]
Then $  \bar p^\varepsilon
  =
  \frac12+\frac{N+k}{N}\varepsilon .$
Since $  \frac12+\varepsilon
  <
  \bar p^\varepsilon
  <
  \frac12+2\varepsilon,$ total of \(k\) districts lie above the statewide mean and $  S^\text{FPTP}\bigl((\mathbf p^\varepsilon)^{1/2}\bigr)=k.$ Therefore
\[
  |PB_A(\mathbf p^\varepsilon)|
  =
  \left|\frac{k}{N}-\frac12\right|
  =
  \begin{cases}
    0, & \text{if }N\text{ is even},\\[0.3em]
    \frac1{2N}, & \text{if }N\text{ is odd}.
  \end{cases}
\]
The odd-\(N\) value is the smallest possible value of \(|PB_A|\), since
\(S^\text{FPTP}(\mathbf p^{1/2})\) is an integer.

At the profile $  \mathbf p^\varepsilon$, every district has party \(A\) vote share above
\(1/2\), so first-past-the-post gives \(S^\Dist=N\). The proportional target is $  \rho
  =
  N\bar p^\varepsilon
  =
  \frac N2+(N+k)\varepsilon .$ Thus
\[
  \Agg(S^\Dist,\rho)
  =
  |N-\rho|
  =
  \frac N2-(N+k)\varepsilon
  >
  \frac N2-\eta .
\]

It remains to show that first-past-the-post stops minimizing at a statewide
weight below \(\eta\). Since $  \rho=\frac N2+(N+k)\varepsilon<N-1,$ moving from \(N\) seats to \(N-1\) lowers statewide misrepresentation by one
seat. The district switched away from party \(A\) has vote share
\(1/2+\varepsilon\), so district-level misrepresentation rises by $  \left(\frac12+\varepsilon\right)
  -
  \left(\frac12-\varepsilon\right)
  =
  2\varepsilon .$ Therefore
\[
  M(N-1,\mathbf p^\varepsilon;w)-M(N,\mathbf p^\varepsilon;w)
  =
  (1-w)2\varepsilon-w, 
\]
which is strictly negative whenever $  w>\frac{2\varepsilon}{1+2\varepsilon}.$ By the choice of \(\varepsilon\), there exists
\(w\in\left(\frac{2\varepsilon}{1+2\varepsilon},\min\{\eta,1\}\right)\). For
this \(w\), first-past-the-post is not misrepresentation-minimizing, and the
statewide weight is smaller than \(\eta\).

For part \textup{(ii)}, choose \(0<\varepsilon<1/N\), and define
\[
  \mathbf q^\varepsilon
  =
  \left(
    \frac12+\varepsilon,
    \underbrace{
      \frac{\frac12-\varepsilon}{N-1},
      \ldots,
      \frac{\frac12-\varepsilon}{N-1}
    }_{N-1}
  \right).
\]
The statewide support is $  \rho
  =
  \frac12+\varepsilon
  +
  (N-1)\frac{\frac12-\varepsilon}{N-1}
  =
  1.$ First-past-the-post awards party \(A\) one seat, so \(S^\Dist=1\). Hence $  \Agg(S^\Dist,\rho)=|1-1|=0.$

The statewide mean is \(\bar q^\varepsilon=1/N\). The first district has vote
share above \(1/N\), while each remaining district has vote share below \(1/N\).
Thus $  S^\text{FPTP}\bigl((\mathbf q^\varepsilon)^{1/2}\bigr)=1,$ and $  |PB_A(\mathbf q^\varepsilon)| = \left|\frac1N-\frac12\right| =  \frac12-\frac1N .$ 

In both constructions, suﬃciently small mean-preserving perturbations make all district shares distinct without changing the first-past-the-post seat total or partisan bias. In part (i), the strict loss and weight inequalities also persist. This proves the result.
\end{proof}

\section{Legal and Legislative Background}
\label{app:statewide-targets}

Legal decisions and legislative proposals evaluate representation both within districts and across a delegation. In congressional elections, the state is the jurisdiction within which districts allocate a fixed number of U.S.\ House seats. Statewide electoral support therefore provides a relevant reference for evaluating the resulting delegation. The sources below illustrate this concern, the different benchmarks used to express it, and its relationship to local representation.

\paragraph{Congressional elections.}

Section~5003 of the 2021 Freedom to Vote Act proposed evaluating congressional plans through quantitative measures relating statewide vote shares to seat shares.\footnote{The act did not identify a particular measure or specify a formula for calculating partisan advantage.} Its test used recent presidential and Senate election results to determine how many proposed districts each party's candidates would have carried. If these measures indicated a partisan advantage or disadvantage exceeding seven percent or one congressional district, whichever was greater, in at least two of the four elections assessed, the plan would face a rebuttable presumption of violating the bill's prohibition on partisan favoritism. 

In \emph{Carter v. Chapman} (2022), the Pennsylvania Supreme Court selected a congressional map after the legislature and governor failed to agree. The majority treated partisan fairness as necessary to ensure responsiveness to statewide voter preferences when choosing among plans meeting traditional districting criteria. It jointly evaluated the efficiency gap, the mean--median difference, and whether a statewide vote majority translated into a seat majority. The selected plan's average efficiency gap of $-0.4\%$ supported its selection. The decision also acknowledged the need to ``balance opposing criteria.'' Its analysis illustrates the relevance of assessing statewide partisan outcomes alongside territorial and community considerations.\footnote{See \emph{Carter v. Chapman} 270 A.3d 444, 470--471 (Pa. 2022).}

The earlier decision in \emph{League of Women Voters of Pennsylvania} (2018) supplies a broader representational principle. In invalidating the state's congressional map, the court interpreted the state constitution to require that ``all voters have an equal opportunity to translate their votes into representation.'' The opinion recorded Republicans winning 13 of 18 seats with 49.2\% of statewide votes in 2012 and discussed expert efficiency-gap evidence. It also linked keeping communities together to voters' ability to elect representatives reflecting their preferences. The holding rested on the subordination of neutral districting criteria to partisan advantage.\footnote{See \emph{League of Women Voters of Pennsylvania v. Commonwealth} 178 A.3d 737, 804 (Pa. 2018).}

In \emph{Adams v. DeWine} (2022), the Ohio Supreme Court invalidated a plan under the state constitution's prohibition on undue partisan favoritism. Its conclusion emphasized that a party generally receiving at most 55\% of statewide votes was positioned to win 75--80\% of congressional seats: ``By any rational measure, that skewed result just does not add up.'' The majority also relied on analyses using the efficiency gap, partisan symmetry, other metrics, and comparisons with simulated plans. The case illustrates judicial concern with statewide vote--seat discrepancies and the use of several benchmarks to assess them.\footnote{See \emph{Adams v. DeWine} 195 N.E.3d 74, para.~100 (Ohio 2022).}

\paragraph{State legislative elections.}

In \emph{Gaffney v. Cummings} (1973), the U.S.\ Supreme Court upheld Connecticut's state legislative districting plan, accepting its effort to ``provide a rough sort of proportional representation in the legislative halls of the State.'' The plan used previous election results to approximate statewide party strength, establishing proportionality as a permissible consideration in the plan under review.\footnote{See \emph{Gaffney v. Cummings} 412 U.S. 735, 754 (1973).}

In \emph{League of Women Voters of Ohio} (2022), the state supreme court applied a constitutional provision requiring partisan shares of state legislative districts to ``correspond closely to the statewide preferences of the voters of Ohio.'' Using statewide partisan elections over the preceding ten years, it identified approximately 54\% Republican and 46\% Democratic support and required an attempt to draw correspondingly leaning districts.\footnote{See \emph{League of Women Voters of Ohio v. Ohio Redistricting Commission} 192 N.E.3d 379, para.~104 (Ohio 2022).}

In \emph{Whitford v. Gill} (2016), concerning Wisconsin State Assembly districts, a federal court described the efficiency gap as measuring ``the magnitude of a plan's deviation from the relationship we would expect to observe between votes and seats.'' It used the measure as corroborating evidence of intentional, durable partisan disadvantage, without adopting a numerical threshold for presumptive unconstitutionality.\footnote{See \emph{Whitford v. Gill} 218 F. Supp. 3d 837, 907 (W.D. Wis. 2016). The judgment was vacated and remanded on standing grounds in \emph{Gill v. Whitford} 585 U.S. 48 (2018).}

\paragraph{Connection to the framework.}

These discussions motivate evaluating both how well each district's representative reflects its electorate and how the delegation as a whole reflects statewide support. Our objective captures these two concerns and makes explicit how much weight each receives when they conflict.

\section{Empirical Appendix}
\label{app:empirical}

\subsection{Constructing Two-Party Vote Shares}
\label{app:construction}
If a candidate appears on multiple party lines and one of those lines is a major party (Democratic or Republican), we combine the candidate's votes across lines and assign the total to that major party. State affiliates of the major parties, such as Minnesota's Democratic-Farmer-Labor Party, are coded as the corresponding major party. Candidates belonging to neither major party are labeled \textit{third party}, and a third-party candidate's votes across several minor-party lines are combined. Votes for unnamed write-in candidates and invalid ballots (blank, void, over-votes and under-votes) are retained but assigned to neither major party. Where a race went to a runoff, we use only the final round. This yields a candidate-level returns dataset.

We then rank named candidates by votes and apply a classification rule to decide whether the observed congressional returns can be used directly or must be proxied. The goal is to use observed congressional returns whenever they provide a credible measure of underlying two-party support, since these races preserve the district-specific information the model requires. A race is used as is if the top two finishers are one Democrat and one Republican\footnote{Classifying races by party vote totals rather than by individual candidates' votes flags exactly the same races. In every race, the top two candidates are one Democrat and one Republican if and only if both parties' summed votes exceed those of every third-party candidate.} and the non-major share (votes for third-party candidates, write-ins, and invalid ballots as a share of all votes cast) is at most $15\%$. In these races, votes for all Democratic candidates and all Republican candidates are summed, non-major votes are discarded, and the two totals are normalized to sum to one. Every other race is flagged for proxy, either because the top two finishers are not one Democrat and one Republican, or because the non-major share exceeds $15\%$, so that the two-party share reflects the preferences of too small a part of the electorate. Before the sample restrictions described below, 618 of the 4{,}381 races ($14\%$) require a proxy. 533 of the races do so because the top two finishers are not one Democrat and one Republican, and 85 because of a large non-major share. In 75 of the latter, the share exceeds $15\%$ only because of invalid ballots, almost all in New York and Massachusetts, which report blank and void ballots alongside candidate votes.

If a flagged race lacks one major-party candidate, we proxy its two-party vote share using unflagged races in the same district under the same congressional map. Specifically, for the party that ran without major-party opposition, we use its highest two-party vote share among the unflagged races it won in that district-map combination.\footnote{Uncontested races create a selection problem because the missing party's absence is itself informative about the district's partisan balance. Using the uncontested party's strongest contested performance within the same district-map combination preserves that information while anchoring the proxy to an observed two-party race in the same electoral geography. Because the proxy is restricted to elections held under the same enacted district plan, it normally reaches no farther than one decennial redistricting cycle. For races in 2000, the pool includes 1990s elections held under the same plan.} If no such congressional proxy is available, we use presidential returns in the same district. The same-year presidential two-party share in presidential years and in midterms, the average of the preceding and following presidential elections, or only the one held under the same congressional map if the other was not.\footnote{Presidential votes by congressional district were sourced from \url{https://sites.google.com/view/presidentialbycongressionaldis/}.} Where neither proxy is available, we drop the state-year.\footnote{This affects two state-years, Texas 2002 and New York 2022, whose maps were each used for only one election.}

When both major parties field candidates but the non-major share exceeds $15\%$, we use the average two-party share of the un-flagged races in the same district at $t-2$ and $t+2$ under the same map. If only one of these races is available, or if their average implies a different winner from the observed result, we use whichever adjacent race agrees with the observed winner. When no adjacent same-map race qualifies, we use the presidential proxy described above. Finally, we drop any state-year in which a proxied share is either unavailable or implies a different winner from the observed result.\footnote{Four state-years are dropped on this ground: New York 2000, 2002 and 2008, and California 2012. State-years in which an independent candidate won a seat, such as Virginia 2000, are excluded entirely.}

The final sample contains 4{,}181 races in 267 state-year elections. Of these races, 534 ($12.8\%$) are proxied.\footnote{The full sample contains 267 state-years (2000–2024); analyses requiring map-drawing classifications use 249 state-years (2002–2024), which also form the reference sample for percentile rankings.} 500 races require proxy because the top two finishers are not one Democrat and one Republican, and 34  races require proxy because of a large non-major share, 26 races require proxy only because of invalid ballots. Of the proxied races, 453 use same-map congressional races and 81 use presidential returns. At least one race is proxied in 154 of the 267 elections.

The results are robust to using presidential elections as the primary proxy and the congressional elections in the same district as the secondary proxy. 

\subsection{Proportional Switch Weights in 2024}
\label{app:proportional-switch-weights}

\begin{figure}[h!]
  \centering
  \includegraphics[width=\linewidth]{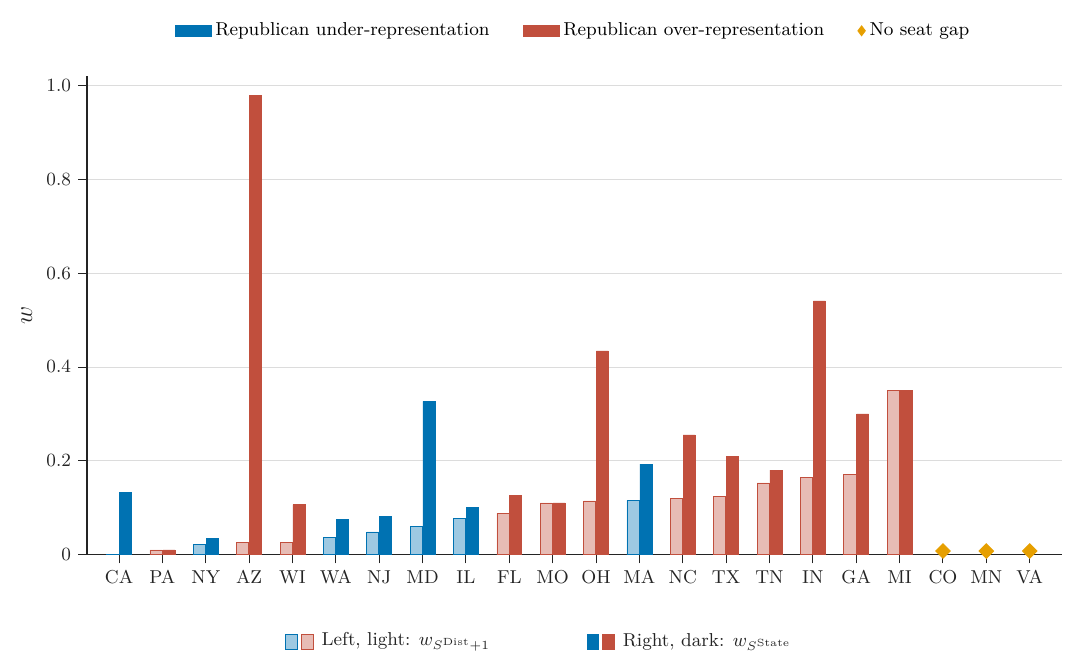}
  \caption{\small First-switch weight $w_{S+1}$ and the weight required to attain
  proportional representation, $w_{PR}$, by state in 2024. Lower values
  indicate that less weight on proportionality is required. Gray markers identify elections requiring no corrective switch.}  
  \label{fig:weights-2024}  
\end{figure}

Figure~\ref{fig:weights-2024} compares the first switch weight with the proportional switch weight, \(w_{S^\Agg}\), the weight at which the proportional seat total \(S^\Agg\) becomes optimal. The first switch weight records when the observed first-past-the-post allocation first ceases to be optimal; the proportional switch weight records how much farther the objective must move before proportionality itself is selected.

The gap between the two weights describes the adjustment path.  When reaching the proportional seat total requires a single seat change, as in Michigan, Pennsylvania, and Missouri, the first switch is also the last and the two weights coincide. Among states with larger seat gaps, the first move can still nearly reach proportionality, as in New York, where three switches are needed but $w_{PR}$ (0.035) is barely above $w_{S^\Dist+1}$ (0.022). Arizona illustrates the opposite pattern. The first switch weight is low, so the first-past-the-post allocation is fragile, but the proportional switch weight is nearly one, so reaching proportionality requires almost all weight on statewide representation.

\subsection{Implied weights over time}
\label{subsec:emp-time}

Figure~\ref{fig:weights-time} plots the full sequence of complementary switch
weights over 2000--2024 for four states that have been central to redistricting
litigation. Starting from the first-past-the-post seat total $S^\Dist$, each
circle marks another switch along the optimal path toward the proportional seat
total $S^\Agg$. The black square marks the final switch. A green diamond marks a
state-year where first-past-the-post and proportional representation agree, for which no switch is required.

\begin{figure}[h!]
  \centering
  \includegraphics[width=\linewidth]{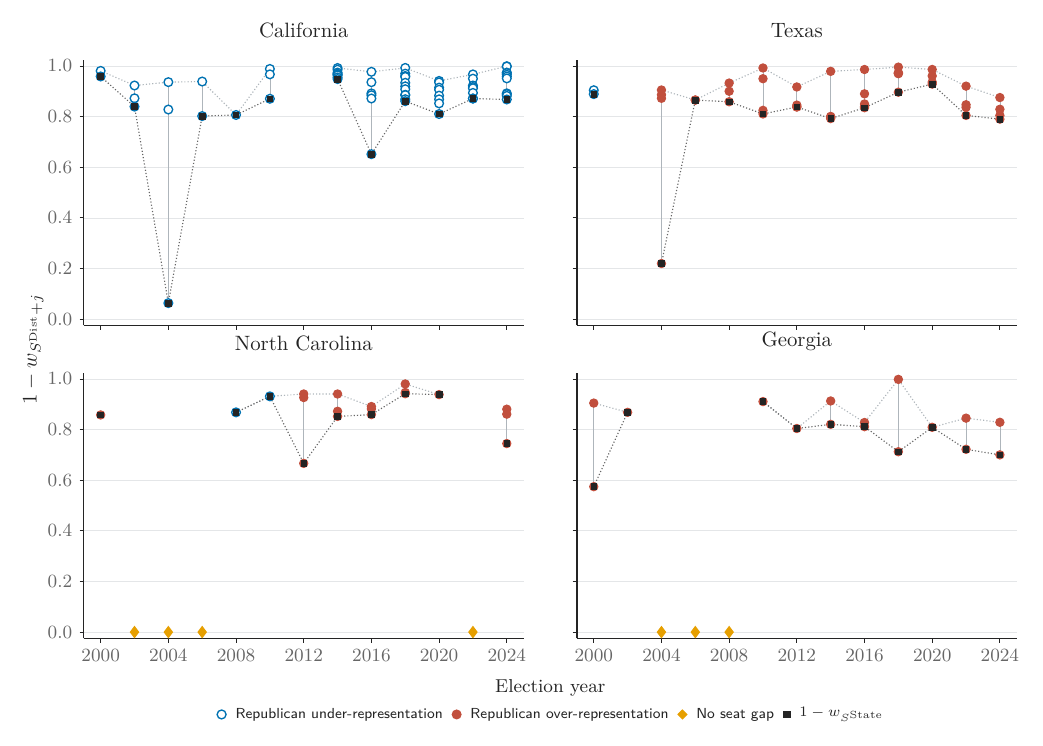}
  \caption{\small Complementary switch weights over time for California, Texas,
  Georgia, and North Carolina. Each circle is a switch along the optimal seat
  path from $S^\Dist$ to $S^\Agg$; the black square marks the final switch.
  Hollow blue and filled red circles denote Republican under- and
  over-representation, respectively. A green diamond marks elections where first-past-the-post and proportional representation agree and
  hence no switch. Missing elections break the connecting paths.}
  \label{fig:weights-time}
\end{figure}

North Carolina is the clearest illustration because its map-drawing authority changes while its seat count holds at thirteen through 2020 and rises only to fourteen after the 2020 census. Under Democratic-drawn lines its delegation is proportional through 2006, then the 2008 and 2010 elections leave Republicans one seat short of proportional. In 2012 with the seat count unchanged at thirteen, a newly enacted Republican map flips the delegation to a three-seat Republican over-representation, holding the first switch weight low ($w_{S^\Dist+1}=0.06$). The advantage persists under the 2011 plan (used in 2012 and 2014), which a federal court struck in 2016 as a racial gerrymander, and under its 2016 replacement (used in 2016 and 2018), which a state court struck in 2019 as a partisan gerrymander. The legislature’s 2019 remedial map narrows the gap to one seat in 2020, the court-commissioned plan returns the 2022 delegation to exact proportionality, and the legislature’s 2024 map moves it back to a three-seat Republican advantage.

Georgia shows a smaller but persistent Republican advantage that is generally robust. Its delegation is proportional in 2004 under the Democratic-drawn map and in 2006 and 2008 under the Republican mid-decade map of 2005, and runs one to two seats above proportional under every Republican map from 2010 on. Since 2012 the first switch weight is moderate to high ($w_{S^\Dist+1}$ between $0.09$ and $0.20$) in all cycles, except 2018, where a nearly tied seat induces a first switch weight of $0.0015$. The advantage generally rests on comfortably won rather than narrowly won seats. The seat gap widens from one seat to two after 2012, apart from 2020, without the delegation becoming more fragile. 

Texas’ 2004 mid-decade map opens a four-seat Republican advantage that persists, between two and five seats, through the court-drawn maps that governed from 2006 to 2020 — the 2006 remedial plan, which redrew five districts of the 2003 legislative map, and the court’s 2012 interim plan, which the legislature adopted as the permanent plan in 2013 — and into the legislature’s 2021 map, used in 2022 and 2024. It is largely unmoved by two rounds of reapportionment (thirty-two to thirty-six to thirty-eight seats). What changes is robustness. Under the returning Republican legislature the first switch weight rises from about $0.08$ in 2022 to $0.12$ in 2024, and $w_{PR}$ from $0.20$ to $0.21$, so the same four-seat advantage rests less on nearly-tied districts than it did under the court maps.

California is the outlier in size. With its 52-53 seats, Democratic advantage grows from one to three seats under the Democratic legislature of the 2000s to eight and then eleven seats under the independent commission of the 2010s and 2020s, and $1-w_{S^{\Dist} + 1}$ sits near one in almost every commission-era cycle, so the advantage is large and fragile.

\subsection{How measures differ}
\label{app:how-measures-differ}

The efficiency gap and switch weight both are governed to some extent by the seat assignment. However, where the efficiency gap concerns itself with the full vote distribution, the switch weight to move away from first-past-the-post depends heavily on the most competitive races. Hence, they need not agree every time. Table \ref{tab:litigated} shows the percentile value of Efficiency Gap, the switch weight and proportionality gap for the election map struck down/ in litigation in courts. To illustrate the similarities and differences, we focus on some extreme examples below.

\noindent
\textbf{When the measures agree.}
Where the statewide vote is near half, the winner's-bonus is small and efficiency gap tracks the proportionality gap. And where the pivotal district is nearly tied, the switch weight tracks how fragile that seat gap is. North Carolina can be used to show both extremes (Figures~\ref{fig:eg-seatgap} and~\ref{fig:eg-w1}). Its 2018 map records an efficiency gap of about $27\%$, a three-seat Republican advantage and Republicans hold the ninth district at $51\%$, so $w_{S^\Dist+1}=0.02$\footnote{The 2018 result in North Carolina’s 9th district was never certified because of absentee-ballot fraud, and the seat remained vacant until a special election on September~10, 2019, held under the same district boundaries. For this district, we use the 2019 special election result in place of the uncertified 2018 count. The other twelve districts use 2018 results.}. The statewide vote is almost even ($\bar p=0.501$), the winner's-bonus is negligible, and both measures sit at the extreme for a map that was struck as a partisan gerrymander. Its 2022 court-commissioned map is the other story. With a proportional seat assignment, an efficiency gap of $-4.4\%$ (inside the band), and no seat to flip, so the first-switch
weight is one, both measures would read this map as \textit{``fair"}.

\medskip
\noindent
\textbf{Disagreement I: Proportional assignment with lopsided vote profile.}\\
There can be vote distributions such that the seat gap is zero, the switch weight is one and the efficiency gap almost equals the winner's bonus
(Proposition~\ref{prop:eg-linear-separations}(i)). Illinois 2008 is the cleanest such profile. Republicans win seven of nineteen seats against a proportional target of $7.2$, a proportional delegation, but they hold only around $38\%$ of the statewide vote. So the winner's bonus is $12$ pp and the efficiency gap is $+10.9$ pp, past the threshold and Republican-favorable. Illinois 2008 has a proportionality gap of $-1.1$ pp, while the efficiency gap is above the $70$th percentile. The efficiency gap reads a proportional delegation as strongly skewed because its target withholds from the statewide minority the winner's bonus it would grant a majority (Proposition~\ref{prop:eg-statewide}).

\medskip
\noindent
\textbf{Disagreement II: Fragile seat assignment with lopsided vote profile.}\\
When a map delivers a large but knife-edge advantage to the statewide majority, the seat gap and the switch weight both flag the map while the efficiency gap clears it. Texas 2010 gives Republicans a four-seat advantage resting on a pivotal district won by a small margin, $50.4\%$ ($w_{S^\Dist+1}=0.008$), yet its efficiency gap is only $+3.3\%$. This is because a large Republican statewide share $(59.3\%)$ earns a correspondingly large winner's bonus. The proportionality gap of $+12.6$ points is offset by a winner’s-bonus adjustment of $-9.3$ points.
New York 2010 is the mirror image on the Democratic side. A three-seat Democratic edge held by a district Democrats win $50.2$ to $49.8$ ($w_{S^\Dist+1}=0.003$), and an efficiency gap of $+3.5\%$, slightly republican favoring. The proportionality gap of $-9.4$ points is more than offset by an adjustment of $+13.0$ points due to a statewide vote share of $63\%$.
In both, the efficiency gap adds together the seat advantage and the majority's bonus resulting in a relatively lower value, while the seat gap records the size of the advantage and the switch weight records how thin the seats delivering it are. In Figure~\ref{fig:eg-w1} the two maps lie well within the $7\%$ band but are above the $90$th percentile in inverse switch weight. 

\medskip
\noindent
\textbf{Disagreement III: Robust assignment with lopsided vote profile.}
The reverse case is a large advantage built on safe seats, where the efficiency gap flags the map but the switch weight is relatively robust.
(Proposition~\ref{prop:eg-linear-separations}(iii)). Georgia 2024 has an efficiency gap of $+11.1\%$ and a two-seat Republican advantage. With a nearly even statewide vote ($51.6\%$ Republican), the winner’s-bonus adjustment is small ($-1.6$ points), so the efficiency gap largely tracks the proportionality gap of $+12.7$ points. However, its two narrowest Republican seats are won at $61.98\%$ and $60.32\%$, comfortably clear of a tie, so $w_{S^\Dist+1}=0.171$. The advantage is robust to a larger range of proportionality weights.\footnote{Georgia’s 2021 congressional map was struck on October~26, 2023 by the U.S.\ District Court for the Northern District of Georgia under Section~2 of the Voting Rights Act for diluting Black voting strength, in \emph{Pendergrass v.\ Raffensperger}. The General Assembly enacted the remedial map used here in December 2023, which the same court approved. As of our May 2026 data cutoff, appeals of both the liability and remedial rulings were pending. The challenge is thus a racial-vote-dilution claim under the Voting Rights Act, distinct from the partisan-gerrymandering standard the efficiency gap was designed for. Case materials: \url{https://redistricting.lls.edu/case/pendergrass-v-raffensperger/}.}

\textbf{Comparison of measures and litigation}

Figures~\ref{fig:eg-seatgap} and~\ref{fig:eg-w1} plot the efficiency gap against the proportional seat gap and the signed inverse first-switch weight, respectively, for all state-years in the sample. Stars identify struck or litigated maps, and dashed lines mark the conventional $\pm 7\%$ efficiency-gap thresholds. The figures illustrate how individual maps can occupy different positions across the measures. Table~\ref{tab:litigated} reports the percentile ranks of each struck or litigated map under all three measures, together with counts exceeding the 90th and 95th percentiles.

\begin{figure}[h]
  \centering
    \includegraphics[width=\textwidth]{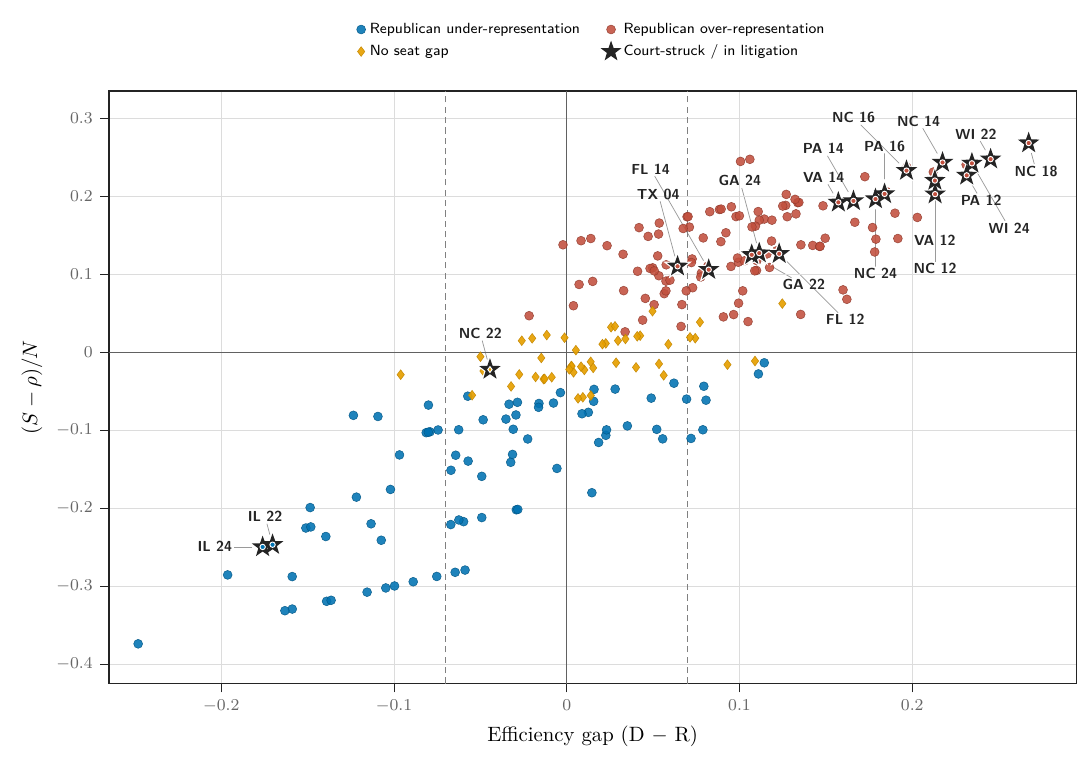}
\caption{\small  Proportional seat gap (\(S^\Dist-\rho\)) / $N$ against the efficiency gap. Positive values indicate Republican over-representation, and negative values indicate Republican under-representation. Stars mark struck or litigated maps; dashed lines mark the conventional \(\pm 7\%\) efficiency-gap thresholds.}  \label{fig:eg-seatgap}
\end{figure}

\begin{figure}[h]
  \centering
     \includegraphics[width=\textwidth]{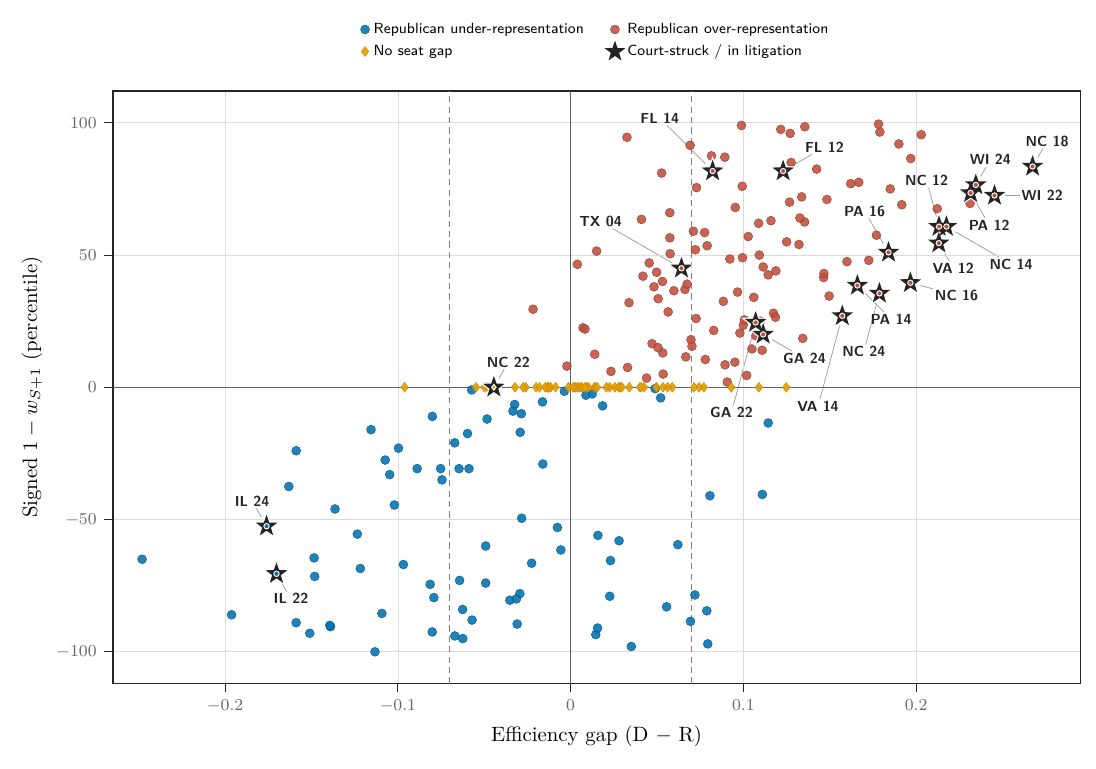}
\caption{\small Signed inverse first-switch weight \(1-w_{S^\Dist+1}\) against the efficiency gap. Positive values indicate Republican over-representation, negative values under-representation. Yellow points denote the elections where first-past-the-post is optimal for statewide representation. Stars mark struck or litigated maps; dashed lines mark the \(\pm 7\%\) reference lines.}
\label{fig:eg-w1}
\end{figure}

\begin{table}[h!]
\centering
\caption{\small Percentile ranks (0--100, higher = more extreme) for litigated and struck congressional maps. Columns report the efficiency gap, inverse first switch weight, and proportional seat gap.}
\label{tab:litigated}
\begin{tabular}{lccc}
\hline\hline
Map
& Efficiency Gap
& \shortstack{Inverse First \\ Switch Weight}
& \shortstack{Proportional\\Seat Gap (per seat)} \\
\hline
GA2024 & 74 & 35 & 58 \\
IL2022 & \textbf{91} & 77 & \textbf{93} \\
IL2024 & \textbf{92} & 63 & \textbf{94} \\
NC2024 & \textbf{93} & 48 & 83 \\
WI2022 & \textbf{99} & 78 & \textbf{94} \\
WI2024 & \textbf{99} & 81 & \textbf{92} \\
GA2022 & 69 & 39 & 57 \\
NC2012 & \textbf{97} & 69 & 85 \\
NC2014 & \textbf{98} & 69 & \textbf{92} \\
TX2004 & 44 & 57 & 51 \\
VA2012 & \textbf{97} & 64 & 87 \\
VA2014 & 88 & 41 & 82 \\
FL2012 & 78 & 86 & 58 \\
FL2014 & 58 & 86 & 49 \\
NC2016 & \textbf{96} & 51 & 90 \\
NC2018 & \textbf{100} & 87 & \textbf{95} \\
NC2022 & 29 & 10 & 12 \\
PA2012 & \textbf{99} & 79 & 89 \\
PA2014 & 90 & 51 & 82 \\
PA2016 & \textbf{94} & 61 & 85 \\
\hline
\multicolumn{4}{l}{\emph{Upper-tail counts among litigated or struck maps}} \\
Above 90th percentile & 12 & 0 & 6 \\
Above 95th percentile & 8 & 0 & 0 \\
\hline\hline
\end{tabular}

\end{table}

\clearpage
\singlespacing
\setcounter{page}{1}
\begin{center}
\Large \textbf{Online Appendix}\\
{Representation in District Based Elections} \\
\end{center}

\section{Alternative Utility Microfoundation}
\label{app:alternative-utility-representation}

This appendix records an equivalent utility-based microfoundation for the linear benchmark. The main text uses a loss formulation: voters suffer disutility from local nonrepresentation and from statewide underrepresentation. The same aggregate objective can also be obtained from a utility formulation in which voters receive positive utility from representation on both dimensions, with an additional penalty when their party is underrepresented.

Assume again that $\rho\in(0,N)$. For a voter who supports party $A$ in district $d$, define utility under allocation $\mathbf x$ by
\begin{equation*}
  u_d^A(\mathbf x)
  =
  (1-w)x_d
  +
  w\left[
    \frac{S(\mathbf x)}{\rho}
    -
    \frac{(\rho-S(\mathbf x))_+}{\rho}
  \right].
\end{equation*}
For a voter who supports party $B$ in district $d$, define
\begin{equation*}
  u_d^B(\mathbf x)
  =
  (1-w)(1-x_d)
  +
  w\left[
    \frac{N-S(\mathbf x)}{N-\rho}
    -
    \frac{(S(\mathbf x)-\rho)_+}{N-\rho}
  \right].
\end{equation*}
The first term is local representation utility. The second term is statewide representation utility: voters value seats held by their party in the statewide delegation, measured relative to the party's statewide support, and incur an additional penalty when their party is underrepresented.

Let integrated voter utility be $  U(\mathbf x)
  :=
  \sum_{d=1}^N
  \left[
    p_d u_d^A(\mathbf x)+(1-p_d)u_d^B(\mathbf x)
  \right].$ Then
\begin{equation*}
  U(\mathbf x)
  =
  N-
  \left[
    (1-w)\Dist(\mathbf x,\mathbf p)+w|S(\mathbf x)-\rho|
  \right].
\end{equation*}
Hence maximizing $U(\mathbf x)$ is equivalent to minimizing the same linear misrepresentation objective as in the main text.

To see this, the local utility component is
\begin{equation*}
\begin{aligned}
  (1-w)
  \sum_{d=1}^N
  \bigl[
    p_d x_d+(1-p_d)(1-x_d)
  \bigr]
  &=
  (1-w)
  \bigl[N-\Dist(\mathbf x,\mathbf p)\bigr].
\end{aligned}
\end{equation*}
The positive statewide representation terms aggregate to an allocation-independent constant:
\begin{equation*}
  w\left[
    \rho\frac{S(\mathbf x)}{\rho}
    +
    (N-\rho)\frac{N-S(\mathbf x)}{N-\rho}
  \right]
  =
  wN.
\end{equation*}
The statewide underrepresentation penalties aggregate to $  w\left[
    (\rho-S(\mathbf x))_+ +(S(\mathbf x)-\rho)_+
  \right]
  =
  w|S(\mathbf x)-\rho|.$ Therefore
\begin{equation*}
  U(\mathbf x)
  =
  (1-w)\bigl[N-\Dist(\mathbf x,\mathbf p)\bigr]
  +
  wN
  -
  w|S(\mathbf x)-\rho|
  =
  N-
  \left[
    (1-w)\Dist(\mathbf x,\mathbf p)+w|S(\mathbf x)-\rho|
  \right].
\end{equation*}

\section{Gerrymandering Costs}
\label{app:gerrymandering-costs}

This appendix records the cost version of Proposition~\ref{prop:gerrymandering-set-inclusion}. Maintain the setup of that proposition: a baseline profile \(\mathbf p\), baseline rule index \(\lambda_0\), and target seat total \(k\) are fixed, with party \(A\) overrepresented at the baseline and \(k>|R_{\lambda_0}(\mathbf p)|\).

Given statewide vote share \(m\), a \emph{gerrymandering cost} is a continuous map \(c:\mathcal P(m)\times\mathcal P(m)\to\mathbb R_+\) such that \(c(\mathbf p,\mathbf p)=0\) and, for every distinct \(\mathbf p,\mathbf p'\in\mathcal P(m)\) and all \(0\le\theta<\theta'\le 1\),
\begin{equation*}
  c\bigl(\mathbf p,(1-\theta)\mathbf p+\theta\mathbf p'\bigr)
  <
  c\bigl(\mathbf p,(1-\theta')\mathbf p+\theta'\mathbf p'\bigr).
\end{equation*}
Thus moving farther away from the baseline profile along any mean-preserving segment is strictly more costly.

For an integer \(k\), define the minimal cost of inducing at least \(k\) seats under the rule \(R_\lambda\) by
\begin{equation*}
  C(\lambda;\mathbf p,k)
  :=
  \inf
  \Big\{
    c(\mathbf p,\mathbf r):
    \mathbf r\in\mathcal G_k(\mathbf p;\lambda)
  \Big\},
\end{equation*}
with the convention \(\inf\emptyset=\infty\). Equivalently, \(C(\lambda;\mathbf p,k)\) is the least cost of moving from \(\mathbf p\) to a mean-preserving profile that gives party \(A\) at least \(k\) seats under \(R_\lambda\).

\begin{corollary}
\label{cor:gerrymandering-cost-from-set-inclusion}
For every \(\lambda',\lambda\) with \(1\ge \lambda'>\lambda\ge \lambda_0\), $ C(\lambda';\mathbf p,k)\ge C(\lambda;\mathbf p,k).$
\end{corollary}

\begin{proof}[\textbf{Proof of Corollary~\ref{cor:gerrymandering-cost-from-set-inclusion}.}]
By Proposition~\ref{prop:gerrymandering-set-inclusion}, $  \mathcal G_k(\mathbf p;\lambda')
  \subseteq
  \mathcal G_k(\mathbf p;\lambda)$ whenever \(\lambda'>\lambda\ge\lambda_0\). Since \(C(\lambda;\mathbf p,k)\) is the infimum of \(c(\mathbf p,\mathbf r)\) over \(\mathcal G_k(\mathbf p;\lambda)\), the infimum defining \(C(\lambda';\mathbf p,k)\) is taken over a weakly smaller feasible set. Therefore \(C(\lambda';\mathbf p,k)\ge C(\lambda;\mathbf p,k)\).
\end{proof}

The corollary is Proposition~\ref{prop:gerrymandering-set-inclusion} written in cost terms. Since \(\mathcal G_k(\mathbf p;\lambda')\subseteq \mathcal G_k(\mathbf p;\lambda)\), the minimization defining \(C(\lambda';\mathbf p,k)\) is taken over a smaller feasible set of maps. A smaller feasible set cannot lower the minimum cost. If the increase in \(\lambda\) removes the lowest-cost successful redistrictings, then the minimum cost rises strictly.

Taking \(\lambda_0=0\) gives the first-past-the-post interpretation. If party \(A\) is overrepresented under first-past-the-post and \(k>|R_0(\mathbf p)|\), then moving away from first-past-the-post by increasing \(\lambda\) weakly raises the minimal cost of obtaining those additional seats. This gives the rule-design interpretation of the model. Redistricting constraints act on maps; the rule index acts on the allocation rule. Both can make gerrymandering harder, but they operate through different levers. Even holding the districting technology fixed, moving from first-past-the-post toward proportionality makes additional seats for an already overrepresented party harder to engineer, because the set of maps on which those extra seats are attainable becomes smaller.

\section{Additional Examples}

\subsection{First-Past-the-Post and Majorization}
\label{app:fptp-majorization-example}

The following example illustrates why majorization monotonicity need not hold for rules with a fixed district winning threshold. Fix any evaluation weight \(w>0\). Choose \(x>0\) small enough that
\[
  x<\frac{1}{20}
  \qquad\text{and}\qquad
  2x(1-w)<w .
\]
Consider the four-district profile
\[
  \mathbf q^x=(0.80,\,1/2+x,\,1/2-x,\,0.10),
  \qquad
  \rho=\sum_{d=1}^4 q_d^x=1.90.
\]
Under first-past-the-post, party \(A\) wins two districts. The seat total is therefore close to proportionality: \(|2-1.90|=0.10\).

\begin{center}
\begin{tikzpicture}[x=1.25cm,y=3.1cm]
  \pgfmathsetmacro{\xx}{0.03}
  \draw[->] (0,0) -- (5.0,0);
  \draw[->] (0,0) -- (0,1.0);
  \node[below] at (1,0) {\small $1$};
  \node[below] at (2,0) {\small $2$};
  \node[below] at (3,0) {\small $3$};
  \node[below] at (4,0) {\small $4$};
  \node[left]  at (0,0.5) {\small $1/2$};

  \draw[thick, draw=red!75!black] (0,0.5) -- (4.6,0.5);
  \node[anchor=west, font=\small, text=red!75!black] at (4.62,0.5) {$t=\tfrac12$};

  \filldraw[fill=blue!25, draw=blue!55!black] (0.70,0) rectangle (1.30,0.80);
  \filldraw[fill=blue!25, draw=blue!55!black] (1.70,0) rectangle (2.30,{0.5+\xx});
  \filldraw[fill=orange!22, draw=orange!65!black] (2.70,0) rectangle (3.30,{0.5-\xx});
  \filldraw[fill=black!10, draw=black!50] (3.70,0) rectangle (4.30,0.10);

  \node[font=\small] at (1.00,0.86) {$0.80$};
  \node[font=\small, anchor=south] at (2.00,{0.5+\xx+0.025}) {$\tfrac12+x$};
  \node[font=\small, text=orange!65!black, anchor=south] at (3.00,{0.5-\xx+0.025}) {\emph{\(\tfrac12-x\)}};
  \node[font=\small] at (4.00,0.16) {$0.10$};
  \node[font=\small, anchor=west, text=orange!65!black] at (2.7,0.75) {\emph{loss}};
\end{tikzpicture}
\end{center}

Transfer \(2x\) units of party \(A\)'s support from district \(4\) to district \(3\):
\[
  \mathbf p^x=(0.80,\,1/2+x,\,1/2+x,\,0.10-2x).
\]
The profile \(\mathbf p^x\) majorizes \(\mathbf q^x\) and has the same statewide support. For any fixed seat total, the top-\(S\) allocation under \(\mathbf p^x\) has weakly lower district-level misrepresentation than the top-\(S\) allocation under \(\mathbf q^x\). Under first-past-the-post, however, the fixed threshold \(1/2\) changes the induced seat total: party \(A\) now wins three districts.

\begin{center}
\begin{tikzpicture}[x=1.25cm,y=3.1cm]
  \pgfmathsetmacro{\xx}{0.03}
  \draw[->] (0,0) -- (5.0,0);
  \draw[->] (0,0) -- (0,1.0);
  \node[below] at (1,0) {\small $1$};
  \node[below] at (2,0) {\small $2$};
  \node[below] at (3,0) {\small $3$};
  \node[below] at (4,0) {\small $4$};
  \node[left]  at (0,0.5) {\small $1/2$};

  \draw[thick, draw=red!75!black] (0,0.5) -- (4.6,0.5);
  \node[anchor=west, font=\small, text=red!75!black] at (4.62,0.5) {$t=\tfrac12$};

  \filldraw[fill=blue!25, draw=blue!55!black] (0.70,0) rectangle (1.30,0.80);
  \filldraw[fill=blue!25, draw=blue!55!black] (1.70,0) rectangle (2.30,{0.5+\xx});
  \filldraw[fill=green!18, draw=green!50!black] (2.70,0) rectangle (3.30,{0.5+\xx});
  \filldraw[fill=black!10, draw=black!50] (3.70,0) rectangle (4.30,{0.10-2*\xx});

  \node[font=\small] at (1.00,0.86) {$0.80$};
  \node[font=\small, anchor=south] at (2.00,{0.5+\xx+0.025}) {$\tfrac12+x$};
  \node[font=\small, text=green!50!black, anchor=south] at (3.00,{0.5+\xx+0.025}) {\emph{\(\tfrac12+x\)}};
  \node[font=\small, anchor=south] at (4.00,{0.10-2*\xx+0.02}) {$0.10-2x$};
  \node[font=\small, anchor=west, text=green!50!black] at (2.7,0.82) {\emph{win}};

  \draw[->, thick, draw=purple!70!black] (4.00,0.2) to[out=110,in=70] (3.00,0.85);
  \node[font=\small, text=purple!70!black] at (4,0.88) {$2x$};

\end{tikzpicture}
\end{center}

The district component under first-past-the-post falls by \(2x\). The statewide component rises from \(|2-1.90|=0.10\) to \(|3-1.90|=1.10\), an increase of one seat. Hence, under the linear objective,
\[
  M(\rfirst,\mathbf p^x;w)-M(\rfirst,\mathbf q^x;w)
  =
  -2x(1-w)+w.
\]
By the choice of \(x\), this expression is positive. Thus, for every \(w>0\), there is a mean-preserving majorization that improves local representation under first-past-the-post but increases total misrepresentation. The reason is that first-past-the-post keeps the district threshold fixed: concentrating support can turn a marginal district from a loss into a win, and the resulting discrete seat gain can dominate the district-level improvement.

\end{document}